\documentclass[11pt]{article}
\usepackage{amssymb}

\usepackage{amsmath}

\newcounter{resultnum}[section]
\newtheorem{conclusion}{Conclusion}[section]

\newcounter{conclusionnum}[section]

\newcounter{conditionnum}[section]

\newcounter{conjecturenum}[section]
\newtheorem{example}{Example}[section]

\newcounter{examplenum}[section]

\newcounter{exercisenum}[section]
\newtheorem{lemma}{Lemma}[section]

\newcounter{lemmanum}[section]

\newcounter{notationnum}[section]
\newtheorem{theorem}{Theorem}[section]

\newcounter{theoremnum}[section]
\newtheorem{definition}{Definition}[section]

\newcounter{definitionnum}[section]
\newtheorem{corollary}{Corollary}[section]

\newcounter{corollarynum}[section]

\newcounter{remarknum}[section]
\newtheorem{proposition}{Proposition}[section]

\newcounter{propositionnum}[section]

\newcounter{acknowledgementnum}[section]

\newcounter{algorithmnum}[section]

\newcounter{axiomnum}[section]

\newcounter{casenum}[section]

\newcounter{claimnum}[section]

\newcounter{summarynum}[section]

\newcounter{problemnum}[section]
\newenvironment{proof}[1][]{\textbf{Proof.} }{}

\begin{document}

\title{Curve Flows and Solitonic Hierarchies \\
Generated by Einstein Metrics }
\date{November 19, 2008}
\author{ Sergiu I. Vacaru\thanks{%
Sergiu.Vacaru@gmail.com } \\
{\quad} \\
\textsl{The Fields Institute for Research in Mathematical Science} \\
\textsl{222 College Street, 2d Floor, } \textsl{Toronto \ M5T 3J1, Canada} \\
and \\
\textsl{Faculty of Mathematics, University "Al. I. Cuza" Ia\c si}, \\
\textsl{\ 700506, Ia\c si, Romania} }
\maketitle

\begin{abstract}
We investigate bi--Hamiltonian structures and mKdV hierarchies of solitonic equations generated by (semi) Riemannian metrics and curve flows of non--stretching curves. There are applied methods of the geometry of nonholonomic manifolds enabled with metric--induced nonlinear connection (N--connection) structure. On spacetime manifolds, we consider a nonholonomic splitting of dimensions and define a new class of liner connections which are 'N--adapted', metric compatible and uniquely defined
by the metric structure. We prove that for such a linear
connection, one yields couples of generalized sine--Gordon
equations when the corresponding geometric curve flows result
in solitonic hierarchies described in explicit form by
nonholonomic wave map equations and mKdV analogs of the
Schr\"{o}dinger map equation. All geometric constructions can be re--defined for the Levi--Civita connection but with ''noholonomic mixing'' of solitonic interactions. Finally, we speculate why certain methods and results from the geometry of
nonholonmic manifolds and solitonic equations have general importance in various directions of modern mathematics,
geometric mechanics, fundamental theories in physics and applications, and briefly analyze possible nonlinear wave configurations for modeling gravitational interactions by effective continuous media effects.

\vskip3pt \textbf{Keywords:}\ Curve flow, (semi) Riemannian spaces,
nonholonomic manifold, nonlinear connection, bi--Hamiltonian, solitonic
equations.

\vskip3pt MSC:\ 37K05, 37K10, 37K25, 35Q53, 53B20, 53B40, 53C21, 53C60
\end{abstract}



\section{ Introduction}

The 'anholonomic frame method' \cite{vncg,ijgmmp1,rf2,rf3,ncrf} was recently
developed as a general geometric approach for constructing exact solutions
in gravity and Ricci flow theory following the formalism of nonlinear
connections and nonholonomic distributions in Riemann--Finsler geometry and
applications in modern physics and mechanics \cite{ma2,vrfg,vsgg}. In
parallel, the differential geometry of plane and space curves received
considerable attention in the theory of nonlinear partial differential
equations and applications to modern physics \cite{chou1,mbsw}. One proved
that curve flows on Riemannian spaces of constant curvature are described
geometrically by hierarchies defined by wave map equations and mKdV analogs
of Schr\"{o}dinger map equation. The main results on vector generalizations
of KdV and mKdV equations and the geometry of their Hamiltonian structures
are summarized in Refs. \cite{ath,saw,serg}, see also a recent work in \cite%
{fours,wang}.

In \cite{anc1,anc2}, the flows of non--stretching curves were analyzed using
moving parallel frames and associated frame connection 1--forms in a
symmetric spaces $M=G/SO(n)$ and the structure equations for torsion and
curvature encoding $O(n-1)$--invariant bi--Hamiltonian operators.\footnote{$%
G $ is a compact semisimple Lie group with an involutive automorphism that
leaves fixed a Lie subgroup $SO(n)\subset G,$ for $n\geq 2$} It was shown
that the bi--Hamiltonian operators produce hierarchies of integrable flows
of curves in which the frame components of the principal normal along the
curve satisfy $O(n-1)$--soliton equations. The crucial condition for
performing such constructions is that the frame curvature matrix is constant
on the curved manifolds like $M=G/SO(n).$

On a general (pseudo) Riemannian manifold, working only with the
Levi--Civita connection, it is not possible to define in explicit form such
systems of reference and coordinates when the curvature would be described
completely by constant ''matrix'' coefficients and satisfy the conditions
for solitonic encoding. One concluded that only for the (pseudo) Riemannian
manifolds of constant curvature the geometric data on curve flows can be
transformed into solitonic hierarchies.

Nevertheless, we argue that the geometry of (pseudo) Riemannian manifolds%
\footnote{%
and a number of generalized Lagrange--Finsler spaces and their Ricci flows
into (pseudo) Riemannian, Eisenhart--Moffat, nonholonomic Fedosov,
noncommutative and other structures, see recent results and references in %
\cite{rf2,rf3,rfns,plafq,ncrf}} can be encoded into corresponding
bi--Hamilton operators and solitonic hierarchies. This is possible if we
work with an auxiliary class of metric compatible linear connections which
are completely defined by a fixed metric structure (in a particular case, we
can chose an exact solution of the Einstein equations in general relativity).

Let us explain the main idea for constructions with alternative linear
connections defined by a fixed (pseudo) Riemannian/Einstein metric:\ For a
metric $\mathbf{g,}$ there is a unique Levi--Civita connection $\ ^{\mathbf{g%
}}\nabla $ satisfying the conditions of metricity, $\ ^{\mathbf{g}}\nabla
\mathbf{g}=0,$ and vanishing torsion, $\ _{\mathbf{\nabla }}\mathbf{T}=0.$%
\footnote{%
we shall use left ''up'' and ''low'' labels in order to emphasize that
certain geometric objects are defined by a fixed metric, connection or other
fields} From the same metric, we can construct an infinite number of metric
compatible linear connections $\{\ ^{\mathbf{g}}\mathbf{D}\}$ satisfying the
conditions
\begin{equation}
\ ^{\mathbf{g}}\mathbf{Dg=}0,\ \mbox{ when }\ ^{\mathbf{g}}\mathbf{D=\ ^{%
\mathbf{g}}\nabla +\ }^{\mathbf{g}}\mathbf{Z.}  \label{mcomp}
\end{equation}%
The distorsion tensor $\mathbf{\ }^{\mathbf{g}}\mathbf{Z}$ (for metric
compatible linear connections, this tensor is an algebraic combination of
the coefficients of torsion $_{\mathbf{D}}^{\ \mathbf{g}}\mathbf{T)}$ is
computed in explicit form: it is defined only by the coefficients of $%
\mathbf{g}$ if a well defined geometric principle is introduced into
consideration (such a principle has to be different from the condition of
zero torsion). For instance, for deformation quantization of the Einstein
gravity \cite{plafq}, it was important to construct a canonical almost K\"{a}%
hler connection $\ _{K}^{\mathbf{g}}\mathbf{D}$ in a form to be compatible
with the so--called canonical almost symplectic structure $\ ^{\mathbf{g}%
}\theta $ (defined by the coefficients of $\mathbf{g),}$ when $\ _{K}^{%
\mathbf{g}}\mathbf{D\ g=\ }_{K}^{\mathbf{g}}\mathbf{D}$ $\left( ^{\mathbf{g}%
}\theta \right) =0.$ We worked with a nontrivial torsion structure $\ _{%
\mathbf{K}}^{\ \mathbf{g}}\mathbf{T}\neq 0,$ which was very important for a
generalization of Fedosov quantization. Nevertheless, we emphasize that all
geometric constructions could be redefined equivalently for the Levi--Civita
connection $\mathbf{\ ^{\mathbf{g}}\nabla }$ because $\ _{\mathbf{K}}^{\
\mathbf{g}}\mathbf{T}$ is also constructed only from the metric
coefficients. Such a torsion field is completely different from that (for
instance) in Riemann--Cartan, or string, gravity, where torsion is considered
as a new physical field subjected to additional field equations, see
discussion in \cite{vrfg}.

In this article, we shall work with two metric compatible linear connections
\ $\mathbf{\ ^{\mathbf{g}}\nabla }$ and $\ ^{\mathbf{g}}\mathbf{D,}$ both
constructed from the coefficients of a metric $\mathbf{g,}$ when the
curvature tensor for the second connection can be represented by a constant
coefficients matrix (with respect to a well--defined frame structure).
Following the geometry of curve flows defined by the connection $\ ^{\mathbf{%
g}}\mathbf{D,}$ we shall derive the corresponding bi--Hamiltonian structure
and related solitonic hierarchies. This way, having encoded the geometrical
data for $\mathbf{g}\ $\ and $\ ^{\mathbf{g}}\mathbf{D}$ into solitonic
equations (and their solutions), we shall be able to re--define them for $%
\mathbf{g}\ $\ and $\ ^{\mathbf{g}}\mathbf{\nabla ,}$ computing $\mathbf{\ }%
^{\mathbf{g}}\mathbf{Z}$ and using the distorsion relation $\mathbf{\ ^{%
\mathbf{g}}\nabla =\ ^{\mathbf{g}}\mathbf{D-}\ }^{\mathbf{g}}\mathbf{Z}$ (%
\ref{mcomp}). Such values were formally introduced for certain
classes of connections in Finsler geometry but they can be similarly
constructed on (pseudo) Riemannian spaces and their nonholonomic
deformations. We shall follow the formalism and conventions established in
Ref. \cite{ijgmmp1}, see also reviews of results in \cite{vrfg,vsgg,vncg},
on the geometry of nonholonomic manifolds, Finsler--Lagrange methods and
applications to modern physics. The approach originates from the geometry of
nonlinear connections (N--connections) and Finsler geometry and
generalizations formally developed on tangent bundles and manifolds enabled
with generalized connections and applications in mechanics, see summaries of
results and references from \cite{ma2,bej}.

The aim of this paper is to prove that respective curve flow solitonic
hierarchies are generated by any (semi) Riemannian metric $\mathbf{g}%
_{\alpha \beta }$ on a manifold $\mathbf{V}$ of dimension $n+m,$ for $n\geq 2
$ and $m\geq 1,$ if such a space is enabled with a nonholonomic distribution
defining a spacetime splitting.\footnote{%
in physical literature, one uses the term ''pseudo Riemannian'' instead of
''semi Riemannian''} For such distributions with associated nonlinear
connection (N--connection) structure, we can define certain classes of frame
and linear connection nonholonomic deformations when the curvature is
characterized by constant matrix coefficients. This allows us to derive the
corresponding hierarchies of gravitational solitonic equations and
conservation laws. We also prove that any solution of the Einstein equations
(a vacuum one, or with nontrivial cosmological constant) can be encoded into
such solitonic hierarchies.

We would like to emphasize the multi--disciplinary character of our approach
based on modelling geometries and physical interactions with generic local
anisotropy. The first constructions and applications to locally anisotro\-pic
thermodynamics, generalized gravity theories and physics of continuous media
were proposed in Finsler geometry and its generalizations to Lagrange and
Hamilton spaces (see monographs and reviews \cite{ma2,bej,vsgg,vrfg} and
references therein; here we also cite \cite{bcs} based on the Chern
connection, which is not metric compatible and considered to be less
relevant for standard physical theories, but nevertheless very important for
a number of other type applications). Such geometrical and physical models
can be unified by the concept of nonholonomic manifold (usual manifolds
endowed with additional nonintegrable/nonholonomic distributions) \cite%
{vr1,vr2,hor}, see reviews of mathematical results in \cite{bejf} and of
various physical applications in \cite{vsgg,vrfg}. Our present work
establishes a bridge between the geometry of nonlinear connections
elaborated in Finsler geometry and nonholonomic physics and the theory of
solitonic equations defined on curved spaces with various applications in
gravity physics, geometric mechanics, locally anisotropic kinetic processes
and thermodynamics and nonholonomic Ricci flows \cite%
{vlank,vlsp,ancvac,vrepmp}.

The paper is organized as follows:

In section 2 we outline the geometry of N--adapted frame transforms on
(pseudo) Riemannian spaces enabled with N--connection structure. We
emphasize the possibility to work equivalently with different classes of
linear connections (the Levi--Civita and various N--adapted ones) completely
defined by a metric structure $\mathbf{g}_{\alpha \beta }$ for a prescribed
splitting $n+m.$ We show how alternative linear connections with constant
Riemannian tensor matrix coefficients can be derived from a (pseudo)
Riemannian metric. A class of nonholonomic Einstein spaces is analyzed.

In section 3 we consider curve flows on nonholonomic (pseudo) Riemannian
spaces. It is constructed a class of nonholonomic Klein spaces for which the
bi--Hamiltonian operators are derived for a linear connection adapted to the
nonlinear connection structure, for which the distinguished curvature
coefficients are constant.

Section 4 is devoted to the formalism of distinguished bi--Hamiltonian
operators and vector soliton equations for arbitrary (semi) Riemannian
spaces. We define the basic equations for nonholonomic curve flows. Then we
consider the properties of cosymplectic and symplectic operators adapted to
the nonlinear connection structure. Finally, there are constructed solitonic
hierarchies of bi--Hamiltonian anholonomic curve flows

We conclude the results in section 5. The Appendix contains necessary
definitions and formulas from the geometry of nonholonomic manifolds.

\section{Curvature Tensors with Constant Coefficients}

The idea behind an alternative description of general relativity is to
provide an equivalent re--formulation of geometric data for a (pseudo)
Riemannian metric $\mathbf{g}$ and Levi--Civita connection $\mathbf{\nabla }$
(in brief, we write such data $[\mathbf{g,\nabla }]),$ into a nonholonomic
structure with $[\mathbf{g,D}],$ where $\mathbf{D}$ is another metric
compatible linear connection, also defined by $\mathbf{g}$ in a unique form.

In this section, we prove that for any (semi) Riemannian metric $\mathbf{g}$
on a nonholonomic manifold $\mathbf{V}$ enabled with a nonlinear connection
(in brief, N--connection) structure $\mathbf{N,}$ defining a conventional
spacetime splitting of dimension $n+m,$ it is possible to construct a metric
compatible linear connection $\widetilde{\mathbf{D}}$ with constant matrix
coefficients of curvature, computed with respect to 'N--adapted' frames. We
outline in Appendix the basic definitions and notations from the geometry of
N--anholonomic Riemann manifolds, see details in \cite{ijgmmp1,vrfg,vsgg}.

\subsection{N--adapted frame transforms and (pseudo) Riemannian metrics}

We consider a manifold $\mathbf{V}$ of necessary smooth class and dimension $%
n+m,$ for $n\geq 2$ and $m\geq 1,$ enabled with a (semi) Riemannian metric
structure, i. e. with a second rank tensor of constant signature \footnote{%
in physical literature, there are used equivalent terms like (pseudo)
Riemannian or locally (pseudo )Euclidean/ Minskowski spaces} $\mathbf{g},$
see local formulas (\ref{metr}). On such a manifold, we can consider any $%
n+m $ splitting defined by a prescribed nonlinear connection (N--connection)
structure $\mathbf{N}$ with local coefficients $N_{i}^{a}(x,y)$ (\ref%
{coeffnc}), \ for indices of type $i,j,k,...=1,2,...n$ and $%
a,b,c,...=n+1,n+2,...,n+m.$\footnote{%
3+1 and 2+2 splitting, with different types of variables, are considered in
modern classical and quantum gravity; for instance, see discussion in Ref. %
\cite{vlqgfq}. Such a splitting establishes a local fibered structure
(holonomic or nonholonomic) and allows us to introduce Hamilton--momentum
like variables, or almost symplectic ones, which is convenient for
definition of conservation laws and quantization schemes of certain classes
of spacetimes.}

Haven defined a frame and respective co--frame (dual) structures on $\mathbf{%
V,}$ denoted correspondingly $e_{\alpha ^{\prime }}=(e_{i^{\prime
}},e_{a^{\prime }})$ and $e_{\ }^{\beta ^{\prime }}=(e^{i^{\prime
}},e^{a^{\prime }}),$ we can consider (in general, nonholonomic) \ frame
transforms
\begin{equation}
e_{\alpha }=A_{\alpha }^{\ \alpha ^{\prime }}(x,y)e_{\alpha ^{\prime }}%
\mbox{\ and\ }e_{\ }^{\beta }=A_{\ \beta ^{\prime }}^{\beta }(x,y)e^{\beta
^{\prime }}.  \label{nhft}
\end{equation}%
There are two important particular cases: 1) we can work with coordinate
bases (following our conventions, with underlined indices), $e_{\alpha
^{\prime }}\Leftrightarrow e_{\underline{\alpha }}=\partial _{\underline{%
\alpha }}=\partial /\partial u^{\underline{\alpha }}$ and $e^{\beta ^{\prime
}}\Leftrightarrow e^{\underline{\beta }}=du^{\underline{\beta }},$ and their
transform to arbitrary vielbeins $e_{\alpha }$ and $e_{\ }^{\beta };$ 2) it
is possible to introduce the so--called ''N--adapted'' bases $\mathbf{e}%
_{\alpha }$ (\ref{dder}) and cobases $\mathbf{e}^{\beta }$ (\ref{ddif})
using frame transforms (\ref{naft}) defined linearly by N--connection
coefficients $N_{i}^{a}.$

Under nonholonomic frame transforms $e_{\alpha ^{\prime }}\rightarrow
e_{\alpha },$ the metric coefficients of any metric structure $\mathbf{g}$
on $\mathbf{V}$ are computed following formulas
\begin{equation}
g_{\alpha \beta }(x,y)=A_{\alpha }^{~\ \alpha ^{\prime }}(x,y)~A_{\beta }^{\
\beta ^{\prime }}(x,y)g_{\alpha ^{\prime }\beta ^{\prime }}(x,y).
\label{auxm}
\end{equation}%
For a fixed frame structure $e_{\alpha ^{\prime }}$ on $\mathbf{V,}$ the
formula (\ref{auxm}) defines 'nonholonomic deformations' of metrics, $\
^{^{\prime }}\mathbf{g\rightarrow g}.$ In a particular case, we can
parametrize $A_{\alpha }^{~\ \alpha ^{\prime }}=\omega (x,y)\delta _{\alpha
}^{~\ \alpha ^{\prime }}$ and generate conformal transforms of metrics, $%
g_{\alpha \beta }=\omega ^{2}\ ^{^{\prime }}g_{\alpha \beta }.$\footnote{%
We shall use also equivalent denotations of type $\ ^{^{\prime }}\mathbf{g}%
_{\alpha \beta },$ instead of $\mathbf{g}_{\alpha ^{\prime }\beta ^{\prime
}},$ or $\ \underline{\mathbf{g}}_{\alpha \beta },$ instead of $\ \mathbf{g}%
_{\underline{\alpha }\underline{\beta }},$ with the coefficients of the same
metric computed with respect to different 'primed', or 'underlined'',
systems of reference (i.e. for frame transforms). In another turn, for a
fixed frame structure both for the ''prime'' and ''target'' geometric
configurations, $\ ^{^{\prime }}\mathbf{g=\{\mathbf{g}_{\alpha ^{\prime
}\beta ^{\prime }}\},}\ \mathbf{g=\{\mathbf{g}_{\alpha \beta }\}}$ and $\
\underline{\mathbf{g}}=\{\mathbf{g}_{\underline{\alpha }\underline{\beta }%
}\} $ mean, in general, different metric structures related via certain
nonholonomomic deformations of metrics (\ref{auxm}).}

\begin{definition}
A subclass of frame transforms (\ref{nhft}) [or deformations of metrics (\ref%
{auxm}), for fixed ''prime'' and ''target'' frame structures] is called
N--adapted if such nonholonomic transformations [deformations] preserve the $%
n+m$ splitting defined by a N--connection structure $\mathbf{N}%
=\{N_{i}^{a}\}.$
\end{definition}

For instance, N--adapted deformations of metrics are parametrized by such $%
A_{\alpha }^{~\ \alpha ^{\prime }}$ in (\ref{auxm}) when $\mathbf{g}_{\alpha
^{\prime }\beta ^{\prime }}=[g_{i^{\prime }j^{\prime }},g_{a^{\prime
}b^{\prime }}]\rightarrow \mathbf{g}_{\alpha \beta }=[g_{ij},g_{ab}].$ In an
alternative way, we can fix a metric structure $\mathbf{g}$ on $\mathbf{V}$
but consider N--adapted frame transforms (\ref{nhft}) preserving a locally
prescribed frame structure.

\begin{lemma}
\label{lemma1}For any fixed metric, \ $\mathbf{g},$ and N--connection, $%
\mathbf{N},$ structures, there are N--adapted frame transforms
\begin{eqnarray}
\mathbf{g} &=&g_{ij}(x,y)\ e^{i}\otimes e^{j}+h_{ab}(x,y)\ \mathbf{e}%
^{a}\otimes \mathbf{e}^{b},  \label{slme} \\
&=&g_{i^{\prime }j^{\prime }}(x,y)\ e^{i^{\prime }}\otimes e^{j^{\prime
}}+h_{a^{\prime }b^{\prime }}(x,y)\ \mathbf{e}^{a^{\prime }}\otimes \mathbf{e%
}^{b^{\prime }},  \notag
\end{eqnarray}%
where $\mathbf{e}^{a}$ and $\mathbf{e}^{a^{\prime }}$ are elongated
following formulas (\ref{ddif}), respectively by $N_{\ j}^{a}$ and
\begin{equation}
N_{\ j^{\prime }}^{a^{\prime }}=A_{a}^{~\ a^{\prime }}(x,y)A_{\ j^{\prime
}}^{j}(x,y)N_{\ j}^{a}(x,y),  \label{ncontri}
\end{equation}%
or, inversely,
\begin{equation}
N_{\ j}^{a}=A_{a^{\prime }}^{~\ a}(x,y)A_{\ j}^{j^{\prime }}(x,y)N_{\
j^{\prime }}^{a^{\prime }}(x,y)  \label{ncontr}
\end{equation}%
with prescribed $N_{\ j^{\prime }}^{a^{\prime }}.$
\end{lemma}

\begin{proof}
Any metric $\mathbf{g}$ (\ref{metr}) on a (pseudo) Riemannian manifold $%
\mathbf{V}$ can be represented as a d--metric $\mathbf{g}_{\alpha \beta
}=[g_{ij},g_{ab}]$ (\ref{m1}) if we prescribe a N--connection structure $%
\mathbf{N}=\{N_{i}^{a}\}$ (\ref{coeffnc}). We preserve the $n+m$ splitting
for any frame transform of type (\ref{nhft}) when
\begin{equation*}
g_{i^{\prime }j^{\prime }}=A_{\ i^{\prime }}^{i}A_{\ j^{\prime
}}^{j}g_{ij},\ h_{a^{\prime }b^{\prime }}=A_{\ a^{\prime }}^{a}A_{\
b^{\prime }}^{b}h_{ab},
\end{equation*}%
for $A_{i}^{~\ i^{\prime }}$ constrained to get holonomic $e^{i^{\prime
}}=A_{i}^{~\ i^{\prime }}e^{i},$ i.e. $[e^{i^{\prime }}e^{i^{\prime }}]=0$
and $\mathbf{e}^{a^{\prime }}=dy^{a^{\prime }}+N_{\ j^{\prime }}^{a^{\prime
}}dx^{j^{\prime }},$ for certain $x^{i^{\prime }}=x^{i^{\prime
}}(x^{i},y^{a})$ and $y^{a^{\prime }}=y^{a^{\prime }}(x^{i},y^{a}),$ with $%
N_{\ j^{\prime }}^{a^{\prime }}$ computed following formulas (\ref{ncontri}%
). The constructions can be equivalently inverted, when $g_{\alpha \beta }$
and $N_{i}^{a}$ are computed from $g_{\alpha ^{\prime }\beta ^{\prime }}$
and $N_{i^{\prime }}^{a^{\prime }},$ if both the metric and N--connection
splitting structures are fixed on $\mathbf{V}.$ $\square $
\end{proof}

\vskip4pt In this paper, we shall work with a fixed metric structure $%
\mathbf{g}$ on a (pseudo) Riemannian manifold $\mathbf{V}$ but consider such
$\mathbf{N}=\{N_{i}^{a}\}$ and N--adapted frame transforms/ deformations,
when certain metric compatible linear connections and their curvatures (also
uniquely defined by $\mathbf{g)}$ will satisfy the conditions necessary for
existence of solitonic hierarchies. This results in N--anholonomic
deformations of the geometric objects but the constructions can be
re--defined equivalently for the Levi--Civita connection.

\subsection{A d--connection with constant N--adapted coefficients}

From the class of metric compatible distinguished connections
(d--connecti\-ons) $\ ^{\mathbf{g}}\mathbf{D}$ (\ref{mcomp}),\footnote{%
see definitions and main formulas in Appendix \ref{assectdcon}} being
uniquely defined by a metric structure $\mathbf{g},$ we chose such a $n+m$
splitting with nontrivial $N_{i}^{a}(x,y)$ when with respect to a N--adapted
frame the canonical d--connection (\ref{candcon}) has constant coefficients.

\begin{proposition}
\label{prop21} Any (pseudo) Riemannian metric $\mathbf{g}$ on $\mathbf{V}$
defines a set of metric compatible d--connections of type
\begin{equation}
\ _{0}\widetilde{\mathbf{\Gamma }}_{\ \alpha ^{\prime }\beta ^{\prime
}}^{\gamma ^{\prime }}=\left( \widehat{L}_{j^{\prime }k^{\prime
}}^{i^{\prime }}=0,\widehat{L}_{b^{\prime }k^{\prime }}^{a^{\prime }}=\ _{0}%
\widehat{L}_{b^{\prime }k^{\prime }}^{a^{\prime }}=const,\widehat{C}%
_{j^{\prime }c^{\prime }}^{i^{\prime }}=0,\widehat{C}_{b^{\prime }c^{\prime
}}^{a^{\prime }}=0\right)  \label{ccandcon}
\end{equation}%
with respect to N--adapted frames (\ref{dder}) and (\ref{ddif}) for\ any $%
\mathbf{N}=\{N_{i^{\prime }}^{a^{\prime }}(x,y)\}$ being a nontrivial
solution of the system of equations%
\begin{equation}
2\ _{0}\widehat{L}_{b^{\prime }k^{\prime }}^{a^{\prime }}=\frac{\partial
N_{k^{\prime }}^{a^{\prime }}}{\partial y^{b^{\prime }}}-\ _{0}h^{a^{\prime
}c^{\prime }}\ _{0}h_{d^{\prime }b^{\prime }}\frac{\partial N_{k^{\prime
}}^{d^{\prime }}}{\partial y^{c^{\prime }}}  \label{auxf1}
\end{equation}%
for any nondegenerate constant--coefficients symmetric matrix $%
_{0}h_{d^{\prime }b^{\prime }}$ and its inverse $\ _{0}h^{a^{\prime
}c^{\prime }}.$
\end{proposition}

\begin{proof}
Using Lemma \ref{lemma1}, we express any metric $\mathbf{g}$ (\ref{metr}) as
a d--metric $\mathbf{g}_{\alpha ^{\prime }\beta ^{\prime }}=[g_{i^{\prime
}j^{\prime }},h_{a^{\prime }b^{\prime }}]$ (\ref{m1}), when certain constant
coefficients $g_{i^{\prime }j^{\prime }}=\ _{0}g_{i^{\prime }j^{\prime }}$
and $h_{a^{\prime }b^{\prime }}=\ _{0}h_{a^{\prime }b^{\prime }}$ are stated
with respect to a N--adapted coframe $\mathbf{e}^{\alpha ^{\prime
}}=[e^{i^{\prime }},\mathbf{e}^{a^{\prime }}]$ (the values $N_{k^{\prime
}}^{d^{\prime }}$ elongating $\mathbf{e}^{a^{\prime }}$ are computed
similarly to (\ref{ncontri})). Introducing such constant d--metric
coefficients into formulas (\ref{candcon}) (with 'primed' indices), we get
the canonical d--connection (\ref{ccandcon}) for any prescribed constant
values $\ _{0}\widehat{L}_{b^{\prime }k^{\prime }}^{a^{\prime }}.$ The
formula (\ref{auxf1}) follows from the formula for computing the
coefficients $\widehat{L}_{b^{\prime }k^{\prime }}^{a^{\prime }}$ in (\ref%
{candcon}). We may consider any $n+m$ splitting with $N_{k^{\prime
}}^{d^{\prime }}$ being a nontrivial solution of (\ref{auxf1}), which states
an explicit class of nonholonomic constraints on prescribed local fibered
structures. Such structures, for any nonholonomic transform of type (\ref%
{ncontr}), are very general ones with coefficients $N_{k}^{d}(x,y).$ We
conclude that for any metric structure $\mathbf{g}$ there is such a
nonholonomic local fibred structure, when the N--adapted coefficients of the
canonical d--connection are constant ones. Having prescribed the constant
values $\ _{0}g_{i^{\prime }j^{\prime }},\ _{0}h_{a^{\prime }b^{\prime }}$
and $\ _{0}\widehat{L}_{b^{\prime }k^{\prime }}^{a^{\prime }},$ a unique
solution of (\ref{auxf1}) and corresponding N--connection structure are
defined by the coefficients of $\mathbf{g}$ up to some N--adapted frame and
coordinate transforms. $\square $
\end{proof}

\vskip4pt It should be noted that the coefficients $\ _{\shortmid }\Gamma
_{\ \alpha ^{\prime }\beta ^{\prime }}^{\gamma ^{\prime }}$ of the
corresponding to $\mathbf{g}$ Levi--Civita connection $\ ^{\mathbf{g}}\nabla
$ are not constant with respect to N--adapted frames. They are computed
following formulas (\ref{cdeft}) and (\ref{cdeftc}).

\begin{theorem}
\label{theor01}The curvature d--tensor of a d--connection $\ _{0}\widetilde{%
\mathbf{\Gamma }}_{\ \alpha ^{\prime }\beta ^{\prime }}^{\gamma ^{\prime }}$
(\ref{ccandcon}) defined by a metric $\mathbf{g}$ has constant coefficients
with respect to N--adapted frames $\mathbf{e}_{\alpha ^{\prime }}=[\mathbf{e}%
_{i^{\prime }},e_{a^{\prime }}]$ and $\mathbf{e}^{\alpha ^{\prime
}}=[e^{i^{\prime }},\mathbf{e}^{a^{\prime }}]$ with $N_{k^{\prime
}}^{d^{\prime }}$ subjected to conditions (\ref{auxf1}).
\end{theorem}

\begin{proof}
It follows from Lemma \ref{lemma1} and Proposition \ref{prop21}. Really,
introducing constant coefficients $\ _{0}\widetilde{\mathbf{\Gamma }}_{\
\alpha ^{\prime }\beta ^{\prime }}^{\gamma ^{\prime }}$ (\ref{ccandcon})
into formulas (\ref{dcurv}), we get
\begin{eqnarray*}
\ _{0}\widetilde{\mathbf{R}}_{\ \beta ^{\prime }\gamma ^{\prime }\delta
^{\prime }}^{\alpha ^{\prime }} &=&(\ _{0}\widetilde{R}_{~h^{\prime
}j^{\prime }k^{\prime }}^{i^{\prime }}=0,\ _{0}\widetilde{R}_{~b^{\prime
}j^{\prime }k^{\prime }}^{a^{\prime }}=\ _{0}\widehat{L}_{\ b^{\prime
}j^{\prime }}^{c^{\prime }}\ _{0}\widehat{L}_{\ c^{\prime }k^{\prime
}}^{a^{\prime }}-\ _{0}\widehat{L}_{\ b^{\prime }k^{\prime }}^{c^{\prime }}\
_{0}\widehat{L}_{\ c^{\prime }j^{\prime }}^{a^{\prime }}= \\
&&cons,\ _{0}\widetilde{P}_{~h^{\prime }j^{\prime }a^{\prime }}^{i^{\prime
}}=0,\ _{0}\widetilde{P}_{~b^{\prime }j^{\prime }a^{\prime }}^{c^{\prime
}}=0,\ _{0}\widetilde{S}_{~j^{\prime }b^{\prime }c^{\prime }}^{i^{\prime
}}=0,\ _{0}\widetilde{S}_{~b^{\prime }d^{\prime }c^{\prime }}^{a^{\prime
}}=0).
\end{eqnarray*}%
Of course, in general, with respect to local coordinate (or other
N--adapted) frames, the curvature d--tensor $\widehat{\mathbf{R}}_{\ \beta
\gamma \delta }^{\alpha }$ does not have constant coefficients. Using
deformation relation (\ref{cdeft}), we can compute the corresponding Ricci
tensor $\ _{\shortmid }R_{\ \beta \gamma \delta }^{\alpha }$ for the
Levi--Civita connection $\ ^{\mathbf{g}}\nabla ,$ which is a general one
with 'non-constant' coefficients with respect to any local frames. $\square $
\end{proof}

\vskip4pt

One holds:

\begin{corollary}
\label{corol01}A d--connection $\ _{0}\widetilde{\mathbf{\Gamma }}_{\ \alpha
^{\prime }\beta ^{\prime }}^{\gamma ^{\prime }}$ (\ref{ccandcon}) has
constant scalar curvature.
\end{corollary}

\begin{proof}
It follows from the conditions of Theorem \ref{theor01} and formulas (\ref%
{dricci}) and (\ref{sdccurv}), when
\begin{equation*}
\ _{0}^{\sim }\overleftrightarrow{\mathbf{R}}\doteqdot \ _{0}\mathbf{g}%
^{\alpha ^{\prime }\beta ^{\prime }}\ _{0}\widetilde{\mathbf{R}}_{\alpha
^{\prime }\beta ^{\prime }}=\ _{0}g^{i^{\prime }j^{\prime }}\ _{0}\widetilde{%
R}_{i^{\prime }j^{\prime }}+\ _{0}h^{a^{\prime }b^{\prime }}\ _{0}\widetilde{%
S}_{a^{\prime }b^{\prime }}=\ _{0}^{\sim }\overrightarrow{R}+\ _{0}^{\sim }%
\overleftarrow{S}=cons.
\end{equation*}%
$\square $
\end{proof}

Even the condition $\ _{0}^{\sim }\overleftrightarrow{\mathbf{R}}=cons.$
holds true for a metric $\mathbf{g,}$ in general, the scalar curvature $\
_{\nabla }R$ of $\ ^{\mathbf{g}}\nabla ,$ for the same metric, is not
constant.

\subsection{Nonholonomic Einstein spaces}

Let us consider a subclass of (pseudo) Riemannian metrics defining exact
solutions of the Einstein equations for the canonical d--connection $%
\widehat{\mathbf{\Gamma }}_{\ \alpha \beta }^{\gamma }$ (\ref{candcon}),
\begin{equation}
\widehat{\mathbf{R}}_{\alpha \beta }-\frac{1}{2}\mathbf{g}_{\alpha \beta }%
\overleftrightarrow{\mathbf{R}}=\kappa \mathbf{\Upsilon }_{\alpha \beta }.
\label{einsteq}
\end{equation}%
For simplicity, we restrict our considerations to four dimensional (4D)
vacuum metrics, $\dim \mathbf{V}=4,n=2$ and $m=2,$ or generated by sources $%
\mathbf{\Upsilon }_{\alpha \beta }$ parametrized in the form
\begin{equation}
\mathbf{\Upsilon }_{\beta }^{\alpha }=[\Upsilon _{1}^{1}=\Upsilon
_{1}(x^{k},v),\Upsilon _{2}^{2}=\Upsilon _{1}(x^{k},v),\Upsilon
_{3}^{3}=\Upsilon _{3}(x^{k}),\Upsilon _{4}^{4}=\Upsilon _{3}(x^{k})],
\label{sdiag}
\end{equation}%
where indices and coordinates are labeled: $i,j,...=1,2$ and $%
a,b,...=3,4;x^{k}=(x^{1},x^{2})$ and $y^{a}=(y^{3}=v,y^{4}).$ There are also
nonholonomic gravitational configurations when the source (\ref{sdiag})
parametrized $\Upsilon _{1}=\lambda _{1}=cons.$ and $\Upsilon _{3}=\lambda
_{3}=cons.$ Here we note that gravitational interactions parametized by
generic off--diagonal metrics model locally anisotropic configurations even
if $\lambda _{1}=\lambda _{3},$ or both values vanish (for a nonholonomic
vacuum case).

We parametrize the N--connection coefficients in the form $%
N_{i}^{3}=w_{i}(x^{k},v)$ and $N_{i}^{4}=n_{i}(x^{k},v);$ consider partial
derivatives $a^{\bullet }=\partial a/\partial x^{1},$\ $a^{\prime }=\partial
a/\partial x^{2},$\ $a^{\ast }=\partial a/\partial v$ and introduce the
'polarization' function
\begin{equation}
\varsigma \left( x^{i},v\right) =\ ^{0}\varsigma -\frac{\epsilon _{3}}{8}\
^{0}h^{2}\int \Upsilon _{1}f^{\ast }\left[ f-\ ^{0}f\right] dv,
\label{aux02}
\end{equation}%
for an arbitrary nontrivial function $f=f\left( x^{i},v\right) $ with $%
f^{\ast }\neq 0$ and given functions $\ ^{0}\varsigma =\ ^{0}\varsigma
\left( x^{i}\right) ,\ ^{0}h=\ ^{0}h(x^{i})$ and $\ ^{0}f=$ $\ ^{0}f(x^{i}),$
when $\epsilon _{\alpha }=\pm 1$ will define the spacetime signature. We
shall consider also some 'integration' functions $\ ^{1}n_{k}=\
^{1}n_{k}\left( x^{i}\right) $ and $\ ^{2}n_{k}=\ ^{2}n_{k}\left(
x^{i}\right) .$

\begin{theorem}
Any metric $\mathbf{g}$ (\ref{metr}) (equivalently, d--metric (\ref{m1})) on
a 4D N--anholonomic manifold $\mathbf{V}$ which by N--adapted frame
transforms of type (\ref{slme}), (\ref{ncontri}) and (\ref{ncontr}) can be
parametrized by an ansatz
\begin{eqnarray}
\ ^{\circ }\mathbf{g} &=&e^{\psi (x^{k})}\left[ \epsilon _{1}\ dx^{1}\otimes
dx^{1}+\epsilon _{2}\ dx^{2}\otimes dx^{2}\right] +  \label{es4s} \\
&&h_{3}\left( x^{k},v\right) \ \delta v\otimes \delta v+h_{4}\left(
x^{k},v\right) \ \delta y^{4}\otimes \delta y^{4},  \notag \\
\delta v &=&dv+w_{i}\left( x^{k},v\right) dx^{i},\ \delta
y^{4}=dy^{4}+n_{i}\left( x^{k},v\right) dx^{i}  \notag
\end{eqnarray}%
defines a class of exact solution the Einstein equations (\ref{einsteq})
with nontrivial sources (\ref{sdiag}) if
\begin{eqnarray}
&&\epsilon _{1}\psi ^{\bullet \bullet }+\epsilon _{2}\psi ^{^{\prime \prime
}}=\Upsilon _{3},h_{3}=\epsilon _{3}h_{0}^{2}\left[ f^{\ast }\right]
^{2}|\varsigma |,\ h_{4}=\epsilon _{4}\left[ f-\ ^{0}f\right] ^{2},
\label{coeff4d} \\
&&w_{i}=-\partial _{i}\varsigma /\varsigma ^{\ast },\ n_{k}=\ ^{1}n_{k}+\
^{2}n_{k}\int \varsigma \left[ f^{\ast }\right] ^{2}\left[ f-\ ^{0}f\right]
^{-3}dv,  \notag
\end{eqnarray}%
where $\varsigma $ is computed for a given $\Upsilon _{1}$ following formula
(\ref{aux02}) with $\varsigma =1$ in the vacuum case.
\end{theorem}

\begin{proof}
Details of the proof (consisting from a straightforward verification that
the ansatz (\ref{es4s}) with coefficients (\ref{coeff4d}) really define
integral varieties of nonholonomic Einstein \ equations) can be found in
Refs. \cite{ijgmmp1,vsgg}. $\square $
\end{proof}

\vskip4pt

It should be emphasized that this class of solutions are generated by an
arbitrary nontrivial function $f\left( x^{i},v\right) $ (with $f^{\ast }\neq
0),$ sources $\Upsilon _{1}(x^{k},v)$ and $\Upsilon _{3}\left( x^{i}\right) $
and integration functions and constants. Such values for the corresponding
signatures $\epsilon _{\alpha }=\pm 1$ have to be defined by certain
boundary conditions and physical symmetries. We can extract integral
varieties of the Einstein equations for the Levi--Civita connection $\ ^{%
\mathbf{g}}\nabla $ if we constrain additionally the metric coefficients to
satisfy the conditions
\begin{eqnarray}
h_{4}^{\ast }\phi /h_{3}h_{4} &=&\Upsilon _{1},  \label{ep2b} \\
w_{1}^{\prime }-w_{2}^{\bullet }+w_{2}w_{1}^{\ast }-w_{1}w_{2}^{\ast } &=&0,
\label{ep2b1} \\
n_{1}^{\prime }-n_{2}^{\bullet } &=&0,  \label{ep2b2}
\end{eqnarray}%
for $w_{i}=\partial _{i}\phi /\phi ^{\ast },$ where $\phi =\ln |h_{4}^{\ast
}/\sqrt{|h_{3}h_{4}|}|.$ The conditions (\ref{ep2b}) for $\Upsilon
_{1}\rightarrow 0$ are satisfied for $h_{4}^{\ast }\neq 1$ if $h_{4}^{\ast
}\phi \rightarrow 0.$

Following Theorem \ref{theor01} and Corollary \ref{corol01}, to any solution
of the Einstein equations $\ ^{\circ }\mathbf{g}$ (\ref{es4s}), constrained/
or not by conditions (\ref{ep2b})--(\ref{ep2b2}), we can associate a
d--connection $\ _{0}\widetilde{\mathbf{\Gamma }}_{\ \alpha ^{\prime }\beta
^{\prime }}^{\gamma ^{\prime }}$ (\ref{ccandcon}) with constant coefficients
for the Riemannian and Ricci tensors with respect to corresponding
N--adapted bases. Different classes of solutions are characterized by
different values of such matrix coefficients.

More general classes of such Einstein spaces with Killing symmetries and
their parametric and nonholonomic deformations considered in Ref. \cite%
{vncg,ijgmmp1,vsgg} can be also generated by (depending on certain
parameters and integration constants) curvature and Ricci d--tensor matrices
and constant scalar curvature of corresponding d--metrics. So, following the
above outlined method we can always encode the data for a known solution of
the Einstein equations and generalizations into terms of constant
coefficients d--connections and curvatures. To invert the problem and
construct from certain solitonic hierachies some solutions of Einstein
equations with non-Killing symmetries it is also possible, but it is still
not clear how to approach this task in general form.

\section{Nonholonomic Curve Flows}

We outline the geometry of curve flows adapted to a N--connection structure
on a (pseudo) Riemannian manifold $\mathbf{V.}$

\subsection{Non--stretching and N--adapted curve flows}

A non--stretching curve $\gamma (\tau ,\mathbf{l})$ on $\mathbf{V,}$ where $%
\tau $ is a real parameter and $\mathbf{l}$ is the arclength of the curve on
$\mathbf{V,}$ is defined with such evolution d--vector $\mathbf{Y}=\gamma
_{\tau }$ and tangent d--vector $\mathbf{X}=\gamma _{\mathbf{l}}$ that $%
\mathbf{g(X,X)=}1\mathbf{.}$ Such a curve $\gamma (\tau ,\mathbf{l})$ swept
out a two--dimensional surface in $T_{\gamma (\tau ,\mathbf{l})}\mathbf{V}%
\subset T\mathbf{V.}$

We shall work with N--adapted bases (\ref{dder}) and (\ref{ddif}) and
connection 1--form $\mathbf{\Gamma }_{\ \beta }^{\alpha }=\mathbf{\Gamma }%
_{\ \beta \gamma }^{\alpha }\mathbf{e}^{\gamma }$ the canonical
d--connection operator $\mathbf{D}$ (\ref{candcon}) acting in the form%
\begin{equation}
\mathbf{D}_{\mathbf{X}}\mathbf{e}_{\alpha }=(\mathbf{X\rfloor \Gamma }%
_{\alpha \ }^{\ \gamma })\mathbf{e}_{\gamma }\mbox{ and }\mathbf{D}_{\mathbf{%
Y}}\mathbf{e}_{\alpha }=(\mathbf{Y\rfloor \Gamma }_{\alpha \ }^{\ \gamma })%
\mathbf{e}_{\gamma },  \label{part01}
\end{equation}%
where ''$\mathbf{\rfloor "}$ denotes the interior product and the indices
are lowered and raised respectively by the d--metric $\mathbf{g}_{\alpha
\beta }=[g_{ij},h_{ab}]$ and its inverse $\mathbf{g}^{\alpha \beta
}=[g^{ij},h^{ab}].$ We note that $\mathbf{D}_{\mathbf{X}}=\mathbf{X}^{\alpha
}\mathbf{D}_{\alpha }$ is a covariant derivation operator along curve $%
\gamma (\tau ,\mathbf{l}).$ It is convenient to fix the N--adapted frame to
be parallel to curve $\gamma (\mathbf{l})$ adapted in the form
\begin{eqnarray}
e^{1} &\doteqdot &h\mathbf{X,}\mbox{ for }i=1,\mbox{ and }e^{\widehat{i}},%
\mbox{ where }h\mathbf{g(}h\mathbf{X,}e^{\widehat{i}}\mathbf{)=}0,
\label{curvframe} \\
\mathbf{e}^{n+1} &\doteqdot &v\mathbf{X,}\mbox{ for }a=n+1,\mbox{ and }%
\mathbf{e}^{\widehat{a}},\mbox{ where }v\mathbf{g(}v\mathbf{X,\mathbf{e}}^{%
\widehat{a}}\mathbf{)=}0,  \notag
\end{eqnarray}%
for $\widehat{i}=2,3,...n$ and $\widehat{a}=n+2,n+3,...,n+m.$ For such
frames, the covariant derivative of each ''normal'' d--vectors $\mathbf{e}^{%
\widehat{\alpha }}$ results into the d--vectors adapted to $\gamma (\tau ,%
\mathbf{l}),$
\begin{eqnarray}
\mathbf{D}_{\mathbf{X}}e^{\widehat{i}} &\mathbf{=}&\mathbf{-}\rho ^{\widehat{%
i}}\mathbf{(}u\mathbf{)\ X}\mbox{ and }\mathbf{D}_{h\mathbf{X}}h\mathbf{X}%
=\rho ^{\widehat{i}}\mathbf{(}u\mathbf{)\ \mathbf{e}}_{\widehat{i}},
\label{part02} \\
\mathbf{D}_{\mathbf{X}}\mathbf{\mathbf{e}}^{\widehat{a}} &\mathbf{=}&\mathbf{%
-}\rho ^{\widehat{a}}\mathbf{(}u\mathbf{)\ X}\mbox{ and }\mathbf{D}_{v%
\mathbf{X}}v\mathbf{X}=\rho ^{\widehat{a}}\mathbf{(}u\mathbf{)\ }e_{\widehat{%
a}},  \notag
\end{eqnarray}%
which holds for certain classes of functions $\rho ^{\widehat{i}}\mathbf{(}u%
\mathbf{)}$ and $\rho ^{\widehat{a}}\mathbf{(}u\mathbf{).}$ The formulas (%
\ref{part01}) and (\ref{part02}) are distinguished into h-- and
v--components for $\mathbf{X=}h\mathbf{X}+v\mathbf{X}$ and $\mathbf{D=(}h%
\mathbf{D},v\mathbf{D)}$ for $\mathbf{D=\{\Gamma }_{\ \alpha \beta }^{\gamma
}\},h\mathbf{D}=\{L_{jk}^{i},L_{bk}^{a}\}$ and $v\mathbf{D=\{}%
C_{jc}^{i},C_{bc}^{a}\}.$

Along $\gamma (\mathbf{l}),$ we can move differential forms in a parallel
N--adapted form. For instance, $\mathbf{\Gamma }_{\ \mathbf{X}}^{\alpha
\beta }\doteqdot \mathbf{X\rfloor \Gamma }_{\ }^{\alpha \beta }.$ An
algebraic characterization of such spaces, can be obtained if we perform a
frame transform preserving the decomposition (\ref{whitney}) to an
orthonormalized basis $\mathbf{e}_{\alpha ^{\prime }},$ when
\begin{equation}
\mathbf{e}_{\alpha }\rightarrow e_{\alpha }^{\ \alpha ^{\prime }}(u)\
\mathbf{e}_{\alpha ^{\prime }},  \label{orthbas}
\end{equation}%
called an orthonormal d--basis. In this case, the coefficients of the
d--metric (\ref{m1}) transform into the (pseudo) Euclidean one\textbf{\ }$%
\mathbf{g}_{\alpha ^{\prime }\beta ^{\prime }}=\eta _{\alpha ^{\prime }\beta
^{\prime }}.$ In distinguished form, we obtain two skew matrices%
\begin{equation*}
\mathbf{\Gamma }_{h\mathbf{X}}^{i^{\prime }j^{\prime }}\doteqdot h\mathbf{%
X\rfloor \Gamma }_{\ }^{i^{\prime }j^{\prime }}=2\ e_{h\mathbf{X}%
}^{[i^{\prime }}\ \rho ^{j^{\prime }]}\mbox{ and }\mathbf{\Gamma }_{v\mathbf{%
X}}^{a^{\prime }b^{\prime }}\doteqdot v\mathbf{X\rfloor \Gamma }_{\
}^{a^{\prime }b^{\prime }}=2\mathbf{\ e}_{v\mathbf{X}}^{[a^{\prime }}\ \rho
^{b^{\prime }]},
\end{equation*}%
where
\begin{equation*}
\ e_{h\mathbf{X}}^{i^{\prime }}\doteqdot g(h\mathbf{X,}e^{i^{\prime }})=[1,%
\underbrace{0,\ldots ,0}_{n-1}]\mbox{ and }\ e_{v\mathbf{X}}^{a^{\prime
}}\doteqdot h(v\mathbf{X,}e^{a^{\prime }})=[1,\underbrace{0,\ldots ,0}_{m-1}]
\end{equation*}%
and
\begin{equation*}
\mathbf{\Gamma }_{h\mathbf{X\,}i^{\prime }}^{\qquad j^{\prime }}=\left[
\begin{array}{cc}
0 & \rho ^{j^{\prime }} \\
-\rho _{i^{\prime }} & \mathbf{0}_{[h]}%
\end{array}%
\right] \mbox{ and }\mathbf{\Gamma }_{v\mathbf{X\,}a^{\prime }}^{\qquad
b^{\prime }}=\left[
\begin{array}{cc}
0 & \rho ^{b^{\prime }} \\
-\rho _{a^{\prime }} & \mathbf{0}_{[v]}%
\end{array}%
\right]
\end{equation*}%
with $\mathbf{0}_{[h]}$ and $\mathbf{0}_{[v]}$ being respectively $%
(n-1)\times (n-1)$ and $(m-1)\times (m-1)$ matrices.

The above presented row--matrices and skew--matrices show that locally an
N--anholonomic manifold $\mathbf{V}$ of dimension $n+m,$ with respect to
distinguished orthonormalized frames are characterized algebraically by
couples of unit vectors in $\mathbb{R}^{n}$ and $\mathbb{R}^{m}$ preserved
respectively by the $SO(n-1)$ and $SO(m-1)$ rotation subgroups of the local
N--adapted frame structure group $SO(n)\oplus SO(m).$ The connection
matrices $\mathbf{\Gamma }_{h\mathbf{X\,}i^{\prime }}^{\qquad j^{\prime }}$
and $\mathbf{\Gamma }_{v\mathbf{X\,}a^{\prime }}^{\qquad b^{\prime }}$
belong to the orthogonal complements of the corresponding Lie subalgebras
and algebras, $\mathfrak{so}(n-1)\subset \mathfrak{so}(n)$ and $\mathfrak{so}%
(m-1)\subset \mathfrak{so}(m).$

The torsion (\ref{dtors}) and curvature (\ref{curv}) tensors can be written
in orthonormalized component form with respect to (\ref{curvframe}) mapped
into a distinguished orthonormalized dual frame (\ref{orthbas}),%
\begin{equation}
\mathcal{T}^{\alpha ^{\prime }}\doteqdot \mathbf{D}_{\mathbf{X}}\mathbf{e}_{%
\mathbf{Y}}^{\alpha ^{\prime }}-\mathbf{D}_{\mathbf{Y}}\mathbf{e}_{\mathbf{X}%
}^{\alpha ^{\prime }}+\mathbf{e}_{\mathbf{Y}}^{\beta ^{\prime }}\mathbf{%
\Gamma }_{\mathbf{X}\beta ^{\prime }}^{\quad \alpha ^{\prime }}-\mathbf{e}_{%
\mathbf{X}}^{\beta ^{\prime }}\mathbf{\Gamma }_{\mathbf{Y}\beta ^{\prime
}}^{\quad \alpha ^{\prime }}  \label{mtors}
\end{equation}%
and
\begin{equation}
\mathcal{R}_{\beta ^{\prime }}^{\;\alpha ^{\prime }}(\mathbf{X,Y})=\mathbf{D}%
_{\mathbf{Y}}\mathbf{\Gamma }_{\mathbf{X}\beta ^{\prime }}^{\quad \alpha
^{\prime }}-\mathbf{D}_{\mathbf{X}}\mathbf{\Gamma }_{\mathbf{Y}\beta
^{\prime }}^{\quad \alpha ^{\prime }}+\mathbf{\Gamma }_{\mathbf{Y}\beta
^{\prime }}^{\quad \gamma ^{\prime }}\mathbf{\Gamma }_{\mathbf{X}\gamma
^{\prime }}^{\quad \alpha ^{\prime }}-\mathbf{\Gamma }_{\mathbf{X}\beta
^{\prime }}^{\quad \gamma ^{\prime }}\mathbf{\Gamma }_{\mathbf{Y}\gamma
^{\prime }}^{\quad \alpha ^{\prime }},  \label{mcurv}
\end{equation}%
where $\mathbf{e}_{\mathbf{Y}}^{\alpha ^{\prime }}\doteqdot \mathbf{g}(%
\mathbf{Y},\mathbf{e}^{\alpha ^{\prime }})$ and $\mathbf{\Gamma }_{\mathbf{Y}%
\beta ^{\prime }}^{\quad \alpha ^{\prime }}\doteqdot \mathbf{Y\rfloor \Gamma
}_{\beta ^{\prime }}^{\;\alpha ^{\prime }}=\mathbf{g}(\mathbf{e}^{\alpha
^{\prime }},\mathbf{D}_{\mathbf{Y}}\mathbf{e}_{\beta ^{\prime }})$ define
respectively the N--adapted orthonormalized frame row--matrix and the
canonical d--connection skew--matrix in the flow directs, and $\mathcal{R}%
_{\beta ^{\prime }}^{\;\alpha ^{\prime }}(\mathbf{X,Y})\doteqdot \mathbf{g}(%
\mathbf{e}^{\alpha ^{\prime }},[\mathbf{D}_{\mathbf{X}},$ $\mathbf{D}_{%
\mathbf{Y}}]\mathbf{e}_{\beta ^{\prime }})$ is the curvature matrix.

\subsection{N--anholonomic manifolds with constant matrix curvature}

For trivial N--connection curvature and torsion but constant matrix
curvature, we get a holonomic Riemannian manifold and the equations (\ref%
{mtors}) and (\ref{mcurv}) directly encode a bi--Hamiltonian structure \cite%
{saw,anc2}. A well known class of Riemannian manifolds for which the frame
curvature matrix constant consists of the symmetric spaces $M=G/H$ for
compact semisimple Lie groups $G\supset H.$ A complete classification and
summary of main results on such spaces are given in Refs. \cite{helag,kob}.

The Riemannian curvature and the metric tensors for $M=G/H$ are covariantly
constant and $G$--invariant resulting in constant curvature matrix. In \cite%
{anc1,anc2}, the bi--Hamiltonian operators were investigated for the
symmetric spaces with $M=G/SO(n)$ with $H=SO(n)\supset O(n-1)$ and two
examples when $G=SO(n+1),SU(n).$ Then it was exploited the existing
canonical soldering of Klein and Riemannian symmetric--space geometries \cite%
{sharpe}. Such results were proven for the Levi--Civita connection on
symmetric spaces. They can be generalized for any (pseudo) Riemannian
manifold but for an auxilliary d--connection $\ _{0}\widetilde{\mathbf{%
\Gamma }}_{\ \alpha ^{\prime }\beta ^{\prime }}^{\gamma ^{\prime }}$ (\ref%
{ccandcon}).

\subsubsection{Symmetric nonholonomic manifolds}

N--anholonomic manifolds are characterized by a conventional nonholonomic
splitting of dimensions. We suppose that the ''horizontal'' distribution is
a symmetric space $hV=hG/SO(n)$ with the isotropy subgroup $hH=SO(n)\supset
O(n)$ and the typical fiber space to be a symmetric space $F=vG/SO(m)$ with
the isotropy subgroup $vH=SO(m)\supset O(m).$ This means that $hG=SO(n+1)$
and $vG=SO(m+1)$ which is enough for a study of real holonomic and
nonholonomic manifolds and geometric mechanics models.\footnote{%
it is necessary to consider $hG=SU(n)$ and $vG=SU(m)$ for the geometric
models with spinor and gauge fields} \

Our aim is to solder in a canonic way (like in the N--connection geometry)
the horizontal and vertical symmetric Riemannian spaces of dimension $n$ and
$m$ with a (total) symmetric Riemannian space $V$ of dimension $n+m,$ when $%
V=G/SO(n+m)$ with the isotropy group $H=SO(n+m)\supset O(n+m)$ and $%
G=SO(n+m+1).$ First, we note that for the just mentioned horizontal,
vertical and total symmetric Riemannian spaces one exists natural settings
to Klein geometry. For instance, the metric tensor $hg=\{\mathring{g}_{ij}\}$
on $h\mathbf{V}$ is defined by the Cartan--Killing inner product $<\cdot
,\cdot >_{h}$ on $T_{x}hG\simeq h\mathfrak{g}$ restricted to the Lie algebra
quotient spaces $h\mathfrak{p=}h\mathfrak{g/}h\mathfrak{h,}$ with $%
T_{x}hH\simeq h\mathfrak{h,}$ where $h\mathfrak{g=}h\mathfrak{h}\oplus h%
\mathfrak{p}$ is stated such that there is an involutive automorphism of $hG$
under $hH$ is fixed, i.e. $[h\mathfrak{h,}h\mathfrak{p]}\subseteq $ $h%
\mathfrak{p}$ and $[h\mathfrak{p,}h\mathfrak{p]}\subseteq h\mathfrak{h.}$ In
a similar form, we can define the group spaces and related inner products
and\ Lie algebras,%
\begin{eqnarray}
\mbox{for\ }vg &=&\{\mathring{h}_{ab}\},\;<\cdot ,\cdot
>_{v},\;T_{y}vG\simeq v\mathfrak{g,\;}v\mathfrak{p=}v\mathfrak{g/}v\mathfrak{%
h,}\mbox{ with }  \notag \\
T_{y}vH &\simeq &v\mathfrak{h,}v\mathfrak{g=}v\mathfrak{h}\oplus v\mathfrak{%
p,}\mbox{where }\mathfrak{\;}[v\mathfrak{h,}v\mathfrak{p]}\subseteq v%
\mathfrak{p,\;}[v\mathfrak{p,}v\mathfrak{p]}\subseteq v\mathfrak{h;}  \notag
\\
&&  \label{algstr} \\
\mbox{for\ }\mathbf{g} &=&\{\mathring{g}_{\alpha \beta }\},\;<\cdot ,\cdot
>_{\mathbf{g}},\;T_{(x,y)}G\simeq \mathfrak{g,\;p=g/h,}\mbox{ with }  \notag
\\
T_{(x,y)}H &\simeq &\mathfrak{h,g=h}\oplus \mathfrak{p,}\mbox{where }%
\mathfrak{\;}[\mathfrak{h,p]}\subseteq \mathfrak{p,\;}[\mathfrak{p,p]}%
\subseteq \mathfrak{h.}  \notag
\end{eqnarray}%
We parametrize the metric structure with constant coefficients on $%
V=G/SO(n+m)$ in the form%
\begin{equation*}
\mathring{g}=\mathring{g}_{\alpha \beta }du^{\alpha }\otimes du^{\beta },
\end{equation*}%
where $u^{\alpha }$ are local coordinates and
\begin{equation}
\mathring{g}_{\alpha \beta }=\left[
\begin{array}{cc}
\mathring{g}_{ij}+\mathring{N}_{i}^{a}N_{j}^{b}\mathring{h}_{ab} & \mathring{%
N}_{j}^{e}\mathring{h}_{ae} \\
\mathring{N}_{i}^{e}\mathring{h}_{be} & \mathring{h}_{ab}%
\end{array}%
\right]  \label{constans}
\end{equation}%
when trivial, constant, N--connection coefficients are computed $\mathring{N}%
_{j}^{e}=\mathring{h}^{eb}\mathring{g}_{jb}$ for any given sets $\mathring{h}%
^{eb}$ and $\mathring{g}_{jb},$ i.e. from the inverse metrics coefficients
defined respectively on $hG=SO(n+1)$ and by off--blocks $(n\times n)$-- and $%
(m\times m)$--terms of the metric $\mathring{g}_{\alpha \beta }.$ As a
result, we define an equivalent d--metric structure of type (\ref{m1})
\begin{eqnarray}
\mathbf{\mathring{g}} &=&\ \mathring{g}_{ij}\ e^{i}\otimes e^{j}+\ \mathring{%
h}_{ab}\ \mathbf{\mathring{e}}^{a}\otimes \mathbf{\mathring{e}}^{b},
\label{m1const} \\
e^{i} &=&dx^{i},\ \;\mathbf{\mathring{e}}^{a}=dy^{a}+\mathring{N}%
_{i}^{a}dx^{i}  \notag
\end{eqnarray}%
defining a trivial $(n+m)$--splitting $\mathbf{\mathring{g}=}\mathring{g}%
\mathbf{\oplus _{\mathring{N}}}\mathring{h}\mathbf{\ }$because all
nonholonomy coefficients $\mathring{W}_{\alpha \beta }^{\gamma }$ and
N--connection curvature coefficients $\mathring{\Omega}_{ij}^{a}$ are zero.

We can consider any covariant coordinate transforms of (\ref{m1const})
preserving the\ $(n+m)$--splitting resulting in any $W_{\alpha \beta
}^{\gamma }=0$ (\ref{anhrel}) and $\Omega _{ij}^{a}=0$ (\ref{ncurv}). It
should be noted that even such trivial parametrizations define algebraic
classifications of \ symmetric Riemannian spaces of dimension $n+m$ with
constant matrix curvature admitting splitting (by certain algebraic
constraints) into symmetric Riemannian subspaces of dimension $n$ and $m,$
also both with constant matrix curvature and introducing the concept of
N--anholonomic Riemannian space of type $\mathbf{\mathring{V}}=[hG=SO(n+1),$
$vG=SO(m+1),\;\mathring{N}_{i}^{e}].$ There are such trivially
N--anholonomic group spaces which possess a Lie d--algebra symmetry $%
\mathfrak{so}_{\mathring{N}}(n+m)\doteqdot \mathfrak{so}(n)\oplus \mathfrak{%
so}(m).$

The next generalization of constructions is to consider nonhlonomic
distributions on $V=G/SO(n+m)$ defined locally by arbitrary N--connection
coefficients $N_{i}^{a}(x,y)$ with nonvanishing $W_{\alpha \beta }^{\gamma }$
and $\Omega _{ij}^{a}$ but with constant d--metric coefficients when
\begin{eqnarray}
\mathbf{g} &=&\ \mathring{g}_{ij}\ e^{i}\otimes e^{j}+\ \mathring{h}_{ab}\
\mathbf{e}^{a}\otimes \mathbf{e}^{b},  \label{m1b} \\
e^{i} &=&dx^{i},\ \mathbf{e}^{a}=dy^{a}+N_{i}^{a}(x,y)dx^{i}.  \notag
\end{eqnarray}%
This metric is equivalent to a metric (\ref{metr}) considered in\
Proposition \ref{prop21} when d--metric $\mathbf{g}_{\alpha ^{\prime }\beta
^{\prime }}=[g_{i^{\prime }j^{\prime }},h_{a^{\prime }b^{\prime }}]$ (\ref%
{m1}) with constant coefficients $g_{i^{\prime }j^{\prime }}=\
_{0}g_{i^{\prime }j^{\prime }}=\ \mathring{g}_{ij}$ and $h_{a^{\prime
}b^{\prime }}=\ _{0}h_{a^{\prime }b^{\prime }}=\mathring{h}_{ab},$ in this
section $\ $\ induced by the corresponding Lie d--algebra structure $%
\mathfrak{so}_{\mathring{N}}(n+m).$ Such spaces transform into
N--anholonomic Riemann--Cartan manifolds $\mathbf{\mathring{V}}_{\mathbf{N}%
}=[hG=SO(n+1),$ $vG=SO(m+1),\;N_{i}^{e}]$ \ with nontrivial N--connection
curvature and induced d--torsion coefficients of the canonical d--connection
(see formulas (\ref{dtors}) computed for constant d--metric coefficients and
the canonical d--connection coefficients in (\ref{candcon})). One has zero
curvature for the canonical d--connection (in general, such spaces are
curved ones with generic off--diagonal metric (\ref{m1b}) and nonzero
curvature tensor for the Levi--Civita connection).\footnote{%
Introducing, constant values for the d--metric coefficients we get zero
coefficients for the canonical d--connection which in its turn results in
zero values of (\ref{dcurv}).} This allows us to classify the N--anholonomic
manifolds (and vector bundles) as having the same group and algebraic
structures of couples of symmetric Riemannian spaces of dimension $n$ and $m$
but nonholonomically soldered to the symmetric Riemannian space of dimension
$n+m.$ With respect to N--adapted orthonormal bases (\ref{orthbas}), with
distinguished h-- and v--subspaces, we obtain the same inner products and
group and Lie algebra spaces as in (\ref{algstr}).

\subsubsection{N--anholonomic Klein spaces}

The bi--Hamiltonian and solitonic constructions are based on an extrinsic
approach soldering the Riemannian symmetric--space geometry to the Klein
geometry \cite{sharpe}. For the N--anholonomic spaces of dimension $n+m,$
with a constant d--curvature, similar constructions hold true but we have to
adapt them to the N--connection structure.

There are two Hamiltonian variables given by the principal normals $%
\;^{h}\nu $ and $\;^{v}\nu ,$ respectively, in the horizontal and vertical
subspaces, defined by the canonical d--connection $\mathbf{D}=(h\mathbf{D},v%
\mathbf{D}),$ see formulas (\ref{curvframe}) and (\ref{part02}),
\begin{equation*}
\;^{h}\nu \doteqdot \mathbf{D}_{h\mathbf{X}}h\mathbf{X}=\nu ^{\widehat{i}}%
\mathbf{\mathbf{e}}_{\widehat{i}}\mbox{\ and \ }\;^{v}\nu \doteqdot \mathbf{D%
}_{v\mathbf{X}}v\mathbf{X}=\nu ^{\widehat{a}}e_{\widehat{a}}.
\end{equation*}%
This normal d--vector $\mathbf{v}=(\;^{h}\nu ,$ $\;^{v}\nu ),$ with
components of type $\mathbf{\nu }^{\alpha }=(\nu ^{i},$ $\;\nu ^{a})=(\nu
^{1},$ $\nu ^{\widehat{i}},\nu ^{n+1},\nu ^{\widehat{a}}),$ is in the
tangent direction of curve $\gamma .$ There is also the principal normal
d--vector $\mathbf{\varpi }=(\;^{h}\varpi ,\;^{v}\varpi )$ with components
of type $\mathbf{\varpi }^{\alpha }=(\varpi ^{i},$ $\;\varpi ^{a})=(\varpi
^{1},\varpi ^{\widehat{i}},\varpi ^{n+1},\varpi ^{\widehat{a}})$ in the flow
direction, with
\begin{equation*}
\;^{h}\varpi \doteqdot \mathbf{D}_{h\mathbf{Y}}h\mathbf{X=}\varpi ^{\widehat{%
i}}\mathbf{\mathbf{e}}_{\widehat{i}},\;^{v}\varpi \doteqdot \mathbf{D}_{v%
\mathbf{Y}}v\mathbf{X}=\varpi ^{\widehat{a}}e_{\widehat{a}},
\end{equation*}%
representing a Hamiltonian d--covector field. We can consider that the
normal part of the flow d--vector $\ \mathbf{h}_{\perp }\doteqdot \mathbf{Y}%
_{\perp }=h^{\widehat{i}}\mathbf{\mathbf{e}}_{\widehat{i}}+h^{\widehat{a}}e_{%
\widehat{a}}$ represents a Hamiltonian d--vector field. For such
configurations, we can consider parallel N--adapted frames $\mathbf{e}%
_{\alpha ^{\prime }}=(\mathbf{e}_{i^{\prime }},e_{a^{\prime }})$ when the
h--variables $\nu ^{\widehat{i^{\prime }}},$ $\varpi ^{\widehat{i^{\prime }}%
},h^{\widehat{i^{\prime }}}$ are respectively encoded in the top row of the
horizontal canonical d--connection matrices $\mathbf{\Gamma }_{h\mathbf{X\,}%
i^{\prime }}^{\qquad j^{\prime }}$ and $\mathbf{\Gamma }_{h\mathbf{Y\,}%
i^{\prime }}^{\qquad j^{\prime }}$ and in the row matrix $\left( \mathbf{e}_{%
\mathbf{Y}}^{i^{\prime }}\right) _{\perp }\doteqdot \mathbf{e}_{\mathbf{Y}%
}^{i^{\prime }}-g_{\parallel }\;\mathbf{e}_{\mathbf{X}}^{i^{\prime }}$ where
$g_{\parallel }\doteqdot g(h\mathbf{Y,}h\mathbf{X})$ is the tangential
h--part of the flow d--vector.

A similar encoding holds for v--variables $\nu ^{\widehat{a^{\prime }}%
},\varpi ^{\widehat{a^{\prime }}},h^{\widehat{a^{\prime }}}$ in the top row
of the vertical canonical d--connection matrices \ $\mathbf{\Gamma }_{v%
\mathbf{X\,}a^{\prime }}^{\qquad b^{\prime }}$ and $\mathbf{\Gamma }_{v%
\mathbf{Y\,}a^{\prime }}^{\qquad b^{\prime }}$ and in the row matrix $\left(
\mathbf{e}_{\mathbf{Y}}^{a^{\prime }}\right) _{\perp }\doteqdot \mathbf{e}_{%
\mathbf{Y}}^{a^{\prime }}-h_{\parallel }\;\mathbf{e}_{\mathbf{X}}^{a^{\prime
}}$ where $h_{\parallel }\doteqdot h(v\mathbf{Y,}v\mathbf{X})$ is the
tangential v--part of the flow d--vector. In a compact form of notations, we
shall write $\mathbf{v}^{\alpha ^{\prime }}$ and $\mathbf{\varpi }^{\alpha
^{\prime }}$ where the primed small Greek indices $\alpha ^{\prime },\beta
^{\prime },...$ will denote both N--adapted and then orthonormalized
components of geometric objects (like d--vectors, d--covectors, d--tensors,
d--groups, d--algebras, d--matrices) admitting further decompositions into
h-- and v--components defined as nonintegrable distributions of such objects.

With respect to N--adapted orthonormalized frames, the geometry of
N--anholonomic manifolds is defined algebraically, on their tangent bundles,
by couples of horizontal and vertical Klein geometries considered in \cite%
{sharpe} and for bi--Hamiltonian soliton constructions in \cite{anc1}. The
N--connection structure induces a N--anholonomic Klein space stated by two
left--invariant $h\mathfrak{g}$-- and $v\mathfrak{g}$--valued Maurer--Cartan
form on the Lie d--group $\mathbf{G}=(h\mathbf{G},v\mathbf{G})$ is
identified with the zero--curvature canonical d--connection 1--form $\;^{%
\mathbf{G}}\mathbf{\Gamma }=\{\;^{\mathbf{G}}\mathbf{\Gamma }_{\ \beta
^{\prime }}^{\alpha ^{\prime }}\},$ where $\;^{\mathbf{G}}\mathbf{\Gamma }%
_{\ \beta ^{\prime }}^{\alpha ^{\prime }}=\;^{\mathbf{G}}\mathbf{\Gamma }_{\
\beta ^{\prime }\gamma ^{\prime }}^{\alpha ^{\prime }}\mathbf{e}^{\gamma
^{\prime }}=\;^{h\mathbf{G}}L_{\;j^{\prime }k^{\prime }}^{i^{\prime }}%
\mathbf{e}^{k^{\prime }}+\;^{v\mathbf{G}}C_{\;j^{\prime }k^{\prime
}}^{i^{\prime }}e^{k^{\prime }}.$ For trivial N--connection structure in
vector bundles with the base and typical fiber spaces being symmetric
Riemannian spaces, we can consider that $\;^{h\mathbf{G}}L_{\;j^{\prime
}k^{\prime }}^{i^{\prime }}$ and $\;^{v\mathbf{G}}C_{\;j^{\prime }k^{\prime
}}^{i^{\prime }}$ are the coefficients of the Cartan connections $\;^{h%
\mathbf{G}}L$ and $\;^{v\mathbf{G}}C,$ respectively for the $h\mathbf{G}$
and $v\mathbf{G,}$ both with vanishing curvatures, i.e. with $d\;^{\mathbf{G}%
}\mathbf{\Gamma +}\frac{1}{2}\mathbf{[\;^{\mathbf{G}}\mathbf{\Gamma ,}\;^{%
\mathbf{G}}\mathbf{\Gamma }]=0}$and h-- and v--components, $d\;^{h\mathbf{G}}%
\mathbf{L}+\frac{1}{2}\mathbf{[\;^{h\mathbf{G}}L,\;^{h\mathbf{G}}L]}=0$ and $%
d\;^{v\mathbf{G}}\mathbf{C}+\frac{1}{2}[\;^{v\mathbf{G}}\mathbf{C},$ $\;^{v%
\mathbf{G}}\mathbf{C}]=0,$ where $d$ denotes the total derivatives on the
d--group manifold $\mathbf{G}=h\mathbf{G}\oplus v\mathbf{G}$ or their
restrictions on $h\mathbf{G}$ or $v\mathbf{G.}$ We can consider that $\;^{%
\mathbf{G}}\mathbf{\Gamma }$ defines the so--called Cartan d--connection for
nonintegrable N--connection structures, see details and noncommutative
developments in \cite{vncg,vsgg}.

Through the Lie d--algebra decompositions $\mathfrak{g}=h\mathfrak{g}\oplus v%
\mathfrak{g,}$ for the horizontal splitting: $h\mathfrak{g}=\mathfrak{so}%
(n)\oplus h\mathfrak{p,}$ when $[h\mathfrak{p},h\mathfrak{p}]\subset
\mathfrak{so}(n)$ and $[\mathfrak{so}(n),h\mathfrak{p}]\subset h\mathfrak{p;}
$ for the vertical splitting $v\mathfrak{g}=\mathfrak{so}(m)\oplus v%
\mathfrak{p,}$ when $[v\mathfrak{p},v\mathfrak{p}]\subset \mathfrak{so}(m)$
and $[\mathfrak{so}(m),v\mathfrak{p}]\subset v\mathfrak{p,}$ the Cartan
d--connection determines an N--anholonomic Riemannian structure on the
nonholonomic bundle $\mathbf{\mathring{E}}=[hG=SO(n+1),$ $%
vG=SO(m+1),\;N_{i}^{e}].$ For $n=m,$ and canonical d--objects
(N--connection, d--metric, d--connection, ...) derived from (\ref{m1b}), or
any N--anholonomic space with constant d--curvatures, the Cartan
d--connection transform just in the canonical d--connection (\ref{candcontm}%
). It is possible to consider a quotient space with distinguished structure
group $\mathbf{V}_{\mathbf{N}}=\mathbf{G}/SO(n)\oplus $ $SO(m)$ regarding $%
\mathbf{G}$ as a principal $\left( SO(n)\oplus SO(m)\right) $--bundle over $%
\mathbf{\mathring{E}},$ which is a N--anholonomic bundle. In this case, we
can always fix a local section of this bundle and pull--back $\;^{\mathbf{G}}%
\mathbf{\Gamma }$ to give a $\left( h\mathfrak{g}\oplus v\mathfrak{g}\right)
$--valued 1--form $^{\mathfrak{g}}\mathbf{\Gamma }$ in a point $u\in \mathbf{%
\mathring{E}}.$ Any change of local sections define $SO(n)\oplus $ $SO(m)$
gauge transforms of the canonical d--connection $^{\mathfrak{g}}\mathbf{%
\Gamma ,}$ all preserving the nonholonomic decomposition (\ref{whitney}).

There are involutive automorphisms $h\sigma =\pm 1$ and $v\sigma =\pm 1,$
respectively, of $h\mathfrak{g}$ and $v\mathfrak{g,}$ defined that $%
\mathfrak{so}(n)$ (or $\mathfrak{so}(m)$) is eigenspace $h\sigma =+1$ (or $%
v\sigma =+1)$ and $h\mathfrak{p}$ (or $v\mathfrak{p}$) is eigenspace $%
h\sigma =-1$ (or $v\sigma =-1).$ It is possible both a N--adapted
decomposition and taking into account the existing eigenspaces, when the
symmetric parts $\mathbf{\Gamma \doteqdot }\frac{1}{2}\left( ^{\mathfrak{g}}%
\mathbf{\Gamma +}\sigma \left( ^{\mathfrak{g}}\mathbf{\Gamma }\right)
\right) ,$ with respective h- and v--splitting $\mathbf{L\doteqdot }\frac{1}{%
2}\left( ^{h\mathfrak{g}}\mathbf{L+}h\sigma \left( ^{h\mathfrak{g}}\mathbf{L}%
\right) \right) $ and $\mathbf{C\doteqdot }\frac{1}{2}(^{v\mathfrak{g}}%
\mathbf{C}+h\sigma (^{v\mathfrak{g}}\mathbf{C})),$ defines a $\left(
\mathfrak{so}(n)\oplus \mathfrak{so}(m)\right) $--valued d--connection
1--form. Under such conditions, the antisymmetric part $\mathbf{e\doteqdot }%
\frac{1}{2}\left( ^{\mathfrak{g}}\mathbf{\Gamma -}\sigma \left( ^{\mathfrak{g%
}}\mathbf{\Gamma }\right) \right) ,$ with respective h- and v--splitting $h%
\mathbf{e\doteqdot }\frac{1}{2}\left( ^{h\mathfrak{g}}\mathbf{L-}h\sigma
\left( ^{h\mathfrak{g}}\mathbf{L}\right) \right) $ and $v\mathbf{e\doteqdot }%
\frac{1}{2}(^{v\mathfrak{g}}\mathbf{C}-h\sigma (^{v\mathfrak{g}}\mathbf{C}%
)), $ defines a $\left( h\mathfrak{p}\oplus v\mathfrak{p}\right) $--valued
N--adapted coframe for the Cartan--Killing inner product $<\cdot ,\cdot >_{%
\mathfrak{p}}$ on $T_{u}\mathbf{G}\simeq h\mathfrak{g}\oplus v\mathfrak{g}$
restricted to $T_{u}\mathbf{V}_{\mathbf{N}}\simeq \mathfrak{p.}$ This inner
product, distinguished into h- and v--components, provides a d--metric
structure of type $\mathbf{g}=[g,h]$ (\ref{m1}),where $g=<h\mathbf{e\otimes }%
h\mathbf{e}>_{h\mathfrak{p}}$ and $h=<v\mathbf{e\otimes }v\mathbf{e}>_{v%
\mathfrak{p}}$ on $\mathbf{V}_{\mathbf{N}}=\mathbf{G}/SO(n)\oplus $ $SO(m).$

We generate a $\mathbf{G(}=h\mathbf{G}\oplus v\mathbf{G)}$--invariant
d--derivative $\mathbf{D}$ whose restriction to the tangent space $T\mathbf{V%
}_{\mathbf{N}}$ for any N--anholonomic curve flow $\gamma (\tau ,\mathbf{l})$
in $\mathbf{V}_{\mathbf{N}}=\mathbf{G}/SO(n)\oplus $ $SO(m)$ is defined via%
\begin{equation}
\mathbf{D}_{\mathbf{X}}\mathbf{e=}\left[ \mathbf{e},\gamma _{\mathbf{l}%
}\rfloor \mathbf{\Gamma }\right] \mbox{\ and \ }\mathbf{D}_{\mathbf{Y}}%
\mathbf{e=}\left[ \mathbf{e},\gamma _{\mathbf{\tau }}\rfloor \mathbf{\Gamma }%
\right] ,  \label{aux33}
\end{equation}%
admitting further h- and v--decompositions. The derivatives $\mathbf{D}_{%
\mathbf{X}}$ and $\mathbf{D}_{\mathbf{Y}}$ are equivalent to those
considered in (\ref{part01}) and obey the Cartan structure equations (\ref%
{mtors}) and (\ref{mcurv}). For the canonical d--connections, a large class
of N--anholonomic spaces of dimension $n=m,$ the d--torsions are zero and
the d--curvatures are with constant coefficients.

Let $\mathbf{e}^{\alpha ^{\prime }}=(e^{i^{\prime }},\mathbf{e}^{a^{\prime
}})$ be a N--adapted orthonormalized coframe being identified with the $%
\left( h\mathfrak{p}\oplus v\mathfrak{p}\right) $--valued coframe $\mathbf{e}
$ in a fixed orthonormal basis for $\mathfrak{p=}h\mathfrak{p}\oplus v%
\mathfrak{p\subset }h\mathfrak{g}\oplus v\mathfrak{g.}$ Considering the
kernel/ cokernel of Lie algebra multiplications in the h- and v--subspaces,
respectively, $\left[ \mathbf{e}_{h\mathbf{X}},\cdot \right] _{h\mathfrak{g}%
} $ and $\left[ \mathbf{e}_{v\mathbf{X}},\cdot \right] _{v\mathfrak{g}},$ we
can decompose the coframes into parallel and perpendicular parts with
respect to $\mathbf{e}_{\mathbf{X}}.$ We write
\begin{equation*}
\mathbf{e=(e}_{C}=h\mathbf{e}_{C}+v\mathbf{e}_{C},\mathbf{e}_{C^{\perp }}=h%
\mathbf{e}_{C^{\perp }}+v\mathbf{e}_{C^{\perp }}\mathbf{),}
\end{equation*}%
for $\mathfrak{p(}=h\mathfrak{p}\oplus v\mathfrak{p)}$--valued mutually
orthogonal d--vectors $\mathbf{e}_{C}$ \ and $\mathbf{e}_{C^{\perp }},$ when
there are satisfied the conditions $\left[ \mathbf{e}_{\mathbf{X}},\mathbf{e}%
_{C}\right] _{\mathfrak{g}}=0$ but $\left[ \mathbf{e}_{\mathbf{X}},\mathbf{e}%
_{C^{\perp }}\right] _{\mathfrak{g}}\neq 0;$ such conditions can be stated
in h- and v--component form, respectively, $\left[ h\mathbf{e}_{\mathbf{X}},h%
\mathbf{e}_{C}\right] _{h\mathfrak{g}}=0,$ $\left[ h\mathbf{e}_{\mathbf{X}},h%
\mathbf{e}_{C^{\perp }}\right] _{h\mathfrak{g}}\neq 0$ and $\left[ v\mathbf{e%
}_{\mathbf{X}},v\mathbf{e}_{C}\right] _{v\mathfrak{g}}=0,$ $\left[ v\mathbf{e%
}_{\mathbf{X}},v\mathbf{e}_{C^{\perp }}\right] _{v\mathfrak{g}}\neq 0.$ One
holds also the algebraic decompositions
\begin{equation*}
T_{u}\mathbf{V}_{\mathbf{N}}\simeq \mathfrak{p=}h\mathfrak{p}\oplus v%
\mathfrak{p}=\mathfrak{g=}h\mathfrak{g}\oplus v\mathfrak{g}/\mathfrak{so}%
(n)\oplus \mathfrak{so}(m)
\end{equation*}%
and
\begin{equation*}
\mathfrak{p=p}_{C}\oplus \mathfrak{p}_{C^{\perp }}=\left( h\mathfrak{p}%
_{C}\oplus v\mathfrak{p}_{C}\right) \oplus \left( h\mathfrak{p}_{C^{\perp
}}\oplus v\mathfrak{p}_{C^{\perp }}\right) ,
\end{equation*}%
with $\mathfrak{p}_{\parallel }\subseteq \mathfrak{p}_{C}$ and $\mathfrak{p}%
_{C^{\perp }}\subseteq \mathfrak{p}_{\perp },$ where $\left[ \mathfrak{p}%
_{\parallel },\mathfrak{p}_{C}\right] =0,$ $<\mathfrak{p}_{C^{\perp }},%
\mathfrak{p}_{C}>=0,$ but $\left[ \mathfrak{p}_{\parallel },\mathfrak{p}%
_{C^{\perp }}\right] \neq 0$ (i.e. $\mathfrak{p}_{C}$ is the centralizer of $%
\mathbf{e}_{\mathbf{X}}$ in $\mathfrak{p=}h\mathfrak{p}\oplus v\mathfrak{%
p\subset }h\mathfrak{g}\oplus v\mathfrak{g);}$ in h- \ and v--components,
one have $h\mathfrak{p}_{\parallel }\subseteq h\mathfrak{p}_{C}$ and $h%
\mathfrak{p}_{C^{\perp }}\subseteq h\mathfrak{p}_{\perp },$ where $\left[ h%
\mathfrak{p}_{\parallel },h\mathfrak{p}_{C}\right] =0,$ $<h\mathfrak{p}%
_{C^{\perp }},h\mathfrak{p}_{C}>=0,$ but $\left[ h\mathfrak{p}_{\parallel },h%
\mathfrak{p}_{C^{\perp }}\right] \neq 0$ (i.e. $h\mathfrak{p}_{C}$ is the
centralizer of $\mathbf{e}_{h\mathbf{X}}$ in $h\mathfrak{p\subset }h%
\mathfrak{g)}$ and $v\mathfrak{p}_{\parallel }\subseteq v\mathfrak{p}_{C}$
and $v\mathfrak{p}_{C^{\perp }}\subseteq v\mathfrak{p}_{\perp },$ where $%
\left[ v\mathfrak{p}_{\parallel },v\mathfrak{p}_{C}\right] =0,$ $<v\mathfrak{%
p}_{C^{\perp }},v\mathfrak{p}_{C}>=0,$ but $\left[ v\mathfrak{p}_{\parallel
},v\mathfrak{p}_{C^{\perp }}\right] \neq 0$ (i.e. $v\mathfrak{p}_{C}$ is the
centralizer of $\mathbf{e}_{v\mathbf{X}}$ in $v\mathfrak{p\subset }v%
\mathfrak{g).}$ Using the canonical d--connection derivative $\mathbf{D}_{%
\mathbf{X}}$ of a d--covector perpendicular (or parallel) to $\mathbf{e}_{%
\mathbf{X}},$ we get a new d--vector which is parallel (or perpendicular) to
$\mathbf{e}_{\mathbf{X}},$ i.e. $\mathbf{D}_{\mathbf{X}}\mathbf{e}_{C}\in
\mathfrak{p}_{C^{\perp }}$ (or $\mathbf{D}_{\mathbf{X}}\mathbf{e}_{C^{\perp
}}\in \mathfrak{p}_{C});$ in h- \ and v--components such formulas are
written $\mathbf{D}_{h\mathbf{X}}h\mathbf{e}_{C}\in h\mathfrak{p}_{C^{\perp
}}$ (or $\mathbf{D}_{h\mathbf{X}}h\mathbf{e}_{C^{\perp }}\in h\mathfrak{p}%
_{C})$ and $\mathbf{D}_{v\mathbf{X}}v\mathbf{e}_{C}\in v\mathfrak{p}%
_{C^{\perp }}$ (or $\mathbf{D}_{v\mathbf{X}}v\mathbf{e}_{C^{\perp }}\in v%
\mathfrak{p}_{C}).$ All such d--algebraic relations can be written in
N--anholonomic manifolds and canonical d--connection settings, for instance,
using certain relations of type
\begin{equation*}
\mathbf{D}_{\mathbf{X}}(\mathbf{e}^{\alpha ^{\prime }})_{C}=\mathbf{v}%
_{~\beta ^{\prime }}^{\alpha ^{\prime }}(\mathbf{e}^{\beta ^{\prime
}})_{C^{\perp }}\mbox{ \ and \ }\mathbf{D}_{\mathbf{X}}(\mathbf{e}^{\alpha
^{\prime }})_{C^{\perp }}=-\mathbf{v}_{~\beta ^{\prime }}^{\alpha ^{\prime
}}(\mathbf{e}^{\beta ^{\prime }})_{C},
\end{equation*}%
for some antisymmetric d--tensors $\mathbf{v}^{\alpha ^{\prime }\beta
^{\prime }}=-\mathbf{v}^{\beta ^{\prime }\alpha ^{\prime }}.$ We get a
N--adapted $\left( SO(n)\oplus SO(m)\right) $--parallel frame defining a
generalization of the concept of Riemannian parallel frame on N--adapted
manifolds whenever $\mathfrak{p}_{C}$ is larger than $\mathfrak{p}%
_{\parallel }.$ Substituting $\mathbf{e}^{\alpha ^{\prime }}=(e^{i^{\prime
}},\mathbf{e}^{a^{\prime }})$ into the last formulas and considering h- and
v--components, we define $SO(n)$--parallel and $SO(m)$--parallel frames (for
simplicity we omit these formulas when the Greek small letter indices are
split into Latin small letter h- and v--indices).

The final conclusion of this section is that the Cartan structure equations
on hypersurfaces swept out by nonholonomic curve flows on N--anholonomic
spaces with constant matrix curvature for the canonical d--connection
geometrically encode \ two $O(n-1)$-- and $O(m-1)$--invariant, respectively,
horizontal and vertical bi--Hamiltonian operators. This holds true if the
distinguished by N--connection freedom of the d--group action $SO(n)\oplus
SO(m)$ on $\mathbf{e}$ and $\mathbf{\Gamma }$ is used to fix them to be a
N--adapted parallel coframe and its associated canonical d--connection
1--form is related to the canonical covariant derivative on N--anholonomic
manifolds.

\section{Bi--Hamiltonians and N--adapted Vector Solitons}

Introducing N--adapted orthonormalized bases, for N--anholonomic manifolds
of dimension $n+m,$ with constant curvatures of the canonical d--connection,
we can derive bi--Hamiltonian and vector soliton structures. In symbolic,
abstract index form, the constructions for nonholonomic vector bundles are
similar to those for the Riemannian symmetric--spaces soldered to Klein
geometry. We have to distinguish the horizontal and vertical components of
geometric objects and related equations.

\subsection{Basic equations for N--anholonomic curve flows}

In this section, we shall prove the results for the h--components of certain
N--anholonomic manifolds with constant d--curvature and then dub the
formulas for the v--components omitting similar details.

There is an isomorphism between the real space $\mathfrak{so}(n)$ and the
Lie algebra of $n\times n$ skew--symmetric matrices. This allows to
establish an isomorphism between $h\mathfrak{p}$ $\simeq \mathbb{R}^{n}$ and
the tangent spaces $T_{x}M=\mathfrak{so}(n+1)/$ $\mathfrak{so}(n)$ of the
Riemannian manifold $M=SO(n+1)/$ $SO(n)$ as described by the following
canonical decomposition
\begin{equation*}
h\mathfrak{g}=\mathfrak{so}(n+1)\supset h\mathfrak{p\in }\left[
\begin{array}{cc}
0 & h\mathbf{p} \\
-h\mathbf{p}^{T} & h\mathbf{0}%
\end{array}%
\right] \mbox{\ for\ }h\mathbf{0\in }h\mathfrak{h=so}(n)
\end{equation*}%
with $h\mathbf{p=\{}p^{i^{\prime }}\mathbf{\}\in }\mathbb{R}^{n}$ being the
h--component of the d--vector $\mathbf{p=(}p^{i^{\prime }}\mathbf{,}%
p^{a^{\prime }}\mathbf{)}$ and $h\mathbf{p}^{T}$ mean the transposition of
the row $h\mathbf{p.}$ The Cartan--Killing inner product on $h\mathfrak{g}$
is stated following the rule%
\begin{eqnarray*}
h\mathbf{p\cdot }h\mathbf{p} &\mathbf{=}&\left\langle \left[
\begin{array}{cc}
0 & h\mathbf{p} \\
-h\mathbf{p}^{T} & h\mathbf{0}%
\end{array}%
\right] ,\left[
\begin{array}{cc}
0 & h\mathbf{p} \\
-h\mathbf{p}^{T} & h\mathbf{0}%
\end{array}%
\right] \right\rangle \\
&\mathbf{\doteqdot }&\frac{1}{2}tr\left\{ \left[
\begin{array}{cc}
0 & h\mathbf{p} \\
-h\mathbf{p}^{T} & h\mathbf{0}%
\end{array}%
\right] ^{T}\left[
\begin{array}{cc}
0 & h\mathbf{p} \\
-h\mathbf{p}^{T} & h\mathbf{0}%
\end{array}%
\right] \right\} ,
\end{eqnarray*}%
where $tr$ denotes the trace of the corresponding product of matrices. This
product identifies canonically $h\mathfrak{p}$ $\simeq \mathbb{R}^{n}$ with
its dual $h\mathfrak{p}^{\ast }$ $\simeq \mathbb{R}^{n}.$ In a similar form,
we can consider
\begin{equation*}
v\mathfrak{g}=\mathfrak{so}(m+1)\supset v\mathfrak{p\in }\left[
\begin{array}{cc}
0 & v\mathbf{p} \\
-v\mathbf{p}^{T} & v\mathbf{0}%
\end{array}%
\right] \mbox{\ for\ }v\mathbf{0\in }v\mathfrak{h=so}(m)
\end{equation*}%
with $v\mathbf{p=\{}p^{a^{\prime }}\mathbf{\}\in }\mathbb{R}^{m}$ being the
v--component of the d--vector $\mathbf{p=(}p^{i^{\prime }}\mathbf{,}%
p^{a^{\prime }}\mathbf{)}$ and define the Cartan--Killing inner product $v%
\mathbf{p\cdot }v\mathbf{p\doteqdot }\frac{1}{2}tr\{...\}.$ In general, in
the tangent bundle of a N--anholonomic manifold, we can consider the
Cartan--Killing N--adapted inner product $\mathbf{p\cdot p=}h\mathbf{p\cdot }%
h\mathbf{p+}v\mathbf{p\cdot }v\mathbf{p.}$

Following the introduced Cartan--Killing parametrizations, we analyze the
flow $\gamma (\tau ,\mathbf{l})$ of a non--stretching curve in $\mathbf{V}_{%
\mathbf{N}}=\mathbf{G}/SO(n)\oplus $ $SO(m).$ Let us introduce a coframe $%
\mathbf{e}\in T_{\gamma }^{\ast }\mathbf{V}_{\mathbf{N}}\otimes (h\mathfrak{%
p\oplus }v\mathfrak{p}),$ which is a N--adapted $\left( SO(n)\mathfrak{%
\oplus }SO(m)\right) $--parallel basis along $\gamma ,$ and its associated
canonical d--con\-nec\-tion 1--form $\mathbf{\Gamma }\in T_{\gamma }^{\ast }%
\mathbf{V}_{\mathbf{N}}\otimes (\mathfrak{so}(n)\mathfrak{\oplus so}(m)).$
Such d--objects are respectively parametrized:%
\begin{equation*}
\mathbf{e}_{\mathbf{X}}=\mathbf{e}_{h\mathbf{X}}+\mathbf{e}_{v\mathbf{X}},
\end{equation*}%
where (for $(1,\overrightarrow{0})\in \mathbb{R}^{n},\overrightarrow{0}\in
\mathbb{R}^{n-1}$ and $(1,\overleftarrow{0})\in \mathbb{R}^{m},%
\overleftarrow{0}\in \mathbb{R}^{m-1},)$ $\mathbf{e}_{h\mathbf{X}}=\gamma _{h%
\mathbf{X}}\rfloor h\mathbf{e=}\left[
\begin{array}{cc}
0 & (1,\overrightarrow{0}) \\
-(1,\overrightarrow{0})^{T} & h\mathbf{0}%
\end{array}%
\right],$ $\mathbf{e}_{v\mathbf{X}}=\gamma _{v\mathbf{X}}\rfloor v\mathbf{e=}%
\left[
\begin{array}{cc}
0 & (1,\overleftarrow{0}) \\
-(1,\overleftarrow{0})^{T} & v\mathbf{0}%
\end{array}%
\right];$
\begin{equation*}
\mathbf{\Gamma =}\left[ \mathbf{\Gamma }_{h\mathbf{X}},\mathbf{\Gamma }_{v%
\mathbf{X}}\right] ,
\end{equation*}%
for
\begin{equation*}
\mathbf{\Gamma }_{h\mathbf{X}}\mathbf{=}\gamma _{h\mathbf{X}}\rfloor \mathbf{%
L=}\left[
\begin{array}{cc}
0 & (0,\overrightarrow{0}) \\
-(0,\overrightarrow{0})^{T} & \mathbf{L}%
\end{array}%
\right] \in \mathfrak{so}(n+1),
\end{equation*}%
where $\mathbf{L=}\left[
\begin{array}{cc}
0 & \overrightarrow{v} \\
-\overrightarrow{v}^{T} & h\mathbf{0}%
\end{array}%
\right] \in \mathfrak{so}(n),~\overrightarrow{v}\in \mathbb{R}^{n-1},~h%
\mathbf{0\in }\mathfrak{so}(n-1),$ and
\begin{equation*}
\mathbf{\Gamma }_{v\mathbf{X}}\mathbf{=}\gamma _{v\mathbf{X}}\rfloor \mathbf{%
C=}\left[
\begin{array}{cc}
0 & (0,\overleftarrow{0}) \\
-(0,\overleftarrow{0})^{T} & \mathbf{C}%
\end{array}%
\right] \in \mathfrak{so}(m+1),
\end{equation*}%
where $\mathbf{C=}\left[
\begin{array}{cc}
0 & \overleftarrow{v} \\
-\overleftarrow{v}^{T} & v\mathbf{0}%
\end{array}%
\right] \in \mathfrak{so}(m),~\overleftarrow{v}\in \mathbb{R}^{m-1},~v%
\mathbf{0\in }\mathfrak{so}(m-1).$

The above parametrizations are fixed in order to preserve the $SO(n)$ and $%
SO(m)$ rotation gauge freedoms on the N--adapted coframe and canonical
d--connec\-ti\-on 1--form, distinguished in h- and v--components.

There are defined decompositions of horizontal $SO(n+1)/$ $SO(n)$ matrices
like%
\begin{eqnarray*}
h\mathfrak{p} &\mathfrak{\ni }&\left[
\begin{array}{cc}
0 & h\mathbf{p} \\
-h\mathbf{p}^{T} & h\mathbf{0}%
\end{array}%
\right] =\left[
\begin{array}{cc}
0 & \left( h\mathbf{p}_{\parallel },\overrightarrow{0}\right) \\
-\left( h\mathbf{p}_{\parallel },\overrightarrow{0}\right) ^{T} & h\mathbf{0}%
\end{array}%
\right] \\
&&+\left[
\begin{array}{cc}
0 & \left( 0,h\overrightarrow{\mathbf{p}}_{\perp }\right) \\
-\left( 0,h\overrightarrow{\mathbf{p}}_{\perp }\right) ^{T} & h\mathbf{0}%
\end{array}%
\right] ,
\end{eqnarray*}%
into tangential and normal parts relative to $\mathbf{e}_{h\mathbf{X}}$ via
corresponding decompositions of h--vectors $h\mathbf{p=(}h\mathbf{\mathbf{p}%
_{\parallel },}h\mathbf{\overrightarrow{\mathbf{p}}_{\perp })\in }\mathbb{R}%
^{n}$ relative to $\left( 1,\overrightarrow{0}\right) ,$ when $h\mathbf{%
\mathbf{p}_{\parallel }}$ is identified with $h\mathfrak{p}_{C}$ and $h%
\mathbf{\overrightarrow{\mathbf{p}}_{\perp }}$ is identified with $h%
\mathfrak{p}_{\perp }=h\mathfrak{p}_{C^{\perp }}.$ In a similar form, it is
possible to decompose vertical $SO(m+1)/$ $SO(m)$ matrices,
\begin{eqnarray*}
v\mathfrak{p} &\mathfrak{\ni }&\left[
\begin{array}{cc}
0 & v\mathbf{p} \\
-v\mathbf{p}^{T} & v\mathbf{0}%
\end{array}%
\right] =\left[
\begin{array}{cc}
0 & \left( v\mathbf{p}_{\parallel },\overleftarrow{0}\right) \\
-\left( v\mathbf{p}_{\parallel },\overleftarrow{0}\right) ^{T} & v\mathbf{0}%
\end{array}%
\right] \\
&&+\left[
\begin{array}{cc}
0 & \left( 0,v\overleftarrow{\mathbf{p}}_{\perp }\right) \\
-\left( 0,v\overleftarrow{\mathbf{p}}_{\perp }\right) ^{T} & v\mathbf{0}%
\end{array}%
\right] ,
\end{eqnarray*}%
into tangential and normal parts relative to $\mathbf{e}_{v\mathbf{X}}$ via
corresponding decompositions of h--vectors $v\mathbf{p=(}v\mathbf{\mathbf{p}%
_{\parallel },}v\overleftarrow{\mathbf{\mathbf{p}}}\mathbf{_{\perp })\in }%
\mathbb{R}^{m}$ relative to $\left( 1,\overleftarrow{0}\right) ,$ when $v%
\mathbf{\mathbf{p}_{\parallel }}$ is identified with $v\mathfrak{p}_{C}$ and
$v\overleftarrow{\mathbf{\mathbf{p}}}\mathbf{_{\perp }}$ is identified with $%
v\mathfrak{p}_{\perp }=v\mathfrak{p}_{C^{\perp }}.$

The canonical d--connection induces matrices decomposed with respect to the
flow direction. In the h--direction, we parametrize%
\begin{equation*}
\mathbf{e}_{h\mathbf{Y}}=\gamma _{\tau }\rfloor h\mathbf{e=}\left[
\begin{array}{cc}
0 & \left( h\mathbf{e}_{\parallel },h\overrightarrow{\mathbf{e}}_{\perp
}\right) \\
-\left( h\mathbf{e}_{\parallel },h\overrightarrow{\mathbf{e}}_{\perp
}\right) ^{T} & h\mathbf{0}%
\end{array}%
\right] ,
\end{equation*}%
when $\mathbf{e}_{h\mathbf{Y}}\in h\mathfrak{p,}\left( h\mathbf{e}%
_{\parallel },h\overrightarrow{\mathbf{e}}_{\perp }\right) \in \mathbb{R}%
^{n} $ and $h\overrightarrow{\mathbf{e}}_{\perp }\in \mathbb{R}^{n-1},$ and
\begin{equation}
\mathbf{\Gamma }_{h\mathbf{Y}}\mathbf{=}\gamma _{h\mathbf{Y}}\rfloor \mathbf{%
L=}\left[
\begin{array}{cc}
0 & (0,\overrightarrow{0}) \\
-(0,\overrightarrow{0})^{T} & h\mathbf{\varpi }_{\tau }%
\end{array}%
\right] \in \mathfrak{so}(n+1),  \label{auxaaa}
\end{equation}%
where $\ \ h\mathbf{\varpi }_{\tau }\mathbf{=}\left[
\begin{array}{cc}
0 & \overrightarrow{\varpi } \\
-\overrightarrow{\varpi }^{T} & h\mathbf{\Theta }%
\end{array}%
\right] \in \mathfrak{so}(n),~\overrightarrow{\varpi }\in \mathbb{R}^{n-1},~h%
\mathbf{\Theta \in }\mathfrak{so}(n-1).$

In the v--direction, we parametrize
\begin{equation*}
\mathbf{e}_{v\mathbf{Y}}=\gamma _{\tau }\rfloor v\mathbf{e=}\left[
\begin{array}{cc}
0 & \left( v\mathbf{e}_{\parallel },v\overleftarrow{\mathbf{e}}_{\perp
}\right) \\
-\left( v\mathbf{e}_{\parallel },v\overleftarrow{\mathbf{e}}_{\perp }\right)
^{T} & v\mathbf{0}%
\end{array}%
\right] ,
\end{equation*}%
when $\mathbf{e}_{v\mathbf{Y}}\in v\mathfrak{p,}\left( v\mathbf{e}%
_{\parallel },v\overleftarrow{\mathbf{e}}_{\perp }\right) \in \mathbb{R}^{m}$
and $v\overleftarrow{\mathbf{e}}_{\perp }\in \mathbb{R}^{m-1},$ and
\begin{equation*}
\mathbf{\Gamma }_{v\mathbf{Y}}\mathbf{=}\gamma _{v\mathbf{Y}}\rfloor \mathbf{%
C=}\left[
\begin{array}{cc}
0 & (0,\overleftarrow{0}) \\
-(0,\overleftarrow{0})^{T} & v\mathbf{\varpi }_{\tau }%
\end{array}%
\right] \in \mathfrak{so}(m+1),
\end{equation*}%
where $v\mathbf{\varpi }_{\tau }\mathbf{=}\left[
\begin{array}{cc}
0 & \overleftarrow{\varpi } \\
-\overleftarrow{\varpi }^{T} & v\mathbf{\Theta }%
\end{array}%
\right] \in \mathfrak{so}(m),~\overleftarrow{\varpi }\in \mathbb{R}^{m-1},~v%
\mathbf{\Theta \in }\mathfrak{so}(m-1).$

The components $h\mathbf{e}_{\parallel }$ and $h\overrightarrow{\mathbf{e}}%
_{\perp }$ correspond to the decomposition
\begin{equation*}
\mathbf{e}_{h\mathbf{Y}}=h\mathbf{g(\gamma }_{\tau },\mathbf{\gamma }_{%
\mathbf{l}}\mathbf{)e}_{h\mathbf{X}}+\mathbf{(\gamma }_{\tau })_{\perp
}\rfloor h\mathbf{e}_{\perp }
\end{equation*}%
into tangential and normal parts relative to $\mathbf{e}_{h\mathbf{X}}.$ In
a similar form, one considers $v\mathbf{e}_{\parallel }$ and $v%
\overleftarrow{\mathbf{e}}_{\perp }$ corresponding to the decomposition
\begin{equation*}
\mathbf{e}_{v\mathbf{Y}}=v\mathbf{g(\gamma }_{\tau },\mathbf{\gamma }_{%
\mathbf{l}}\mathbf{)e}_{v\mathbf{X}}+\mathbf{(\gamma }_{\tau })_{\perp
}\rfloor v\mathbf{e}_{\perp }.
\end{equation*}

Using the above stated matrix parametrizations, we get%
\begin{eqnarray}
\left[ \mathbf{e}_{h\mathbf{X}},\mathbf{e}_{h\mathbf{Y}}\right] &=&-\left[
\begin{array}{cc}
0 & 0 \\
0 & h\mathbf{e}_{\perp }%
\end{array}%
\right] \in \mathfrak{so}(n+1),  \label{aux41} \\
\mbox{ \ for \ }h\mathbf{e}_{\perp } &=&\left[
\begin{array}{cc}
0 & h\overrightarrow{\mathbf{e}}_{\perp } \\
-(h\overrightarrow{\mathbf{e}}_{\perp })^{T} & h\mathbf{0}%
\end{array}%
\right] \in \mathfrak{so}(n);  \notag \\
\left[ \mathbf{\Gamma }_{h\mathbf{Y}},\mathbf{e}_{h\mathbf{Y}}\right] &=&-%
\left[
\begin{array}{cc}
0 & \left( 0,\overrightarrow{\varpi }\right) \\
-\left( 0,\overrightarrow{\varpi }\right) ^{T} & 0%
\end{array}%
\right] \in h\mathfrak{p}_{\perp };  \notag \\
\left[ \mathbf{\Gamma }_{h\mathbf{X}},\mathbf{e}_{h\mathbf{Y}}\right] &=&-%
\left[
\begin{array}{cc}
0 & \left( -\overrightarrow{v}\cdot h\overrightarrow{\mathbf{e}}_{\perp },h%
\mathbf{e}_{\parallel }\overrightarrow{v}\right) \\
-\left( -\overrightarrow{v}\cdot h\overrightarrow{\mathbf{e}}_{\perp },h%
\mathbf{e}_{\parallel }\overrightarrow{v}\right) ^{T} & h\mathbf{0}%
\end{array}%
\right] \in h\mathfrak{p};  \notag
\end{eqnarray}%
and
\begin{eqnarray}
\left[ \mathbf{e}_{v\mathbf{X}},\mathbf{e}_{v\mathbf{Y}}\right] &=&-\left[
\begin{array}{cc}
0 & 0 \\
0 & v\mathbf{e}_{\perp }%
\end{array}%
\right] \in \mathfrak{so}(m+1),  \label{aux41a} \\
\mbox{ \ for \ }v\mathbf{e}_{\perp } &=&\left[
\begin{array}{cc}
0 & v\overrightarrow{\mathbf{e}}_{\perp } \\
-(v\overrightarrow{\mathbf{e}}_{\perp })^{T} & v\mathbf{0}%
\end{array}%
\right] \in \mathfrak{so}(m);  \notag \\
\left[ \mathbf{\Gamma }_{v\mathbf{Y}},\mathbf{e}_{v\mathbf{Y}}\right] &=&-%
\left[
\begin{array}{cc}
0 & \left( 0,\overleftarrow{\varpi }\right) \\
-\left( 0,\overleftarrow{\varpi }\right) ^{T} & 0%
\end{array}%
\right] \in v\mathfrak{p}_{\perp };  \notag \\
\left[ \mathbf{\Gamma }_{v\mathbf{X}},\mathbf{e}_{v\mathbf{Y}}\right] &=&-%
\left[
\begin{array}{cc}
0 & \left( -\overleftarrow{v}\cdot v\overleftarrow{\mathbf{e}}_{\perp },v%
\mathbf{e}_{\parallel }\overleftarrow{v}\right) \\
-\left( -\overleftarrow{v}\cdot v\overleftarrow{\mathbf{e}}_{\perp },v%
\mathbf{e}_{\parallel }\overleftarrow{v}\right) ^{T} & v\mathbf{0}%
\end{array}%
\right] \in v\mathfrak{p}.  \notag
\end{eqnarray}

We can use formulas (\ref{aux41}) and (\ref{aux41a}) in order to write the
structure equations (\ref{mtors}) and (\ref{mcurv}) in terms of N--adapted
curve flow operators soldered to the geometry Klein N--anholonomic spaces
using the relations (\ref{aux33}). One obtains respectively the $\mathbf{G}$%
--invariant N--adapted torsion and curvature generated by the canonical
d--connection,
\begin{equation}
\mathbf{T}(\gamma _{\tau },\gamma _{\mathbf{l}})=\left( \mathbf{D}_{\mathbf{X%
}}\gamma _{\tau }-\mathbf{D}_{\mathbf{Y}}\gamma _{\mathbf{l}}\right) \rfloor
\mathbf{e=D}_{\mathbf{X}}\mathbf{e}_{\mathbf{Y}}-\mathbf{D}_{\mathbf{Y}}%
\mathbf{e}_{\mathbf{X}}+\left[ \mathbf{\Gamma }_{\mathbf{X}},\mathbf{e}_{%
\mathbf{Y}}\right] -\left[ \mathbf{\Gamma }_{\mathbf{Y}},\mathbf{e}_{\mathbf{%
X}}\right]  \label{torscf}
\end{equation}%
and
\begin{equation}
\mathbf{R}(\gamma _{\tau },\gamma _{\mathbf{l}})\mathbf{e=}\left[ \mathbf{D}%
_{\mathbf{X}},\mathbf{D}_{\mathbf{Y}}\right] \mathbf{e=D}_{\mathbf{X}}%
\mathbf{\Gamma }_{\mathbf{Y}}-\mathbf{D}_{\mathbf{Y}}\mathbf{\Gamma }_{%
\mathbf{X}}+\left[ \mathbf{\Gamma }_{\mathbf{X}},\mathbf{\Gamma }_{\mathbf{Y}%
}\right]  \label{curvcf}
\end{equation}%
where $\mathbf{e}_{\mathbf{X}}\doteqdot \gamma _{\mathbf{l}}\rfloor \mathbf{%
e,}$ $\mathbf{e}_{\mathbf{Y}}\doteqdot \gamma _{\mathbf{\tau }}\rfloor
\mathbf{e,}$ $\mathbf{\Gamma }_{\mathbf{X}}\doteqdot \gamma _{\mathbf{l}%
}\rfloor \mathbf{\Gamma }$ and $\mathbf{\Gamma }_{\mathbf{Y}}\doteqdot
\gamma _{\mathbf{\tau }}\rfloor \mathbf{\Gamma .}$ The formulas (\ref{torscf}%
) and (\ref{curvcf}) are equivalent, respectively, to (\ref{dtors}) and (\ref%
{dcurv}). In general, $\mathbf{T}(\gamma _{\tau },\gamma _{\mathbf{l}})\neq
0 $ and $\mathbf{R}(\gamma _{\tau },\gamma _{\mathbf{l}})\mathbf{e}$ can not
be defined to have constant matrix coefficients with respect to a N--adapted
basis. For N--anholonomic spaces with dimensions $n=m,$ we have $\ $a more
special d--connections also with possible constant, or vanishing,
d--curvature and d--torsion coefficients (see discussions related to
formulas (\ref{candcontm})). For such cases, we can consider the h-- and
v--components of (\ref{torscf}) and (\ref{curvcf}) in a similar manner as
for symmetric Riemannian spaces but for a d--connection $\mathbf{D}$ instead
of the Levi--Civita one $\nabla .$ We obtain
\begin{eqnarray}
0 &=&\left( \mathbf{D}_{h\mathbf{X}}\gamma _{\tau }-\mathbf{D}_{h\mathbf{Y}%
}\gamma _{\mathbf{l}}\right) \rfloor h\mathbf{e}  \label{torseq} \\
&\mathbf{=}&\mathbf{D}_{h\mathbf{X}}\mathbf{e}_{h\mathbf{Y}}-\mathbf{D}_{h%
\mathbf{Y}}\mathbf{e}_{h\mathbf{X}}+\left[ \mathbf{L}_{h\mathbf{X}},\mathbf{e%
}_{h\mathbf{Y}}\right] -\left[ \mathbf{L}_{h\mathbf{Y}},\mathbf{e}_{h\mathbf{%
X}}\right] ;  \notag \\
0 &=&\left( \mathbf{D}_{v\mathbf{X}}\gamma _{\tau }-\mathbf{D}_{v\mathbf{Y}%
}\gamma _{\mathbf{l}}\right) \rfloor v\mathbf{e}  \notag \\
&\mathbf{=}&\mathbf{D}_{v\mathbf{X}}\mathbf{e}_{v\mathbf{Y}}-\mathbf{D}_{v%
\mathbf{Y}}\mathbf{e}_{v\mathbf{X}}+\left[ \mathbf{C}_{v\mathbf{X}},\mathbf{e%
}_{v\mathbf{Y}}\right] -\left[ \mathbf{C}_{v\mathbf{Y}},\mathbf{e}_{v\mathbf{%
X}}\right] ,  \notag \\
h\mathbf{R}(\gamma _{\tau },\gamma _{\mathbf{l}})h\mathbf{e} &\mathbf{=}&%
\left[ \mathbf{D}_{h\mathbf{X}},\mathbf{D}_{h\mathbf{Y}}\right] h\mathbf{e=D}%
_{h\mathbf{X}}\mathbf{L}_{h\mathbf{Y}}-\mathbf{D}_{h\mathbf{Y}}\mathbf{L}_{h%
\mathbf{X}}+\left[ \mathbf{L}_{h\mathbf{X}},\mathbf{L}_{h\mathbf{Y}}\right]
\notag \\
v\mathbf{R}(\gamma _{\tau },\gamma _{\mathbf{l}})v\mathbf{e} &\mathbf{=}&%
\left[ \mathbf{D}_{v\mathbf{X}},\mathbf{D}_{v\mathbf{Y}}\right] v\mathbf{e=D}%
_{v\mathbf{X}}\mathbf{C}_{v\mathbf{Y}}-\mathbf{D}_{v\mathbf{Y}}\mathbf{C}_{v%
\mathbf{X}}+\left[ \mathbf{C}_{v\mathbf{X}},\mathbf{C}_{v\mathbf{Y}}\right] .
\notag
\end{eqnarray}

Following N--adapted curve flow parametrizations (\ref{aux41}) and (\ref%
{aux41a}), the equations (\ref{torseq}) are written
\begin{eqnarray}
&&0 =\mathbf{D}_{h\mathbf{X}}h\mathbf{e}_{\parallel }+\overrightarrow{v}%
\cdot h\overrightarrow{\mathbf{e}}_{\perp },~0=\mathbf{D}_{v\mathbf{X}}v%
\mathbf{e}_{\parallel }+\overleftarrow{v}\cdot v\overleftarrow{\mathbf{e}}%
_{\perp};  \label{torseqd} \\
&&0=\overrightarrow{\varpi }-h\mathbf{e}_{\parallel }\overrightarrow{v}+%
\mathbf{D}_{h\mathbf{X}}h\overrightarrow{\mathbf{e}}_{\perp },~0=%
\overleftarrow{\varpi }-v\mathbf{e}_{\parallel }\overleftarrow{v}+\mathbf{D}%
_{v\mathbf{X}}v\overleftarrow{\mathbf{e}}_{\perp };  \notag \\
&&\mathbf{D}_{h\mathbf{X}}\overrightarrow{\varpi }-\mathbf{D}_{h\mathbf{Y}}%
\overrightarrow{v}+\overrightarrow{v}\rfloor h\mathbf{\Theta } = h%
\overrightarrow{\mathbf{e}}_{\perp },~\mathbf{D}_{v\mathbf{X}}\overleftarrow{%
\varpi }-\mathbf{D}_{v\mathbf{Y}}\overleftarrow{v}+\overleftarrow{v}\rfloor v%
\mathbf{\Theta =}v\overleftarrow{\mathbf{e}}_{\perp };  \notag \\
&&\mathbf{D}_{h\mathbf{X}}h\mathbf{\Theta -}\overrightarrow{v}\otimes
\overrightarrow{\varpi }+\overrightarrow{\varpi }\otimes \overrightarrow{v}
=0,\ \mathbf{D}_{v\mathbf{X}}v\mathbf{\Theta -}\overleftarrow{v}\otimes
\overleftarrow{\varpi }+\overleftarrow{\varpi }\otimes \overleftarrow{v}=0.
\notag
\end{eqnarray}%
The tensor and interior products, for instance, for the h--components, are
defined in the form: $\otimes $ denotes the outer product of pairs of
vectors ($1\times n$ row matrices), producing $n\times n$ matrices $%
\overrightarrow{A}\otimes \overrightarrow{B}=\overrightarrow{A}^{T}%
\overrightarrow{B},$ and $\rfloor $ denotes multiplication of $n\times n$
matrices on vectors ($1\times n$ row matrices); one holds the properties $%
\overrightarrow{A}\rfloor \left( \overrightarrow{B}\otimes \overrightarrow{C}%
\right) =\left( \overrightarrow{A}\cdot \overrightarrow{B}\right)
\overrightarrow{C}$ which is the transpose of the standard matrix product on
column vectors, and $\left( \overrightarrow{B}\otimes \overrightarrow{C}%
\right) \overrightarrow{A}=\left( \overrightarrow{C}\cdot \overrightarrow{A}%
\right) \overrightarrow{B}.$ Here we note that similar formulas hold for the
v--components but, for instance, we have to change, correspondingly, $%
n\rightarrow m$ and $\overrightarrow{A}\rightarrow \overleftarrow{A}.$

The variables $\mathbf{e}_{\parallel }$ and $\mathbf{\Theta ,}$ written in
h-- and v--components, can be expressed correspondingly in terms of
variables $\overrightarrow{v},\overrightarrow{\varpi },h\overrightarrow{%
\mathbf{e}}_{\perp }$ and $\overleftarrow{v},\overleftarrow{\varpi },v%
\overleftarrow{\mathbf{e}}_{\perp }$ (see equations (\ref{torseqd})),%
\begin{equation*}
h\mathbf{e}_{\parallel }=-\mathbf{D}_{h\mathbf{X}}^{-1}(\overrightarrow{v}%
\cdot h\overrightarrow{\mathbf{e}}_{\perp }),~v\mathbf{e}_{\parallel }=-%
\mathbf{D}_{v\mathbf{X}}^{-1}(\overleftarrow{v}\cdot v\overleftarrow{\mathbf{%
e}}_{\perp }),
\end{equation*}%
and $\ h\mathbf{\Theta =D}_{h\mathbf{X}}^{-1}\left( \overrightarrow{v}%
\otimes \overrightarrow{\varpi }-\overrightarrow{\varpi }\otimes
\overrightarrow{v}\right) ,~v\mathbf{\Theta =D}_{v\mathbf{X}}^{-1}\left(
\overleftarrow{v}\otimes \overleftarrow{\varpi }-\overleftarrow{\varpi }%
\otimes \overleftarrow{v}\right) .$ Substituting these values, respectively,
in equations in (\ref{torseqd}), we express
\begin{equation*}
\overrightarrow{\varpi }=-\mathbf{D}_{h\mathbf{X}}h\overrightarrow{\mathbf{e}%
}_{\perp }-\mathbf{D}_{h\mathbf{X}}^{-1}(\overrightarrow{v}\cdot h%
\overrightarrow{\mathbf{e}}_{\perp })\overrightarrow{v},~\overleftarrow{%
\varpi }=-\mathbf{D}_{v\mathbf{X}}v\overleftarrow{\mathbf{e}}_{\perp }-%
\mathbf{D}_{v\mathbf{X}}^{-1}(\overleftarrow{v}\cdot v\overleftarrow{\mathbf{%
e}}_{\perp })\overleftarrow{v},
\end{equation*}%
contained in the h-- and v--flow equations respectively on $\overrightarrow{v%
}$ and $\overleftarrow{v},$ considered as scalar components when $\mathbf{D}%
_{h\mathbf{Y}}\overrightarrow{v}=\overrightarrow{v}_{\tau }$ and $\mathbf{D}%
_{h\mathbf{Y}}\overleftarrow{v}=\overleftarrow{v}_{\tau },$
\begin{eqnarray}
\overrightarrow{v}_{\tau } &=&\mathbf{D}_{h\mathbf{X}}\overrightarrow{\varpi
}-\overrightarrow{v}\rfloor \mathbf{D}_{h\mathbf{X}}^{-1}\left(
\overrightarrow{v}\otimes \overrightarrow{\varpi }-\overrightarrow{\varpi }%
\otimes \overrightarrow{v}\right) -\overrightarrow{R}h\overrightarrow{%
\mathbf{e}}_{\perp },  \label{floweq} \\
\overleftarrow{v}_{\tau } &=&\mathbf{D}_{v\mathbf{X}}\overleftarrow{\varpi }-%
\overleftarrow{v}\rfloor \mathbf{D}_{v\mathbf{X}}^{-1}\left( \overleftarrow{v%
}\otimes \overleftarrow{\varpi }-\overleftarrow{\varpi }\otimes
\overleftarrow{v}\right) -\overleftarrow{S}v\overleftarrow{\mathbf{e}}%
_{\perp },  \notag
\end{eqnarray}%
where the scalar curvatures of chosen d--connection, $\overrightarrow{R}$
and $\overleftarrow{S}$ are defined by formulas (\ref{sdccurv}) in Appendix.
For symmetric Riemannian spaces like $SO(n+1)/SO(n)\simeq S^{n},$ the value $%
\overrightarrow{R}$ is just the scalar curvature $\chi =1,$ \ see \cite{anc2}%
. On N--anholonomic (pseudo) Riemannian manifolds, it is possible to define
such d--connections that $\overrightarrow{R}$ and $\overleftarrow{S}$ are
certain zero or nonzero constants, see Corollary \ref{corol01}.

The above presented considerations consist the proof of

\begin{lemma}
On N--anholonomic (pseudo) Riemannian manifolds with constant curvature
matrix coefficients for a d--connection, there are N--adapted Hamiltonian
symplectic operators,
\begin{equation}
h\mathcal{J}=\mathbf{D}_{h\mathbf{X}}+\mathbf{D}_{h\mathbf{X}}^{-1}\left(
\overrightarrow{v}\cdot \right) \overrightarrow{v}\mbox{ \ and \ }v\mathcal{J%
}=\mathbf{D}_{v\mathbf{X}}+\mathbf{D}_{v\mathbf{X}}^{-1}\left(
\overleftarrow{v}\cdot \right) \overleftarrow{v},  \label{sop}
\end{equation}%
and cosymplectic operators%
\begin{equation}
h\mathcal{H}\doteqdot \mathbf{D}_{h\mathbf{X}}+\overrightarrow{v}\rfloor
\mathbf{D}_{h\mathbf{X}}^{-1}\left( \overrightarrow{v}\wedge \right)
\mbox{
\ and \ }v\mathcal{H}\doteqdot \mathbf{D}_{v\mathbf{X}}+\overleftarrow{v}%
\rfloor \mathbf{D}_{v\mathbf{X}}^{-1}\left( \overleftarrow{v}\wedge \right) ,
\label{csop}
\end{equation}%
where, for instance, $\overrightarrow{A}\wedge \overrightarrow{B}=%
\overrightarrow{A}\otimes \overrightarrow{B}-\overrightarrow{B}\otimes $ $%
\overrightarrow{A}.$\
\end{lemma}

The properties of operators (\ref{sop}) and (\ref{csop}) are defined by

\begin{theorem}
\label{mr1}The d--operators $\mathcal{J=}\left( h\mathcal{J},v\mathcal{J}%
\right) $ and $\mathcal{H=}\left( h\mathcal{H},v\mathcal{H}\right) $ $\ $are
respectively $\left( O(n-1),O(m-1)\right) $--invariant Hamiltonian
symplectic and cosymplectic d--operators with respect to the Hamiltonian
d--variables $\left( \overrightarrow{v},\overleftarrow{v}\right) .$ Such
d--operators defines the Hamiltonian form for the curve flow equations on
N--anholonomic (pseudo) Riemannian manifolds with constant d--connection
curvature: the h--flows are given by%
\begin{eqnarray}
\overrightarrow{v}_{\tau } &=&h\mathcal{H}\left( \overrightarrow{\varpi }%
\right) -\overrightarrow{R}~h\overrightarrow{\mathbf{e}}_{\perp }=h\mathfrak{%
R}\left( h\overrightarrow{\mathbf{e}}_{\perp }\right) -\overrightarrow{R}~h%
\overrightarrow{\mathbf{e}}_{\perp },  \notag \\
\overrightarrow{\varpi } &=&h\mathcal{J}\left( h\overrightarrow{\mathbf{e}}%
_{\perp }\right) ;  \label{hhfeq1}
\end{eqnarray}%
the v--flows are given by
\begin{eqnarray}
\overleftarrow{v}_{\tau } &=&v\mathcal{H}\left( \overleftarrow{\varpi }%
\right) -\overleftarrow{S}~v\overleftarrow{\mathbf{e}}_{\perp }=v\mathfrak{R}%
\left( v\overleftarrow{\mathbf{e}}_{\perp }\right) -\overleftarrow{S}~v%
\overleftarrow{\mathbf{e}}_{\perp },  \notag \\
\overleftarrow{\varpi } &=&v\mathcal{J}\left( v\overleftarrow{\mathbf{e}}%
_{\perp }\right) ,  \label{vhfeq1}
\end{eqnarray}%
where the so--called heriditary recursion d--operator has the respective h--
and v--components
\begin{equation}
h\mathfrak{R}=h\mathcal{H}\circ h\mathcal{J}\mbox{ \ and \ }v\mathfrak{R}=v%
\mathcal{H}\circ v\mathcal{J}.  \label{reqop}
\end{equation}
\end{theorem}

\begin{proof}
One follows from the Lemma and (\ref{floweq}). In a detailed form, for
holonomic structures, it is given in Ref. \cite{saw} and discussed in \cite%
{anc2}. The above considerations, in this section, provides a soldering of
certain classes of N--anholonomic (pseudo) Riemannian manifolds with $\left(
O(n-1),O(m-1)\right) $--gauge symmetry to the geometry of Klein
N--anholonomic spaces.$\square $
\end{proof}

\subsection{Bi--Hamiltonian curve flows and solitonic hierarchies}

Following a usual solitonic techniques, see details in Ref. \cite{anc1,anc2}%
, the recursion h--operator from (\ref{reqop}),%
\begin{eqnarray}
h\mathfrak{R} &=&\mathbf{D}_{h\mathbf{X}}\left( \mathbf{D}_{h\mathbf{X}}+%
\mathbf{D}_{h\mathbf{X}}^{-1}\left( \overrightarrow{v}\cdot \right)
\overrightarrow{v}\right) +\overrightarrow{v}\rfloor \mathbf{D}_{h\mathbf{X}%
}^{-1}\left( \overrightarrow{v}\wedge \mathbf{D}_{h\mathbf{X}}\right)
\label{reqoph} \\
&=&\mathbf{D}_{h\mathbf{X}}^{2}+|\mathbf{D}_{h\mathbf{X}}|^{2}+\mathbf{D}_{h%
\mathbf{X}}^{-1}\left( \overrightarrow{v}\cdot \right) \overrightarrow{v}_{%
\mathbf{l}}-\overrightarrow{v}\rfloor \mathbf{D}_{h\mathbf{X}}^{-1}(%
\overrightarrow{v}_{\mathbf{l}}\wedge ),  \notag
\end{eqnarray}%
generates a horizontal hierarchy of commuting Hamiltonian vector fields $h%
\overrightarrow{\mathbf{e}}_{\perp }^{(k)}$ starting from $h\overrightarrow{%
\mathbf{e}}_{\perp }^{(0)}=\overrightarrow{v}_{\mathbf{l}}$ given by the
infinitesimal generator of $\mathbf{l}$--translations in terms of arclength $%
\mathbf{l}$ along the curve (we use a boldface $\mathbf{l}$ in order to
emphasized that the curve is on a N--anholonomic manifold). A vertical
hierarchy of commuting vector fields $v\overleftarrow{\mathbf{e}}_{\perp
}^{(k)}$ starting from $v\overleftarrow{\mathbf{e}}_{\perp }^{(0)}$ $=%
\overleftarrow{v}_{\mathbf{l}}$ is generated by the recursion v--operator%
\begin{eqnarray}
v\mathfrak{R} &=&\mathbf{D}_{v\mathbf{X}}\left( \mathbf{D}_{v\mathbf{X}}+%
\mathbf{D}_{v\mathbf{X}}^{-1}\left( \overleftarrow{v}\cdot \right)
\overleftarrow{v}\right) +\overleftarrow{v}\rfloor \mathbf{D}_{v\mathbf{X}%
}^{-1}\left( \overleftarrow{v}\wedge \mathbf{D}_{v\mathbf{X}}\right)
\label{reqopv} \\
&=&\mathbf{D}_{v\mathbf{X}}^{2}+|\mathbf{D}_{v\mathbf{X}}|^{2}+\mathbf{D}_{v%
\mathbf{X}}^{-1}\left( \overleftarrow{v}\cdot \right) \overleftarrow{v}_{%
\mathbf{l}}-\overleftarrow{v}\rfloor \mathbf{D}_{v\mathbf{X}}^{-1}(%
\overleftarrow{v}_{\mathbf{l}}\wedge ).  \notag
\end{eqnarray}%
There are related hierarchies, generated by adjoint operators $\mathfrak{R}%
^{\ast }=(h\mathfrak{R}^{\ast },$ $v\mathfrak{R}^{\ast }),$ of involuntive
Hamiltonian h--covector fields $\overrightarrow{\varpi }^{(k)}=\delta \left(
hH^{(k)}\right) /\delta \overrightarrow{v}$ in terms of Hamiltonians $%
hH=hH^{(k)}(\overrightarrow{v},\overrightarrow{v}_{\mathbf{l}},%
\overrightarrow{v}_{2\mathbf{l}},...)$ starting from $\overrightarrow{\varpi
}^{(0)}=\overrightarrow{v},hH^{(0)}=\frac{1}{2}|\overrightarrow{v}|^{2}$ and
of involutive Hamiltonian v--covector fields $\overleftarrow{\varpi }%
^{(k)}=\delta \left( vH^{(k)}\right) /$ $\delta \overleftarrow{v}$ in terms
of Hamiltonians $vH=vH^{(k)}(\overleftarrow{v},\overleftarrow{v}_{\mathbf{l}%
},\overleftarrow{v}_{2\mathbf{l}},...)$ starting from $\overleftarrow{\varpi
}^{(0)}=\overleftarrow{v},vH^{(0)}=\frac{1}{2}|\overleftarrow{v}|^{2}.$ The
relations between hierarchies are established correspondingly by formulas%
\begin{eqnarray*}
h\overrightarrow{\mathbf{e}}_{\perp }^{(k)} &=&h\mathcal{H}\left(
\overrightarrow{\varpi }^{(k)},\overrightarrow{\varpi }^{(k+1)}\right) =h%
\mathcal{J}\left( h\overrightarrow{\mathbf{e}}_{\perp }^{(k)}\right) , \\
v\overleftarrow{\mathbf{e}}_{\perp }^{(k)} &=&v\mathcal{H}\left(
\overleftarrow{\varpi }^{(k)},\overleftarrow{\varpi }^{(k+1)}\right) =v%
\mathcal{J}\left( v\overleftarrow{\mathbf{e}}_{\perp }^{(k)}\right) ,
\end{eqnarray*}%
where $k=0,1,2,....$ All hierarchies (horizontal, vertical and their adjoint
ones) have a typical mKdV scaling symmetry, for instance, $\mathbf{%
l\rightarrow \lambda l}$ and $\overrightarrow{v}\rightarrow \mathbf{\lambda }%
^{-1}\overrightarrow{v}$ under which the values $h\overrightarrow{\mathbf{e}}%
_{\perp }^{(k)}$ and $hH^{(k)}$ have scaling weight $2+2k,$ while $%
\overrightarrow{\varpi }^{(k)}$ has scaling weight $1+2k.$

The above presented considerations prove

\begin{corollary}
\label{c2} There are N--adapted hierarchies of distinguished horizontal and
vertical commuting bi--Hamiltonian flows, correspondingly, on $%
\overrightarrow{v}$ and $\overleftarrow{v}$ associated to the recursion
d--operator (\ref{reqop}) given by $O(n-1)\oplus O(m-1)$ --invariant
d--vector evolution equations,%
\begin{eqnarray*}
\overrightarrow{v}_{\tau } &=&h\overrightarrow{\mathbf{e}}_{\perp }^{(k+1)}-%
\overrightarrow{R}~h\overrightarrow{\mathbf{e}}_{\perp }^{(k)}=h\mathcal{H}%
\left( \delta \left( hH^{(k,\overrightarrow{R})}\right) /\delta
\overrightarrow{v}\right) \\
&=&\left( h\mathcal{J}\right) ^{-1}\left( \delta \left( hH^{(k+1,%
\overrightarrow{R})}\right) /\delta \overrightarrow{v}\right)
\end{eqnarray*}%
with horizontal Hamiltonians $hH^{(k+1,\overrightarrow{R})}=hH^{(k+1,%
\overrightarrow{R})}-\overrightarrow{R}~hH^{(k,\overrightarrow{R})}$ and
\begin{eqnarray*}
\overleftarrow{v}_{\tau } &=&v\overleftarrow{\mathbf{e}}_{\perp }^{(k+1)}-%
\overleftarrow{S}~v\overleftarrow{\mathbf{e}}_{\perp }^{(k)}=v\mathcal{H}%
\left( \delta \left( vH^{(k,\overleftarrow{S})}\right) /\delta
\overleftarrow{v}\right) \\
&=&\left( v\mathcal{J}\right) ^{-1}\left( \delta \left( vH^{(k+1,%
\overleftarrow{S})}\right) /\delta \overleftarrow{v}\right)
\end{eqnarray*}%
with vertical Hamiltonians $vH^{(k+1,\overleftarrow{S})}=vH^{(k+1,%
\overleftarrow{S})}-\overleftarrow{S}~vH^{(k,\overleftarrow{S})},$ for $%
k=0,1,2,.....$ The d--operators $\mathcal{H}$ and $\mathcal{J}$ are
N--adapted and mutually compatible from which one can be constructed an
alternative (explicit) Hamilton d--operator $~^{a}\mathcal{H=H\circ J}$ $%
\circ \mathcal{H=}\mathfrak{R\circ }\mathcal{H}.$
\end{corollary}

\subsubsection{Formulation of the main theorem}

The main goal of this paper is to prove that the geometric data for any
(pseudo) Riemannian metric naturally define a N--adapted bi--Hamiltonian
flow hierarchy inducing anholonomic solitonic configurations.

\begin{theorem}
\label{mt} For any N--anholonomic (pseudo) Riemannian manifold with
prescribed d--metric structure, there is a hierarchy of bi-Hamiltonian
N--adapted flows of curves $\gamma (\tau ,\mathbf{l})=h\gamma (\tau ,\mathbf{%
l})+v\gamma (\tau ,\mathbf{l})$ described by geometric nonholonomic map
equations. The $0$ flows are defined as convective (travelling wave) maps%
\begin{equation}
\gamma _{\tau }=\gamma _{\mathbf{l}},\mbox{\ distinguished \ }\left( h\gamma
\right) _{\tau }=\left( h\gamma \right) _{h\mathbf{X}}\mbox{\ and \ }\left(
v\gamma \right) _{\tau }=\left( v\gamma \right) _{v\mathbf{X}}.
\label{trmap}
\end{equation}%
There are +1 flows defined as non--stretching mKdV maps%
\begin{eqnarray}
-\left( h\gamma \right) _{\tau } &=&\mathbf{D}_{h\mathbf{X}}^{2}\left(
h\gamma \right) _{h\mathbf{X}}+\frac{3}{2}\left| \mathbf{D}_{h\mathbf{X}%
}\left( h\gamma \right) _{h\mathbf{X}}\right| _{h\mathbf{g}}^{2}~\left(
h\gamma \right) _{h\mathbf{X}},  \label{1map} \\
-\left( v\gamma \right) _{\tau } &=&\mathbf{D}_{v\mathbf{X}}^{2}\left(
v\gamma \right) _{v\mathbf{X}}+\frac{3}{2}\left| \mathbf{D}_{v\mathbf{X}%
}\left( v\gamma \right) _{v\mathbf{X}}\right| _{v\mathbf{g}}^{2}~\left(
v\gamma \right) _{v\mathbf{X}},  \notag
\end{eqnarray}%
and the +2,... flows as higher order analogs. Finally, the -1 flows are
defined by the kernels of recursion operators (\ref{reqoph}) and (\ref%
{reqopv}) inducing non--stretching maps%
\begin{equation}
\mathbf{D}_{h\mathbf{Y}}\left( h\gamma \right) _{h\mathbf{X}}=0%
\mbox{\ and \
}\mathbf{D}_{v\mathbf{Y}}\left( v\gamma \right) _{v\mathbf{X}}=0.
\label{-1map}
\end{equation}
\end{theorem}

\begin{proof}
It is given in the next section \ref{ssp}.
\end{proof}

For similar constructions in gravity models with nontrivial torsion and
nonholonomic structure and related geometry of noncommutative spaces and
anholonomic spinors, it is important \cite{vncg,vsgg}.

\subsubsection{Proof of the main theorem}

\label{ssp}We provide a proof of Theorem \ref{mt} for the horizontal flows.
The approach is based on the method provided in Section 3 of Ref. \cite{anc1}
but in this work the Levi-- Civita connection on symmetric Riemannian spaces
is substituted by the horizontal components of a d--connection with constant
d--curvature coefficients. The vertical constructions are similar but with
respective changing of h-- variables / objects into v- variables/ objects.

One obtains a vector mKdV equation up to a convective term (can be absorbed
by redefinition of coordinates) defining the +1 flow for $h\overrightarrow{%
\mathbf{e}}_{\perp }=\overrightarrow{v}_{\mathbf{l}},$%
\begin{equation*}
\overrightarrow{v}_{\tau }=\overrightarrow{v}_{3\mathbf{l}}+\frac{3}{2}|%
\overrightarrow{v}|^{2}-\overrightarrow{R}~\overrightarrow{v}_{\mathbf{l}},
\end{equation*}%
when the $+(k+1)$ flow gives a vector mKdV equation of higher order $3+2k$
on $\overrightarrow{v}$ and there is a $0$ h--flow $\overrightarrow{v}_{\tau
}=\overrightarrow{v}_{\mathbf{l}}$ arising from $h\overrightarrow{\mathbf{e}}%
_{\perp }=0$ and $h\overrightarrow{\mathbf{e}}_{\parallel }=1$ belonging
outside the hierarchy generated by $h\mathfrak{R.}$ Such flows correspond to
N--adapted horizontal motions of the curve $\gamma (\tau ,\mathbf{l}%
)=h\gamma (\tau ,\mathbf{l})+v\gamma (\tau ,\mathbf{l}),$ given by
\begin{equation*}
\left( h\gamma \right) _{\tau }=f\left( \left( h\gamma \right) _{h\mathbf{X}%
},\mathbf{D}_{h\mathbf{X}}\left( h\gamma \right) _{h\mathbf{X}},\mathbf{D}_{h%
\mathbf{X}}^{2}\left( h\gamma \right) _{h\mathbf{X}},...\right)
\end{equation*}%
subject to the non--stretching condition $|\left( h\gamma \right) _{h\mathbf{%
X}}|_{h\mathbf{g}}=1,$ when the equation of motion is to be derived from the
identifications
\begin{equation*}
\left( h\gamma \right) _{\tau }\longleftrightarrow \mathbf{e}_{h\mathbf{Y}},%
\mathbf{D}_{h\mathbf{X}}\left( h\gamma \right) _{h\mathbf{X}%
}\longleftrightarrow \mathcal{D}_{h\mathbf{X}}\mathbf{e}_{h\mathbf{X}}=\left[
\mathbf{L}_{h\mathbf{X}},\mathbf{e}_{h\mathbf{X}}\right]
\end{equation*}%
and so on, which maps the constructions from the tangent space of the curve
to the space $h\mathfrak{p}.$ For such identifications, we have
\begin{eqnarray*}
\left[ \mathbf{L}_{h\mathbf{X}},\mathbf{e}_{h\mathbf{X}}\right] &=&-\left[
\begin{array}{cc}
0 & \left( 0,\overrightarrow{v}\right) \\
-\left( 0,\overrightarrow{v}\right) ^{T} & h\mathbf{0}%
\end{array}%
\right] \in h\mathfrak{p}, \\
\left[ \mathbf{L}_{h\mathbf{X}},\left[ \mathbf{L}_{h\mathbf{X}},\mathbf{e}_{h%
\mathbf{X}}\right] \right] &=&-\left[
\begin{array}{cc}
0 & \left( |\overrightarrow{v}|^{2},\overrightarrow{0}\right) \\
-\left( |\overrightarrow{v}|^{2},\overrightarrow{0}\right) ^{T} & h\mathbf{0}%
\end{array}%
\right]
\end{eqnarray*}%
and so on, see similar calculus in (\ref{aux41}). At the next step, stating
for the +1 h--flow
\begin{equation*}
h\overrightarrow{\mathbf{e}}_{\perp }=\overrightarrow{v}_{\mathbf{l}}%
\mbox{
and }h\overrightarrow{\mathbf{e}}_{\parallel }=-\mathbf{D}_{h\mathbf{X}%
}^{-1}\left( \overrightarrow{v}\cdot \overrightarrow{v}_{\mathbf{l}}\right)
=-\frac{1}{2}|\overrightarrow{v}|^{2},
\end{equation*}%
we compute
\begin{eqnarray*}
\mathbf{e}_{h\mathbf{Y}} &=&\left[
\begin{array}{cc}
0 & \left( h\mathbf{e}_{\parallel },h\overrightarrow{\mathbf{e}}_{\perp
}\right) \\
-\left( h\mathbf{e}_{\parallel },h\overrightarrow{\mathbf{e}}_{\perp
}\right) ^{T} & h\mathbf{0}%
\end{array}%
\right] \\
&=&-\frac{1}{2}|\overrightarrow{v}|^{2}\left[
\begin{array}{cc}
0 & \left( 1,\overrightarrow{\mathbf{0}}\right) \\
-\left( 0,\overrightarrow{\mathbf{0}}\right) ^{T} & h\mathbf{0}%
\end{array}%
\right] +\left[
\begin{array}{cc}
0 & \left( 0,\overrightarrow{v}_{h\mathbf{X}}\right) \\
-\left( 0,\overrightarrow{v}_{h\mathbf{X}}\right) ^{T} & h\mathbf{0}%
\end{array}%
\right] \\
&=&\mathbf{D}_{h\mathbf{X}}\left[ \mathbf{L}_{h\mathbf{X}},\mathbf{e}_{h%
\mathbf{X}}\right] +\frac{1}{2}\left[ \mathbf{L}_{h\mathbf{X}},\left[
\mathbf{L}_{h\mathbf{X}},\mathbf{e}_{h\mathbf{X}}\right] \right] \\
&=&-\mathcal{D}_{h\mathbf{X}}\left[ \mathbf{L}_{h\mathbf{X}},\mathbf{e}_{h%
\mathbf{X}}\right] -\frac{3}{2}|\overrightarrow{v}|^{2}\mathbf{e}_{h\mathbf{X%
}}.
\end{eqnarray*}%
Following above presented identifications related to the first and second
terms, when
\begin{eqnarray*}
|\overrightarrow{v}|^{2} &=&<\left[ \mathbf{L}_{h\mathbf{X}},\mathbf{e}_{h%
\mathbf{X}}\right] ,\left[ \mathbf{L}_{h\mathbf{X}},\mathbf{e}_{h\mathbf{X}}%
\right] >_{h\mathfrak{p}}\longleftrightarrow h\mathbf{g}\left( \mathbf{D}_{h%
\mathbf{X}}\left( h\gamma \right) _{h\mathbf{X}},\mathbf{D}_{h\mathbf{X}%
}\left( h\gamma \right) _{h\mathbf{X}}\right) \\
&=&\left| \mathbf{D}_{h\mathbf{X}}\left( h\gamma \right) _{h\mathbf{X}%
}\right| _{h\mathbf{g}}^{2},
\end{eqnarray*}%
we can identify $\mathcal{D}_{h\mathbf{X}}\left[ \mathbf{L}_{h\mathbf{X}},%
\mathbf{e}_{h\mathbf{X}}\right] $ to $\mathbf{D}_{h\mathbf{X}}^{2}\left(
h\gamma \right) _{h\mathbf{X}}$ and write
\begin{equation*}
-\mathbf{e}_{h\mathbf{Y}}\longleftrightarrow \mathbf{D}_{h\mathbf{X}%
}^{2}\left( h\gamma \right) _{h\mathbf{X}}+\frac{3}{2}\left| \mathbf{D}_{h%
\mathbf{X}}\left( h\gamma \right) _{h\mathbf{X}}\right| _{h\mathbf{g}%
}^{2}~\left( h\gamma \right) _{h\mathbf{X}}
\end{equation*}%
which is just the first equation (\ref{1map}) in the Theorem \ref{mt}
defining a non--stretching mKdV map h--equation induced by the h--part of
the canonical d--connection.

Using the adjoint representation $ad\left( \cdot \right) $ acting in the Lie
algebra $h\mathfrak{g}=h\mathfrak{p}\oplus \mathfrak{so}(n),$ with
\begin{equation*}
ad\left( \left[ \mathbf{L}_{h\mathbf{X}},\mathbf{e}_{h\mathbf{X}}\right]
\right) \mathbf{e}_{h\mathbf{X}}=\left[
\begin{array}{cc}
0 & \left( 0,\overrightarrow{\mathbf{0}}\right) \\
-\left( 0,\overrightarrow{\mathbf{0}}\right) ^{T} & \overrightarrow{\mathbf{v%
}}%
\end{array}%
\right] \in \mathfrak{so}(n+1),
\end{equation*}%
where $\overrightarrow{\mathbf{v}}=-\left[
\begin{array}{cc}
0 & \overrightarrow{v} \\
-\overrightarrow{v}^{T} & h\mathbf{0}%
\end{array}%
\in \mathfrak{so}(n)\right] ,$ and (applying $\ ad\left( \left[ \mathbf{L}_{h%
\mathbf{X}},\mathbf{e}_{h\mathbf{X}}\right] \right) $ \ again)
\begin{equation*}
ad\left( \left[ \mathbf{L}_{h\mathbf{X}},\mathbf{e}_{h\mathbf{X}}\right]
\right) ^{2}\mathbf{e}_{h\mathbf{X}}=-|\overrightarrow{v}|^{2}\left[
\begin{array}{cc}
0 & \left( 1,\overrightarrow{\mathbf{0}}\right) \\
-\left( 1,\overrightarrow{\mathbf{0}}\right) ^{T} & \mathbf{0}%
\end{array}%
\right] =-|\overrightarrow{v}|^{2}\mathbf{e}_{h\mathbf{X}},
\end{equation*}%
the equation (\ref{1map}) can be represented in alternative form
\begin{equation*}
-\left( h\gamma \right) _{\tau }=\mathbf{D}_{h\mathbf{X}}^{2}\left( h\gamma
\right) _{h\mathbf{X}}-\frac{3}{2}\overrightarrow{R}^{-1}ad\left( \mathbf{D}%
_{h\mathbf{X}}\left( h\gamma \right) _{h\mathbf{X}}\right) ^{2}~\left(
h\gamma \right) _{h\mathbf{X}},
\end{equation*}%
which is more convenient for analysis of higher order flows on $%
\overrightarrow{v}$ subjected to higher--order geometric partial
differential equations. Here we note that the $0$ flow one $\overrightarrow{v%
}$ corresponds to just a convective (linear travelling h--wave but subjected
to certain nonholonomic constraints ) map equation (\ref{trmap}).

Now we consider a -1 flow contained in the h--hierarchy derived from the
property that $h\overrightarrow{\mathbf{e}}_{\perp }$ is annihilated by the
h--operator $h\mathcal{J}$ and mapped into $h\mathfrak{R}(h\overrightarrow{%
\mathbf{e}}_{\perp })=0.$This mean that $h\mathcal{J}(h\overrightarrow{%
\mathbf{e}}_{\perp })=\overrightarrow{\varpi }=0.$ Such properties together
with (\ref{auxaaa}) and equations (\ref{floweq}) imply $\mathbf{L}_{\tau }=0$
and hence $h\mathcal{D}_{\tau }\mathbf{e}_{h\mathbf{X}}=[\mathbf{L}_{\tau },%
\mathbf{e}_{h\mathbf{X}}]=0$ for $h\mathcal{D}_{\tau }=h\mathbf{D}_{\tau }+[%
\mathbf{L}_{\tau },\cdot ].$ We obtain the equation of motion for the
h--component of curve, $h\gamma (\tau ,\mathbf{l}),$ following the
correspondences $\mathbf{D}_{h\mathbf{Y}}\longleftrightarrow h\mathcal{D}%
_{\tau }$ and $h\gamma _{\mathbf{l}}\longleftrightarrow \mathbf{e}_{h\mathbf{%
X}},$ $\mathbf{D}_{h\mathbf{Y}}\left( h\gamma (\tau ,\mathbf{l})\right) =0,$
which is just the first equation in (\ref{-1map}).

Finally, we note that the formulas for the v--components, stated by Theorem %
\ref{mt} can be derived in a similar form by respective substitution in the
the above proof of the h--operators and h--variables into v--ones, for
instance, $h\gamma \rightarrow v\gamma ,$ $h\overrightarrow{\mathbf{e}}%
_{\perp }\rightarrow v\overleftarrow{\mathbf{e}}_{\perp },$ $\overrightarrow{%
v}\rightarrow \overleftarrow{v},\overrightarrow{\varpi }\rightarrow
\overleftarrow{\varpi },\mathbf{D}_{h\mathbf{X}}\rightarrow \mathbf{D}_{v%
\mathbf{X}},$ $\mathbf{D}_{h\mathbf{Y}}\rightarrow \mathbf{D}_{v\mathbf{Y}},%
\mathbf{L\rightarrow C,}\overrightarrow{R}\rightarrow \overleftarrow{S},h%
\mathcal{D\rightarrow }v\mathcal{D},$ $h\mathfrak{R\rightarrow }v\mathfrak{R,%
}h\mathcal{J\rightarrow }v\mathcal{J}$,...

\subsection{Nonholonomic mKdV and SG hierarchies}

We consider explicit constructions when solitonic hierarchies are derived
following the conditions of Theorem \ref{mt}.

The h--flow and v--flow equations resulting from (\ref{-1map}) are%
\begin{equation}
\overrightarrow{v}_{\tau }=-\overrightarrow{R}h\overrightarrow{\mathbf{e}}%
_{\perp }\mbox{ \ and \ }\overleftarrow{v}_{\tau }=-\overleftarrow{S}v%
\overleftarrow{\mathbf{e}}_{\perp },  \label{deveq}
\end{equation}%
when, respectively,%
\begin{equation*}
0=\overrightarrow{\varpi }=-\mathbf{D}_{h\mathbf{X}}h\overrightarrow{\mathbf{%
e}}_{\perp }+h\mathbf{e}_{\parallel }\overrightarrow{v},~\mathbf{D}_{h%
\mathbf{X}}h\mathbf{e}_{\parallel }=h\overrightarrow{\mathbf{e}}_{\perp
}\cdot \overrightarrow{v}
\end{equation*}%
and
\begin{equation*}
0=\overleftarrow{\varpi }=-\mathbf{D}_{v\mathbf{X}}v\overleftarrow{\mathbf{e}%
}_{\perp }+v\mathbf{e}_{\parallel }\overleftarrow{v},~\mathbf{D}_{v\mathbf{X}%
}v\mathbf{e}_{\parallel }=v\overleftarrow{\mathbf{e}}_{\perp }\cdot
\overleftarrow{v}.
\end{equation*}%
The d--flow equations possess horizontal and vertical conservation laws%
\begin{equation*}
\mathbf{D}_{h\mathbf{X}}\left( (h\mathbf{e}_{\parallel })^{2}+|h%
\overrightarrow{\mathbf{e}}_{\perp }|^{2}\right) =0,
\end{equation*}%
for $(h\mathbf{e}_{\parallel })^{2}+|h\overrightarrow{\mathbf{e}}_{\perp
}|^{2}=<h\mathbf{e}_{\tau },h\mathbf{e}_{\tau }>_{h\mathfrak{p}}=|\left(
h\gamma \right) _{\tau }|_{h\mathbf{g}}^{2},$ and
\begin{equation*}
\mathbf{D}_{v\mathbf{Y}}\left( (v\mathbf{e}_{\parallel })^{2}+|v%
\overleftarrow{\mathbf{e}}_{\perp }|^{2}\right) =0,
\end{equation*}%
for $(v\mathbf{e}_{\parallel })^{2}+|v\overleftarrow{\mathbf{e}}_{\perp
}|^{2}=<v\mathbf{e}_{\tau },v\mathbf{e}_{\tau }>_{v\mathfrak{p}}=|\left(
v\gamma \right) _{\tau }|_{v\mathbf{g}}^{2}.$ This corresponds to
\begin{equation*}
\mathbf{D}_{h\mathbf{X}}|\left( h\gamma \right) _{\tau }|_{h\mathbf{g}}^{2}=0%
\mbox{ \ and \ }\mathbf{D}_{v\mathbf{X}}|\left( v\gamma \right) _{\tau }|_{v%
\mathbf{g}}^{2}=0.
\end{equation*}

It is possible to rescale conformally the variable $\tau $ in order to get $%
|\left( h\gamma \right) _{\tau }|_{h\mathbf{g}}^{2}$ $=1$ and (it could be
for other rescaling) $|\left( v\gamma \right) _{\tau }|_{v\mathbf{g}}^{2}=1,
$ i.e. to have%
\begin{equation*}
(h\mathbf{e}_{\parallel })^{2}+|h\overrightarrow{\mathbf{e}}_{\perp }|^{2}=1%
\mbox{ \ and \ }(v\mathbf{e}_{\parallel })^{2}+|v\overleftarrow{\mathbf{e}}%
_{\perp }|^{2}=1.
\end{equation*}%
In this case, we can express $h\mathbf{e}_{\parallel }$ and $h%
\overrightarrow{\mathbf{e}}_{\perp }$ in terms of $\overrightarrow{v}$ and
its derivatives and, similarly, we can express $v\mathbf{e}_{\parallel }$
and $v\overleftarrow{\mathbf{e}}_{\perp }$ in terms of $\overleftarrow{v}$
and its derivatives, which follows from (\ref{deveq}). The N--adapted wave
map equations describing the -1 flows reduce to a system of two independent
nonlocal evolution equations for the h-- and v--components,%
\begin{equation*}
\overrightarrow{v}_{\tau }=-\mathbf{D}_{h\mathbf{X}}^{-1}\left( \sqrt{%
\overrightarrow{R}^{2}-|\overrightarrow{v}_{\tau }|^{2}}~\overrightarrow{v}%
\right) \mbox{ \ and \ }\overleftarrow{v}_{\tau }=-\mathbf{D}_{v\mathbf{X}%
}^{-1}\left( \sqrt{\overleftarrow{S}^{2}-|\overleftarrow{v}_{\tau }|^{2}}~%
\overleftarrow{v}\right) .
\end{equation*}%
For d--connections with constant scalar d--curvatures, we can rescale the
equations on $\tau $ to the case when the terms $\overrightarrow{R}^{2},%
\overleftarrow{S}^{2}=1,$ and the evolution equations transform into a
system of hyperbolic d--vector equations,%
\begin{equation}
\mathbf{D}_{h\mathbf{X}}(\overrightarrow{v}_{\tau })=-\sqrt{1-|%
\overrightarrow{v}_{\tau }|^{2}}~\overrightarrow{v}\mbox{ \ and \ }\mathbf{D}%
_{v\mathbf{X}}(\overleftarrow{v}_{\tau })=-\sqrt{1-|\overleftarrow{v}_{\tau
}|^{2}}~\overleftarrow{v},  \label{heq}
\end{equation}%
where $\mathbf{D}_{h\mathbf{X}}=\partial _{h\mathbf{l}}$ and $\mathbf{D}_{v%
\mathbf{X}}=\partial _{v\mathbf{l}}$ are usual partial derivatives on
direction $\mathbf{l=}h\mathbf{l+}v\mathbf{l}$ with $\overrightarrow{v}%
_{\tau }$ and $\overleftarrow{v}_{\tau }$ considered as scalar functions for
the covariant derivatives $\mathbf{D}_{h\mathbf{X}}$ and $\mathbf{D}_{v%
\mathbf{X}}$ defined by the canonical d--connection. It also follows that $h%
\overrightarrow{\mathbf{e}}_{\perp }$ and $v\overleftarrow{\mathbf{e}}%
_{\perp }$ obey corresponding vector sine--Gordon (SG) equations%
\begin{equation}
\left( \sqrt{(1-|h\overrightarrow{\mathbf{e}}_{\perp }|^{2})^{-1}}~\partial
_{h\mathbf{l}}(h\overrightarrow{\mathbf{e}}_{\perp })\right) _{\tau }=-h%
\overrightarrow{\mathbf{e}}_{\perp }  \label{sgeh}
\end{equation}%
and
\begin{equation}
\left( \sqrt{(1-|v\overleftarrow{\mathbf{e}}_{\perp }|^{2})^{-1}}~\partial
_{v\mathbf{l}}(v\overleftarrow{\mathbf{e}}_{\perp })\right) _{\tau }=-v%
\overleftarrow{\mathbf{e}}_{\perp }.  \label{sgev}
\end{equation}

The above presented formulas and Corollary \ref{c2} imply

\begin{conclusion}
The recursion d--operator $\mathfrak{R}=(h\mathfrak{R,}h\mathfrak{R})$ (\ref%
{reqop}), see (\ref{reqoph}) and (\ref{reqopv}), generates two hierarchies
of vector mKdV symmetries: the first one is horizontal,
\begin{eqnarray}
\overrightarrow{v}_{\tau }^{(0)} &=&\overrightarrow{v}_{h\mathbf{l}},~%
\overrightarrow{v}_{\tau }^{(1)}=h\mathfrak{R}(\overrightarrow{v}_{h\mathbf{l%
}})=\overrightarrow{v}_{3h\mathbf{l}}+\frac{3}{2}|\overrightarrow{v}|^{2}~%
\overrightarrow{v}_{h\mathbf{l}},  \label{mkdv1a} \\
\overrightarrow{v}_{\tau }^{(2)} &=&h\mathfrak{R}^{2}(\overrightarrow{v}_{h%
\mathbf{l}})=\overrightarrow{v}_{5h\mathbf{l}}+\frac{5}{2}\left( |%
\overrightarrow{v}|^{2}~\overrightarrow{v}_{2h\mathbf{l}}\right) _{h\mathbf{l%
}}  \notag \\
&&+\frac{5}{2}\left( (|\overrightarrow{v}|^{2})_{h\mathbf{l~}h\mathbf{l}}+|%
\overrightarrow{v}_{h\mathbf{l}}|^{2}+\frac{3}{4}|\overrightarrow{v}%
|^{4}\right) ~\overrightarrow{v}_{h\mathbf{l}}-\frac{1}{2}|\overrightarrow{v}%
_{h\mathbf{l}}|^{2}~\overrightarrow{v},  \notag \\
&&...,  \notag
\end{eqnarray}%
with all such terms commuting with the -1 flow
\begin{equation}
(\overrightarrow{v}_{\tau })^{-1}=h\overrightarrow{\mathbf{e}}_{\perp }
\label{mkdv1b}
\end{equation}%
associated to the vector SG equation (\ref{sgeh}); the second one is
vertical,
\begin{eqnarray}
\overleftarrow{v}_{\tau }^{(0)} &=&\overleftarrow{v}_{v\mathbf{l}},~%
\overleftarrow{v}_{\tau }^{(1)}=v\mathfrak{R}(\overleftarrow{v}_{v\mathbf{l}%
})=\overleftarrow{v}_{3v\mathbf{l}}+\frac{3}{2}|\overleftarrow{v}|^{2}~%
\overleftarrow{v}_{v\mathbf{l}},  \label{mkdv2a} \\
\overleftarrow{v}_{\tau }^{(2)} &=&v\mathfrak{R}^{2}(\overleftarrow{v}_{v%
\mathbf{l}})=\overleftarrow{v}_{5v\mathbf{l}}+\frac{5}{2}\left( |%
\overleftarrow{v}|^{2}~\overleftarrow{v}_{2v\mathbf{l}}\right) _{v\mathbf{l}}
\notag \\
&&+\frac{5}{2}\left( (|\overleftarrow{v}|^{2})_{v\mathbf{l~}v\mathbf{l}}+|%
\overleftarrow{v}_{v\mathbf{l}}|^{2}+\frac{3}{4}|\overleftarrow{v}%
|^{4}\right) ~\overleftarrow{v}_{v\mathbf{l}}-\frac{1}{2}|\overleftarrow{v}%
_{v\mathbf{l}}|^{2}~\overleftarrow{v},  \notag \\
&&...,  \notag
\end{eqnarray}%
with all such terms commuting with the -1 flow
\begin{equation}
(\overleftarrow{v}_{\tau })^{-1}=v\overleftarrow{\mathbf{e}}_{\perp }
\label{mkdv2b}
\end{equation}%
associated to the vector SG equation (\ref{sgev}).
\end{conclusion}

In its turn, using the above Conclusion, we derive that the adjoint
d--operator $\mathfrak{R}^{\ast }=\mathcal{J\circ H}$ generates a horizontal
hierarchy of Hamiltonians,%
\begin{eqnarray}
hH^{(0)} &=&\frac{1}{2}|\overrightarrow{v}|^{2},~hH^{(1)}=-\frac{1}{2}|%
\overrightarrow{v}_{h\mathbf{l}}|^{2}+\frac{1}{8}|\overrightarrow{v}|^{4},
\label{hhh} \\
hH^{(2)} &=&\frac{1}{2}|\overrightarrow{v}_{2h\mathbf{l}}|^{2}-\frac{3}{4}|%
\overrightarrow{v}|^{2}~|\overrightarrow{v}_{h\mathbf{l}}|^{2}-\frac{1}{2}%
\left( \overrightarrow{v}\cdot \overrightarrow{v}_{h\mathbf{l}}\right) +%
\frac{1}{16}|\overrightarrow{v}|^{6},...,  \notag
\end{eqnarray}%
and vertical hierarchy of Hamiltonians%
\begin{eqnarray}
vH^{(0)} &=&\frac{1}{2}|\overleftarrow{v}|^{2},~vH^{(1)}=-\frac{1}{2}|%
\overleftarrow{v}_{v\mathbf{l}}|^{2}+\frac{1}{8}|\overleftarrow{v}|^{4},
\label{hhv} \\
vH^{(2)} &=&\frac{1}{2}|\overleftarrow{v}_{2v\mathbf{l}}|^{2}-\frac{3}{4}|%
\overleftarrow{v}|^{2}~|\overleftarrow{v}_{v\mathbf{l}}|^{2}-\frac{1}{2}%
\left( \overleftarrow{v}\cdot \overleftarrow{v}_{v\mathbf{l}}\right) +\frac{1%
}{16}|\overleftarrow{v}|^{6},...,  \notag
\end{eqnarray}%
all of which are conserved densities for respective horizontal and vertical
-1 flows and determining higher conservation laws for the corresponding
hyperbolic equations (\ref{sgeh}) and (\ref{sgev}).

The above presented horizontal equations (\ref{sgeh}), (\ref{mkdv1a}), (\ref%
{mkdv1b}) and (\ref{hhh}) and of vertical equations (\ref{sgev}), (\ref%
{mkdv2a}), (\ref{mkdv2b}) and (\ref{hhv}) have similar mKdV scaling
symmetries but on different parameters $\lambda _{h}$ and $\lambda _{v}$
because, in general, there are two independent values of scalar curvatures $%
\overrightarrow{R}$ and $\overleftarrow{S},$ see (\ref{sdccurv}). The
horizontal scaling symmetries are $h\mathbf{l\rightarrow }\lambda _{h}h%
\mathbf{l,}\overrightarrow{v}\rightarrow \left( \lambda _{h}\right) ^{-1}%
\overrightarrow{v}$ and $\tau \rightarrow \left( \lambda _{h}\right)
^{1+2k}, $ for $k=-1,0,1,2,...$ For the vertical scaling symmetries, one has
$v\mathbf{l\rightarrow }\lambda _{v}v\mathbf{l,}\overleftarrow{v}\rightarrow
\left( \lambda _{v}\right) ^{-1}\overleftarrow{v}$ and $\tau \rightarrow
\left( \lambda _{v}\right) ^{1+2k},$ for $k=-1,0,1,2,...$

\begin{example}
The simplest way to generate a solitonic hierachy, for instance, defining a
solution of vacuum Einstein equations is to take the value $h_{4}$
(equivalently, $f)$ in (\ref{coeff4d}) to be a solution of a three
dimensional solionic equation. For instance, for $h_{4}$ being a solution of
\begin{equation}
h_{4}{}^{\bullet \bullet }+\epsilon (h_{4}{}^{\prime }+6h_{4}\ h_{4}^{\ast
}+h_{4}^{\ast \ast \ast })^{\ast }=0,\ \epsilon =\pm 1,  \label{solit1}
\end{equation}%
a class of solitonic generic off--diagonal metrics (\ref{es4s}) is obtained
for $\varsigma =1$ and $\ ^{2}n_{k}=0$ and any $w_{i}$ and $\ ^{1}n_{k}$
solving the constraints (\ref{ep2b1}) and (\ref{ep2b2}).
\end{example}

In Refs. \cite{vncg,ijgmmp1,rf3,vrfg,vsgg}, there were provided and reviewed
a number of exact solutions in Einstein, string, gauge, extra dimension,
metric--affine, generalized Finsler--Lagrange and other gravity and/or Ricci
flow theories with nonholonomic commutative and noncommutative variables
generated by nonlinear superpositions of two and/or three dimensional
gravitational solitonic waves on nontrivial (black hole, Taub NUT, pp--wave,
wormhole, black ellipsoid etc) backgrounds. Those classes of solutions where
derived using the N--connection and canonical d--connection and constraints
to the Levi--Civita connection. Such constructions can be equivalently
reformulated in terms of a metric compatible d--connection $\ _{0}\widetilde{%
\mathbf{\Gamma }}_{\ \alpha ^{\prime }\beta ^{\prime }}^{\gamma ^{\prime }}$
(\ref{ccandcon}) and considered as explicit examples of solitonic hierachies
constructed in general form following Theorem \ref{mt}.

\section{Conclusion}

In this paper we have developed a method of converting geometric data
\footnote{%
and physical data, for instance, in Einstein gravity} for a (pseudo)
Riemannian metric into alternative nonholonomic structures and metric
compatible linear connections completely defined by the 'original' metric
tensor. We proved that for any (semi) Riemannian metric on a nonholonomic
manifold $\mathbf{V,}\ \dim \mathbf{V}=n+m,\ n\geq 2$ and $m\geq 1,$ and
corresponding classes of nonholonomic frame deformations, there is a choice
for a linear connection (and corresponding Riemannian and Ricci tensors)
with constant coefficients with respect to a class of nonholonomic frames
with associated nonlinear connection (N--connection) structure.

So, the general conclusion is that a (pseudo) Riemannian geometry can be
described not only in terms of the Levi--Civita connection but also using
any metric compatible linear connection if such an alternative connection is
completely defined the same metric structure. We outline in Table \ref{tab1}
the basic formulas for decomposition of the fundamental geometric objects
under such nonholonomic deformations (in the simplest case) determined by a
N--connection structure.

The local algebraic structure of modelled nonholonomic spaces is defined by
a conventional splitting of dimensions with certain holonomic and
nonholonomic variables (defining a distribution of horizontal and vertical
subspaces). Such subspaces are modelled locally as Riemannian symmetric
manifolds and their properties are exhausted by the geometry of
distinguished Lie groups $\mathbf{G}=GO(n)\oplus $ $GO(m)$ and $\mathbf{G}%
=SU(n)\oplus $ $SU(m)$ and the geometry of N--connections on a conventional
vector bundle with base manifold $\ M,$ $\dim M=n,$ and typical fiber $F,$ $%
\dim F=n.$ This can be formulated equivalently in terms of geometric objects
on couples of Klein spaces. The bi--Hamiltonian and related solitonic (of
type mKdV and SG) hierarchies are generated naturally by wave map equations
and recursion operators associated to the horizontal and vertical flows of
curves on such spaces.

\begin{table}[tbp]
\caption{Metric connections \& geometric structures generated by $\mathbf{g}%
=\{g_{\protect\alpha \protect\beta} \}$ }
\label{tab1}\vskip5pt {\footnotesize \textrm{%
\begin{tabular*}{\textwidth}{@{}l@{\extracolsep{0pt plus11pt}}l|@{\extracolsep{0pt plus11pt}}l|@{\extracolsep{0pt plus11pt}}l|}
\hline\hline
\vline \qquad Geometric & \vline \qquad Levi--Civita & \ canonical
d--connection & \ constant coefficients \\
\vline \qquad objects for: & \vline \qquad connection &  & \qquad
d--connection \\ \hline\hline
\vline Co-frames & \vline \ $e_{\ }^{\beta }= A_{\ \underline{\beta}}^{\beta
}(u)du^{\underline{\beta}}$ & \ $\mathbf{e}^{\alpha}= [e^i=dx^i, $ & \ $%
\mathbf{e}^{{\alpha}^{\prime}}= [e^{i^{\prime}}=dx^{i^{\prime}}, $ \\
\vline & \vline & \qquad$\mathbf{e}^a=dy^a-N^a_j dx^j]$ & \quad ${\mathbf{e}}%
^{a^{\prime}}=dy^{a^{\prime}}- N^{a^{\prime}}_{j^{\prime}} dx^{j^{\prime}}]$
\\ \hline
\vline Metric decomp. & \vline \ $g_{\alpha \beta}= A^{\ \underline{\alpha}%
}_{\alpha } A^{\ \underline{\beta}}_{\beta }g_{\underline{\alpha} \underline{%
\beta}}$ & \quad $\mathbf{g}_{\alpha \beta}= [g_{ij},h_{ab}],\ $ \  & $\ _{0}%
\mathbf{g}_{\alpha ^{\prime }\beta ^{\prime }}= [\ _{0}g_{i ^{\prime }j
^{\prime }},\ _{0}h_{a ^{\prime }b ^{\prime }}],\ $ \\
\vline & \vline & \quad $\mathbf{g}=g_{ij}\ e^{i}\otimes e^{j}$ & \quad $%
\mathbf{g}=\ _{0}g_{i ^\prime j ^\prime}\ e^{i ^{\prime }}\otimes e^{j
^{\prime }}$ \\
\vline & \vline & \qquad $+h_{ab}\ \mathbf{e}^{a}\otimes \mathbf{e}^{b}$ &
\qquad $+\ _{0}h_{a ^\prime b ^\prime}\ \mathbf{e}^{a ^\prime}\otimes
\mathbf{e}^{b ^\prime}$ \\
\vline & \vline &  & \ $g_{i^{\prime }j^{\prime }}=A_{\ i^{\prime }}^{i}A_{\
j^{\prime}}^{j}g_{ij},$ \\
\vline & \vline &  & \ $h_{a^{\prime }b^{\prime }}=A_{\ a^{\prime }}^{a}A_{\
b^{\prime }}^{b}h_{ab}$ \\ \hline
\vline Connections & \vline \qquad $_{\shortmid }\Gamma _{\ \alpha \beta
}^{\gamma }$ & $\ _{\shortmid }\Gamma _{\ \alpha \beta }^{\gamma }=\widehat{%
\mathbf{\Gamma }}_{\ \alpha \beta }^{\gamma }+\ _{\shortmid }Z_{\ \alpha
\beta }^{\gamma }$ & $\ _{0}\widetilde{\mathbf{\Gamma }}_{\ \alpha ^{\prime
}\beta ^{\prime}}^{\gamma ^{\prime }} =( \widehat{L}_{j^{\prime
}k^{\prime}}^{i^{\prime }}=0,$ \\
\vline and distorsions & \vline &  & \ $\widehat{L}_{b^{\prime }k^{\prime
}}^{a^{\prime }}=\ _{0}\widehat{L}_{b^{\prime }k^{\prime }}^{a^{\prime
}}=const.,$ \\
\vline & \vline &  & \qquad $\widehat{C}_{j^{\prime }c^{\prime
}}^{i^{\prime}}=0, \widehat{C}_{b^{\prime }c^{\prime }}^{a^{\prime }}=0)$ \\
\hline
\vline Riemannian & \vline \qquad $_{\shortmid }R_{~\beta \gamma \delta
}^{\alpha }$ & $\qquad \widehat{\mathbf{R}}_{~\beta \gamma \delta }^{\alpha
} $ & $\ _{0}\widetilde{\mathbf{R}}_{\ \beta ^{\prime }\gamma ^{\prime
}\delta ^{\prime }}^{\alpha ^{\prime }} =(0,\ _{0}\widetilde{R}%
_{~b^{\prime}j^{\prime}k^{\prime}}^{a^{\prime}}$ \\
\vline (d--)tensors & \vline &  & $\qquad =const.,0,0,0,0)$ \\ \hline
\vline \ Ricci(d-)tensors & \vline & $\qquad \widehat{R}_{ij}=\ ^h\lambda \
g_{ij},$ & \qquad constraints \\
\vline \ Einstein eqs. & \vline \quad $_{\shortmid }R_{\beta \gamma}=\lambda
g_{\beta \gamma}$ & $\qquad \widehat{R}_{ab}=\ ^v\lambda \ h_{ab},$ & \qquad
on distorsion \\
\vline & \vline & \quad $\widehat{R}_{ib}=0,\ \widehat{R}_{bi}=0$ & \qquad
d--tensors \\ \hline\hline
\end{tabular*}
}}
\end{table}

One should be emphasized that N--connections can be considered both in
(pseudo) Riemannian and Finsler--Lagrange geometries, see discussions in
Refs. \cite{ijgmmp1,vrfg,vsgg,ma2}, but in the first case to prescribe a
N--connection is to fix a conventional (in general, nonholonomic) splitting
on the manifold under consideration. The point is to consider such a
splitting and relevant nonholonomic distribution which are convenient for
further solitonic constructions This allowed us to elaborate a "solitonic"
approach when the geometry of (semi) Riemannian / Einstein manifolds is
encoded into nonholonomic hierarchies of bi--Hamiltonian structures and
related solitonic equations derived for curve flows on spaces with
conventional splitting of dimensions.

The main result of this work is the proof that any metric structure on a
(pseudo) Riemannian manifold can be decomposed into solitonic data with
corresponding hierarchies of nonlinear waves. Such constructions hold true
for more general classes of commutative and noncommutative metric--affine,
Finsler--Lagrange--Hamilton, their generalizations to nonsymmetric metrics
and/or nonholonomic Fedosov manifolds and their Ricci flows \cite%
{vncg,vsgg,rf2,rf3,ncrf,rfns,plafq,vlqgfq}. Nevertheless, the solution of
the "inverse" problem to state the conditions when it is possible to extract
certain general (non) commutative / (non) holonomic / (non) symmetric
geometries etc from a given solitonic hierarchy it is a purpose for future
work.

Let us speculate on possible important consequences for further research and  developments in different branches of mathematics, classical and quantum physics and mechanics of the results obtained in this work:
\begin{enumerate}
\item We illustrated how the geometric constructions on Riemannian spaces can re--defined equivalently in terms of nonholonomic structures, nonlinear connections and associated nonholonomic frames and encoded into alternative connections with constant coefficient curvatures. Such nonholonomic transforms allow us to apply a number of geometric methods formally elaborated in Finsler geometry and Lagrange and Hamilton mechanics for constructing new classes of exact solutions in gravity theories and Ricci flow models.
\item The method of nonholonomic frames and deformations can be considered with "inverse" purposes, when various types of nonlinear fundamental physical interactions are modelled as certain effective mechanical or continuous media theories. This presents certain interest for possible computer simulations and experimental laboratory research for a number of effects in modern quantum gravity, black hole physics, astrophysics and cosmology when real experiments with high energies are not possible.
 \item Analyzing possible curve flows on various type of  phase spaces, classical manifolds and fibred spaces, it is possible to draw very fundamental conclusions on the type of interactions and their symmetries, invariants and conservation laws. Such information can be encoded as geometric data on a corresponding class of nonholonomic Klein spaces and relevant bi--Hamilton operators.
 \item A  fundamental result advocated in this paper is that the dynamics of gravitational field interactions in the Einstein gravity and generalizations can be encoded into certain series of solitonic hierarchies and associated conservation laws and nonholonomic constraints.
 \item The mentioned "solitonic nonlinear decomposition" of gravitational interactions presents a substantial interest for elaborating new methods of quantization and their verification by associated models with effective quantum liquids, in low--temperature physics etc.
 \item Inversely, to the previous point, further developments of the theory of quantum solitonic equations are possible by using former methods of deformation and geometric quantization, new concepts from quantum gravity and gauge theories.
 \item In our recent works, we also investigated some  connections between the theory of nonholonomic Ricci flows and constrained curve flows. The approach has strong connections to the theory of diffusion, stochastic and kinetic processes and non--equilibrium thermodynamics in locally anisotropic spaces/media.
\item The elaborated theoretical/geometrical methods from classical and quantum gravity, for instance, can be applied for a study of nonlinear solitonic interactions in wave mechanics, electrodynamic processes in various continuous and nonhomogeneous media.
\end{enumerate}

Finally we note that there are many physically interesting models when
solitonic hierarchies were constructed in modern theories of gravity and
Ricci flows of physically valuable solutions in gravity \cite%
{vncg,ijgmmp1,rf3,vrfg,vsgg}. They positively can be imbedded as particular
cases of bi--Hamiltonian structures and related solitonic (of type mKdV and
SG) hierarchies constructed in this work. There are various ideas how to
consider anholonomic and parametric deformations (like in Refs. \cite%
{ijgmmp1,rf3}) of such a hierarchy into another nonlinear solitonic
superposition in order to generate a new class of classical or quantum
solutions of the Einstein equations. This way, by nonholonomic parametric
distributions, we naturally model on (pseudo) Riemannian manifolds various
classes of generalized geometries encoded into (non) commutative / classical
and/or quantum generalizations of soltionic equations. We are continuing to
work in such directions.

\vskip5pt

\textbf{Acknowledgement: }The work was performed during a visit at Fields
Institute.

\setcounter{equation}{0} \renewcommand{\theequation}
{A.\arabic{equation}} \setcounter{subsection}{0}
\renewcommand{\thesubsection}
{A.\arabic{subsection}}

\appendix

\section{N--anholonomic Riemann Manifolds}

\label{s2} In this section we briefly recall some basic definitions and
facts concerning the geometry of (pseudo) Riemannian nonholonomic manfolds.

A pair $(\mathbf{V},\mathcal{N}),$ where $\mathbf{V}$ is a manifold and $%
\mathcal{N}$ \ is a nonintegrable distribution on $\mathbf{V}$, is a
nonholonomic manifold (in this work, we consider real manifolds of necessary
smooth class).

The concept of nonholonomic manifold was introduced independently by G. Vr\v{%
a}nceanu \cite{vr1,vr2} and Z. Horak \cite{hor} for geometric
interpretations of nonholonomic mechanical systems and considered new
classes of linear connections, which were different from the Levi--Civita
connection (see modern approaches and historical remarks in Refs. \cite%
{bejf,vsgg,vrfg,ijgmmp1}).

\subsection{Nonholonomic distributions and N--connections}

Let us consider a real smooth (pseudo) Riemann $(n+m)$--dimensional manifold
$\mathbf{V,}$ with $n\geq 2$ and $m\geq 1.$ The local coordinates on $%
\mathbf{V}$ are denoted $u=(x,y),$ or $u^{\alpha }=\left( x^{i},y^{a}\right)
,$ where the ''horizontal'' (h) indices run the values $i,j,k,\ldots
=1,2,\ldots ,n$ and the ''vertical'' (v) indices run the values $%
a,b,c,\ldots =n+1,n+2,\ldots ,n+m.$ With respect to a local coordinate base,
we parameterize a metric structure on $\mathbf{V}$ in the form

\begin{equation}
\mathbf{\ g}=\underline{g}_{\alpha \beta }\left( u\right) du^{\alpha
}\otimes du^{\beta }  \label{metr}
\end{equation}%
with coefficients%
\begin{equation}
\underline{g}_{\alpha \beta }=\left[
\begin{array}{cc}
g_{ij}\left( u\right) +\underline{N}_{i}^{a}\left( u\right) \underline{N}%
_{j}^{b}\left( u\right) h_{ab}\left( u\right) & \underline{N}_{j}^{e}\left(
u\right) h_{ae}\left( u\right) \\
\underline{N}_{i}^{e}\left( u\right) h_{be}\left( u\right) & h_{ab}\left(
u\right)%
\end{array}%
\right] .  \label{ansatz}
\end{equation}

We consider a map $\pi :\mathbf{V}\rightarrow V,$ $\dim V=n,$ and denote by $%
\pi ^{\top }:T\mathbf{V}\rightarrow TV$ the differential of $\pi $ defined
by fiber preserving morphisms of the tangent bundles $T\mathbf{V}$ and $TV.$
The kernel of $\pi ^{\top }$ is just the vertical subspace $v\mathbf{V}$
with a related inclusion mapping $i:v\mathbf{V}\rightarrow T\mathbf{V}.$

\begin{definition}
A nonlinear connection (N--connection) $\mathbf{N}$ on a manifold $\mathbf{V}
$ is defined by the splitting on the left of an exact sequence
\begin{equation*}
0\rightarrow v\mathbf{V}\overset{i}{\rightarrow} T\mathbf{V}\rightarrow T%
\mathbf{V}/v\mathbf{V}\rightarrow 0,
\end{equation*}%
i. e. by a morphism of submanifolds $\mathbf{N:\ \ }T\mathbf{V}\rightarrow v%
\mathbf{V}$ such that $\mathbf{N\circ i}$ is the unity in $v\mathbf{V}.$
\end{definition}

Locally, a N--connection is defined by its coefficients $N_{i}^{a}(u),$%
\begin{equation}
\mathbf{N}=N_{i}^{a}(u)dx^{i}\otimes \frac{\partial }{\partial y^{a}}.
\label{coeffnc}
\end{equation}%
Globalizing the local splitting, one prove that any N--connection is defined
by a Whitney sum of conventional horizontal (h) subspace, $\left( h\mathbf{V}%
\right) ,$ and vertical (v) subspace, $\left( v\mathbf{V}\right) ,$
\begin{equation}
T\mathbf{V}=h\mathbf{V}\oplus v\mathbf{V}.  \label{whitney}
\end{equation}

The sum (\ref{whitney}) states on $T\mathbf{V}$ a nonholonomic distribution
of horizontal and vertical subspaces. The linear connections those which are
linear on $y^{a},$ i.e. $N_{i}^{a}(u)=\Gamma _{bj}^{a}(x)y^{b}.$

For simplicity, we shall work with a particular class of nonholonomic
manifolds:

\begin{definition}
\label{defanhm} A manifold $\mathbf{V}$ is N--anholonomic if its tangent
space $T\mathbf{V}$ is enabled with a N--connection structure (\ref{whitney}%
).
\end{definition}

On a (pseudo) Riemannian manifold, we can define a N--connection structure
induced by a formal $(n+m)$--splitting, when the N--connection coefficients (%
\ref{coeffnc}) are determined by certain off--diagonal terms in (\ref{ansatz}%
) for\ $N_{i}^{a}=\underline{N}_{i}^{a}.$ Such a N--anholonomic manifold is
provided with a local fibered structure which is fixed following certain
symmetry conditions and/or constraints imposed on the dynamics of
gravitational fields.

A N--anholonomic manifold is characterized by its curvature:

\begin{definition}
The N--connection curvature is defined as the Neijenhuis tensor,%
\begin{equation*}
\mathbf{\Omega }(\mathbf{X,Y})\doteqdot \lbrack v\mathbf{X},v\mathbf{Y}]+\ v[%
\mathbf{X,Y}]- v[v\mathbf{X},\mathbf{Y}]-v[\mathbf{X},v\mathbf{Y}].
\label{njht}
\end{equation*}
\end{definition}

In local form, we have for $\mathbf{\Omega }=\frac{1}{2}\Omega _{ij}^{a}\
d^{i}\wedge d^{j}\otimes \partial _{a}$ the coefficients%
\begin{equation}
\Omega _{ij}^{a}=\frac{\partial N_{i}^{a}}{\partial x^{j}}-\frac{\partial
N_{j}^{a}}{\partial x^{i}}+N_{i}^{b}\frac{\partial N_{j}^{a}}{\partial y^{b}}%
-N_{j}^{b}\frac{\partial N_{i}^{a}}{\partial y^{b}}.  \label{ncurv}
\end{equation}

Performing a frame (vielbein) transform $\mathbf{e}_{\alpha }=\mathbf{A}%
_{\alpha }^{\ \underline{\alpha }}\partial _{\underline{\alpha }}$ and $%
\mathbf{e}_{\ }^{\beta }=\mathbf{A}_{\ \underline{\beta }}^{\beta }du^{%
\underline{\beta }},$ where we underline the local coordinate indices, when $%
\partial _{\underline{\alpha }}=\partial /\partial u^{\underline{\alpha }%
}=(\partial _{\underline{i}}=\partial /\partial x^{\underline{i}},\partial
/\partial y^{\underline{a}}),$ with coefficients

\begin{equation}
\mathbf{A}_{\alpha }^{\ \underline{\alpha }}(u)=\left[
\begin{array}{cc}
e_{i}^{\ \underline{i}}(u) & N_{i}^{b}(u)e_{b}^{\ \underline{a}}(u) \\
0 & e_{a}^{\ \underline{a}}(u)%
\end{array}%
\right] ,~\mathbf{A}_{\ \underline{\beta }}^{\beta }(u)=\left[
\begin{array}{cc}
e_{\ \underline{i}}^{i\ }(u) & -N_{k}^{b}(u)e_{\ \underline{i}}^{k\ }(u) \\
0 & e_{\ \underline{a}}^{a\ }(u)%
\end{array}%
\right] ,  \label{naft}
\end{equation}%
we transform the metric (\ref{whitney}) into a distinguished metric
(d--metric)
\begin{equation}
\mathbf{g}=~^{h}g+~^{v}h=\ g_{ij}(x,y)\ e^{i}\otimes e^{j}+\ h_{ab}(x,y)\
\mathbf{e}^{a}\otimes \mathbf{e}^{b},  \label{m1}
\end{equation}%
for an associated, to a N--connection, frame (vielbein) structure $\mathbf{e}%
_{\nu }=(\mathbf{e}_{i},e_{a}),$ where
\begin{equation}
\mathbf{e}_{i}=\frac{\partial }{\partial x^{i}}-N_{i}^{a}(u)\frac{\partial }{%
\partial y^{a}}\mbox{ and
}e_{a}=\frac{\partial }{\partial y^{a}},  \label{dder}
\end{equation}%
and the dual frame (coframe) structure $\mathbf{e}^{\mu }=(e^{i},\mathbf{e}%
^{a}),$ where
\begin{equation}
e^{i}=dx^{i}\mbox{ and }\mathbf{e}^{a}=dy^{a}+N_{i}^{a}(u)dx^{i}.
\label{ddif}
\end{equation}

The geometric objects on $\mathbf{V}$ can be defined in a form adapted to
the N--connection structure following certain decompositions which are
invariant under parallel transports preserving the splitting (\ref{whitney}%
). In this case, we call them to be distinguished (by the N--connection
structure), i.e. d--objects. For instance, a vector field $\mathbf{X}\in T%
\mathbf{V}$ \ is expressed
\begin{equation*}
\mathbf{X}=(hX,\ vX),\mbox{ \ or \ }\mathbf{X}=X^{\alpha }\mathbf{e}_{\alpha
}=X^{i}\mathbf{e}_{i}+X^{a}e_{a},
\end{equation*}%
where $hX=X^{i}\mathbf{e}_{i}$ and $vX=X^{a}e_{a}$ state, respectively, the
adapted to the N--connection structure horizontal (h) and vertical (v)
components of the vector. In brief, $\mathbf{X}$ is called a distinguished
vectors, in brief, d--vector). In a similar fashion, the geometric objects
on $\mathbf{V}$ like tensors, spinors, connections, ... are called
respectively d--tensors, d--spinors, d--connections if they are adapted to
the N--connection splitting (\ref{whitney}).

The vielbeins (\ref{dder}) and (\ref{ddif}) are called respectively
N--adapted frames and coframes. In order to preserve a relation with some
previous our notations \cite{vncg,vsgg}, we emphasize that $\mathbf{e}_{\nu
}=(\mathbf{e}_{i},e_{a})$ and $\mathbf{e}^{\mu }=(e^{i},\mathbf{e}^{a})$ are
correspondingly the former ''N--elongated'' partial derivatives $\delta
_{\nu }=\delta /\partial u^{\nu }=(\delta _{i},\partial _{a})$ and
''N--elongated'' differentials $\delta ^{\mu }=\delta u^{\mu }=(d^{i},\delta
^{a}).$

The vielbeins (\ref{ddif}) satisfy the nonholonomy relations
\begin{equation}
\lbrack \mathbf{e}_{\alpha },\mathbf{e}_{\beta }]=\mathbf{e}_{\alpha }%
\mathbf{e}_{\beta }-\mathbf{e}_{\beta }\mathbf{e}_{\alpha }=W_{\alpha \beta
}^{\gamma }\mathbf{e}_{\gamma }  \label{anhrel}
\end{equation}%
with anholonomy coefficients $W_{ia}^{b}=\partial _{a}N_{i}^{b}$ and $%
W_{ji}^{a}=\Omega _{ij}^{a}.$

\subsection{D--Connections}

\label{assectdcon}We perform all geometric constructions on N--anholonomic
manifolds.

\begin{definition}
A distinguished connection (in brief, d--connection) $\mathbf{D}=(h\mathbf{D}%
,v\mathbf{D})$ is a linear connection preserving under parallel transports
the nonholonomic decomposition (\ref{whitney}).
\end{definition}

The N--adapted components $\mathbf{\Gamma }_{\ \beta \gamma }^{\alpha }$ of
a d--connection $\mathbf{D}_{\alpha }=(\mathbf{e}_{\alpha }\rfloor \mathbf{D}%
)$ are defined by equations
\begin{equation}
\mathbf{D}_{\alpha }\mathbf{e}_{\beta }=\mathbf{\Gamma }_{\ \alpha \beta
}^{\gamma }\mathbf{e}_{\gamma },\mbox{\ or \ }\mathbf{\Gamma }_{\ \alpha
\beta }^{\gamma }\left( u\right) =\left( \mathbf{D}_{\alpha }\mathbf{e}%
_{\beta }\right) \rfloor \mathbf{e}^{\gamma }.  \label{dcon1}
\end{equation}%
The N--adapted splitting into h-- and v--covariant derivatives is stated by
\begin{equation*}
h\mathbf{D}=\{\mathbf{D}_{k}=\left( L_{jk}^{i},L_{bk\;}^{a}\right) \},%
\mbox{
and }\ v\mathbf{D}=\{\mathbf{D}_{c}=\left( C_{jk}^{i},C_{bc}^{a}\right) \},
\end{equation*}%
where, by definition, $L_{jk}^{i}=\left( \mathbf{D}_{k}\mathbf{e}_{j}\right)
\rfloor e^{i},$ $L_{bk}^{a}=\left( \mathbf{D}_{k}e_{b}\right) \rfloor
\mathbf{e}^{a},$ $C_{jc}^{i}=\left( \mathbf{D}_{c}\mathbf{e}_{j}\right)
\rfloor e^{i},$ $C_{bc}^{a}=\left( \mathbf{D}_{c}e_{b}\right) \rfloor
\mathbf{e}^{a}.$ The components $\mathbf{\Gamma }_{\ \alpha \beta }^{\gamma
}=\left( L_{jk}^{i},L_{bk}^{a},C_{jc}^{i},C_{bc}^{a}\right) $ completely
define a d--connection $\mathbf{D}$ on $\mathbf{E}.$

From the class of arbitrary d--connections $\mathbf{D}$ on $\mathbf{V,}$ one
distinguishes those which are metric compatible (metrical) satisfying the
condition%
\begin{equation}
\mathbf{Dg=0};  \label{metcomp}
\end{equation}%
i.e. for h- and v-projections $D_{j}g_{kl}=0,$ $D_{a}g_{kl}=0,$ $%
D_{j}h_{ab}=0,$ $D_{a}h_{bc}=0.$

On a N--anholonomic (semi) Riemannian manifold $\mathbf{V},$ there are two
types of preferred linear connections uniquely determined by a generic
off--diagonal metric structure with $n+m$ splitting, see $\mathbf{g}=g\oplus
_{N}h$ (\ref{m1}):

\begin{enumerate}
\item The Levi--Civita connection $\nabla =\{\Gamma _{\beta \gamma }^{\alpha
}\}$ is by definition torsionless, $~\ _{\shortmid }\mathcal{T}=0,$ and
satisfies the metric compatibility condition, $\nabla \mathbf{g}=0.$

\item The canonical d--connection $\widehat{\mathbf{\Gamma }}_{\ \alpha
\beta }^{\gamma }=\left( \widehat{L}_{jk}^{i},\widehat{L}_{bk}^{a},\widehat{C%
}_{jc}^{i},\widehat{C}_{bc}^{a}\right) $ is also metric compatible, i. e. $%
\widehat{\mathbf{D}}\mathbf{g}=0,$ but the torsion vanishes only on h-- and
v--subspaces, i.e. $\widehat{T}_{jk}^{i}=0$ and $\widehat{T}_{bc}^{a}=0,$
for certain nontrivial values of $\widehat{T}_{ja}^{i},\widehat{T}_{bi}^{a},%
\widehat{T}_{ji}^{a}.$
\end{enumerate}

For simplicity, we omit hats on symbols and write $L_{jk}^{i}$ instead of $%
\widehat{L}_{jk}^{i},$ $T_{ja}^{i}$ instead of $\widehat{T}_{ja}^{i}$ and so
on, for a d--connection $\mathbf{\Gamma }_{\ \alpha \beta }^{\gamma }.$

By a straightforward calculus with respect to N--adapted frames (\ref{dder})
and (\ref{ddif}), one can verify that the requested properties for $\widehat{%
\mathbf{D}}$ on $\mathbf{V}$ are satisfied if
\begin{eqnarray}
L_{jk}^{i} &=&\frac{1}{2}g^{ir}\left( \mathbf{e}_{k}g_{jr}+\mathbf{e}%
_{j}g_{kr}-\mathbf{e}_{r}g_{jk}\right) ,  \label{candcon} \\
L_{bk}^{a} &=&e_{b}(N_{k}^{a})+\frac{1}{2}h^{ac}\left( \mathbf{e}%
_{k}h_{bc}-h_{dc}\ e_{b}N_{k}^{d}-h_{db}\ e_{c}N_{k}^{d}\right) ,  \notag \\
C_{jc}^{i} &=&\frac{1}{2}g^{ik}e_{c}g_{jk},\ C_{bc}^{a}=\frac{1}{2}%
h^{ad}\left( e_{c}h_{bd}+e_{c}h_{cd}-e_{d}h_{bc}\right) .  \notag
\end{eqnarray}%
For dimensions $n=m,$ we can consider the so--called normal d--connection%
\footnote{%
i.e. it has the same coefficients as the Levi--Civita connection with
respect to N--elongated bases (\ref{dder}) and (\ref{ddif})} $\ \mathbf{%
\tilde{D}}=(h\tilde{D},v\tilde{D})$ with the coefficients $\Gamma _{\ \beta
\gamma }^{\alpha }=(L_{\ jk}^{i},L_{bc}^{a}),$
\begin{eqnarray}
L_{\ jk}^{i} &=&\frac{1}{2}g^{ih}(\mathbf{e}_{k}g_{jh}+\mathbf{e}_{j}g_{kh}-%
\mathbf{e}_{h}g_{jk}),  \label{candcontm} \\
C_{\ bc}^{a} &=&\frac{1}{2}h^{ae}(e_{c}h_{be}+e_{b}h_{ce}-e_{e}h_{bc}).
\notag
\end{eqnarray}%
A straightforward calculus shows that the coefficients of the Levi--Civita
connection can be expressed in the form
\begin{equation}
\ _{\shortmid }\Gamma _{\ \alpha \beta }^{\gamma }=\widehat{\mathbf{\Gamma }}%
_{\ \alpha \beta }^{\gamma }+\ _{\shortmid }Z_{\ \alpha \beta }^{\gamma },
\label{cdeft}
\end{equation}%
where
\begin{eqnarray}
\ _{\shortmid }Z_{jk}^{i} &=&0,\ _{\shortmid
}Z_{jk}^{a}=-C_{jb}^{i}g_{ik}h^{ab}-\frac{1}{2}\Omega _{jk}^{a},~_{\shortmid
}Z_{bk}^{i}=\frac{1}{2}\Omega _{jk}^{c}h_{cb}g^{ji}-\Xi
_{jk}^{ih}~C_{hb}^{j},  \notag \\
_{\shortmid }Z_{bk}^{a} &=&~^{+}\Xi _{cd}^{ab}~\left[
L_{bk}^{c}-e_{b}(N_{k}^{c})\right] ,\ _{\shortmid }Z_{kb}^{i}=\frac{1}{2}%
\Omega _{jk}^{a}h_{cb}g^{ji}+\Xi _{jk}^{ih}~C_{hb}^{j},  \label{cdeftc} \\
\ _{\shortmid }Z_{jb}^{a} &=&-~^{-}\Xi _{cb}^{ad}~~^{\circ }L_{dj}^{c},\
_{\shortmid }Z_{bc}^{a}=0,\ _{\shortmid }Z_{ab}^{i}=-\frac{g^{ij}}{2}\left[
~^{\circ }L_{aj}^{c}h_{cb}+~^{\circ }L_{bj}^{c}h_{ca}\right] ,  \notag \\
\Xi _{jk}^{ih} &=&\frac{1}{2}(\delta _{j}^{i}\delta
_{k}^{h}-g_{jk}g^{ih}),~^{\pm }\Xi _{cd}^{ab}=\frac{1}{2}(\delta
_{c}^{a}\delta _{d}^{b}{\pm }h_{cd}h^{ab}),  \notag
\end{eqnarray}%
for $\Omega _{jk}^{a}$ computed as in formula (\ref{ncurv}), $~^{\circ
}L_{aj}^{c}=L_{aj}^{c}-e_{a}(N_{j}^{c})$ and
\begin{eqnarray*}
\ _{\shortmid }\Gamma _{\beta \gamma }^{\alpha } &=&\left( \ _{\shortmid
}L_{jk}^{i},\ _{\shortmid }L_{jk}^{a},\ _{\shortmid }L_{bk}^{i},\
_{\shortmid }L_{bk}^{a},\ _{\shortmid }C_{jb}^{i},\ _{\shortmid
}C_{jb}^{a},\ _{\shortmid }C_{bc}^{i},\ _{\shortmid }C_{bc}^{a}\right) , \\
\bigtriangledown _{\mathbf{e}_{k}}(\mathbf{e}_{j}) &=&\ _{\shortmid
}L_{jk}^{i}\mathbf{e}_{i}+\ _{\shortmid }L_{jk}^{a}e_{a},\ \bigtriangledown
_{\mathbf{e}_{k}}(e_{b})=\ _{\shortmid }L_{bk}^{i}\mathbf{e}_{i}+\
_{\shortmid }L_{bk}^{a}e_{a}, \\
\bigtriangledown _{e_{b}}(\mathbf{e}_{j}) &=&~_{\shortmid }C_{jb}^{i}\mathbf{%
e}_{i}+\ _{\shortmid }C_{jb}^{a}e_{a},\ \bigtriangledown _{e_{c}}(e_{b})=\
_{\shortmid }C_{bc}^{i}\mathbf{e}_{i}+\ _{\shortmid }C_{bc}^{a}e_{a}.
\end{eqnarray*}%
It should be emphasized that all components of $\ _{\shortmid }\Gamma _{\
\alpha \beta }^{\gamma },\widehat{\mathbf{\Gamma }}_{\ \alpha \beta
}^{\gamma }$ and$\ _{\shortmid }Z_{\ \alpha \beta }^{\gamma }$ are defined
by the coefficients of \ d--metric $\mathbf{g}$ (\ref{m1}) and N--connection
$\mathbf{N}$ (\ref{coeffnc}), or equivalently by the coefficients of the
corresponding generic off--diagonal metric\ (\ref{ansatz}).

The simplest way to perform computations with d--connections is to use
N--adapted differential forms like $\mathbf{\Gamma }_{\ \beta }^{\alpha }=%
\mathbf{\Gamma }_{\ \beta \gamma}^{\alpha }\mathbf{e}^{\gamma }$ with the
coefficients defined with respect to (\ref{ddif}) and (\ref{dder}). Torsion
of a d--connection can be computed
\begin{equation*}
\mathcal{T}^{\alpha }\doteqdot \mathbf{De}^{\alpha }=d\mathbf{e}^{\alpha
}+\Gamma _{\ \beta }^{\alpha }\wedge \mathbf{e}^{\beta }.
\end{equation*}%
Locally it is characterized by (N--adapted) d--torsion coefficients
\begin{eqnarray}
T_{\ jk}^{i} &=&L_{\ jk}^{i}-L_{\ kj}^{i},\ T_{\ ja}^{i}=-T_{\ aj}^{i}=C_{\
ja}^{i},\ T_{\ ji}^{a}=\Omega _{\ ji}^{a},\   \label{dtors} \\
T_{\ bi}^{a} &=&-T_{\ ib}^{a}=\frac{\partial N_{i}^{a}}{\partial y^{b}}-L_{\
bi}^{a},\ T_{\ bc}^{a}=C_{\ bc}^{a}-C_{\ cb}^{a}.  \notag
\end{eqnarray}

The curvature of a d--connection $\mathbf{D,}$
\begin{equation}
\mathcal{R}_{~\beta }^{\alpha }\doteqdot \mathbf{D\Gamma }_{\ \beta
}^{\alpha }=d\mathbf{\Gamma }_{\ \beta }^{\alpha }-\mathbf{\Gamma }_{\ \beta
}^{\gamma }\wedge \mathbf{\Gamma }_{\ \gamma }^{\alpha },  \label{curv}
\end{equation}%
splits into six types of N--adapted components with respect to (\ref{dder})
and (\ref{ddif}),
\begin{equation*}
\mathbf{R}_{~\beta \gamma \delta }^{\alpha }=\left(
R_{~hjk}^{i},R_{~bjk}^{a},P_{~hja}^{i},P_{~bja}^{c},S_{~jbc}^{i},S_{~bdc}^{a}\right) ,
\end{equation*}%
\begin{eqnarray}
R_{\ hjk}^{i} &=&\mathbf{e}_{k}L_{\ hj}^{i}-\mathbf{e}_{j}L_{\ hk}^{i}+L_{\
hj}^{m}L_{\ mk}^{i}-L_{\ hk}^{m}L_{\ mj}^{i}-C_{\ ha}^{i}\Omega _{\ kj}^{a},
\label{dcurv} \\
R_{\ bjk}^{a} &=&\mathbf{e}_{k}L_{\ bj}^{a}-\mathbf{e}_{j}L_{\ bk}^{a}+L_{\
bj}^{c}L_{\ ck}^{a}-L_{\ bk}^{c}L_{\ cj}^{a}-C_{\ bc}^{a}\Omega _{\ kj}^{c},
\notag \\
P_{\ jka}^{i} &=&e_{a}L_{\ jk}^{i}-D_{k}C_{\ ja}^{i}+C_{\ jb}^{i}T_{\
ka}^{b},~P_{\ bka}^{c}=e_{a}L_{\ bk}^{c}-D_{k}C_{\ ba}^{c}+C_{\ bd}^{c}T_{\
ka}^{c},  \notag \\
S_{\ jbc}^{i} &=&e_{c}C_{\ jb}^{i}-e_{b}C_{\ jc}^{i}+C_{\ jb}^{h}C_{\
hc}^{i}-C_{\ jc}^{h}C_{\ hb}^{i},  \notag \\
S_{\ bcd}^{a} &=&e_{d}C_{\ bc}^{a}-e_{c}C_{\ bd}^{a}+C_{\ bc}^{e}C_{\
ed}^{a}-C_{\ bd}^{e}C_{\ ec}^{a}.  \notag
\end{eqnarray}

Contracting respectively the components, $\mathbf{R}_{\alpha \beta
}\doteqdot \mathbf{R}_{\ \alpha \beta \tau }^{\tau },$ one computes the h-
v--components of the Ricci d--tensor
\begin{equation}
R_{ij}\doteqdot R_{\ ijk}^{k},\ \ R_{ia}\doteqdot -P_{\ ika}^{k},\
R_{ai}\doteqdot P_{\ aib}^{b},\ S_{ab}\doteqdot S_{\ abc}^{c}.
\label{dricci}
\end{equation}%
The scalar curvature is defined by contracting the Ricci d--tensor with the
inverse metric $\mathbf{g}^{\alpha \beta },$
\begin{equation}
\overleftrightarrow{\mathbf{R}}\doteqdot \mathbf{g}^{\alpha \beta }\mathbf{R}%
_{\alpha \beta }=g^{ij}R_{ij}+h^{ab}S_{ab}=\overrightarrow{R}+\overleftarrow{%
S}.  \label{sdccurv}
\end{equation}

For any $\mathbf{\Gamma (g)},$ there is a nontrivial torsion $\mathbf{T}(%
\mathbf{g})$ \ with coefficients (\ref{dtors}). This torsion is induced
nonholonomically as an effective one (by anholonomy coefficients, see (\ref%
{anhrel}) and (\ref{ncurv})) and constructed only from the coefficients of
metric $\mathbf{g}.$ Being defined by certain off--diagonal metric
coefficients, such a torsion is completely deferent from that in string, or
Einstein--Cartan, theory when the torsion tensor is an additional (to
metric) field defined by an antisymmetric $H$--field, or spinning matter,
discussing in Ref. \cite{vncg,vrfg,vsgg}.

\end{document}